\documentclass[11pt,a4paper]{article}

\usepackage{fontspec}
\usepackage[titletoc,title,header]{appendix}

\usepackage{amsmath}
\usepackage{amssymb}
\usepackage{amsthm}
\usepackage{mathtools}
\usepackage{xparse}    

\usepackage[margin=1in]{geometry}

\usepackage{graphicx}
\usepackage{xcolor}

\usepackage{hyperref}
\hypersetup{
  colorlinks=true,
  linkcolor=blue,
  citecolor=blue,
  urlcolor=blue,
  pdfauthor={Yoshitsugu Sekine},
  pdftitle={Interacting Lattice Bosons on the Concrete Buchholz Algebra}
}

\usepackage[numbers]{natbib}
\theoremstyle{plain}
\newtheorem{thm}{Theorem}[section]
\AtEndEnvironment{thm}{\qed}
\newtheorem{lem}[thm]{Lemma}
\AtEndEnvironment{lem}{\qed}
\newtheorem{prop}[thm]{Proposition}
\AtEndEnvironment{prop}{\qed}
\newtheorem{cor}[thm]{Corollary}
\AtEndEnvironment{cor}{\qed}

\AtEndEnvironment{conj}{\qed}
\newtheorem{fact}[thm]{Fact}
\AtEndEnvironment{fact}{\qed}

\theoremstyle{definition}
\newtheorem{defn}[thm]{Definition}
\AtEndEnvironment{defn}{\qed}

\AtEndEnvironment{ex}{\qed}

\theoremstyle{remark}

\AtEndEnvironment{rem}{\qed}

\makeatletter
\providecommand*{\dashv}{\mathrel{\mathpalette\@Dashv\vDash}}
\newcommand*{\@dashv}[2]{\reflectbox{$\m@th#1#2$}}
\makeatother
\newcommand{\abscard}[1]{\abs{#1}} 
\newcommand{\abs}[1]{\left| #1 \right|} 
\NewDocumentCommand{\weaknorm}{O{\dbk} m}{#1{#2}} 
\newcommand{\norm}[1]{\left\Vert #1 \right\Vert} 

\newcommand{\bkt}[2]{\left\langle #1,\,#2 \right\rangle} 
\newcommand{\rbkt}[2]{\left( #1,\,#2 \right)} 
\newcommand{\isomto}{\mathrel{\rightarrowtail\kern-1.9ex\twoheadrightarrow}} 
\newcommand{\cbk}[1]{\left\{ #1 \right\}} 
\newcommand{\dbk}[1]{\left\langle #1 \right\rangle} 
\newcommand{\pairbk}[1]{\rbk{#1}} 
\newcommand{\rbk}[1]{\left( #1 \right)} 
\newcommand{\sqbk}[1]{\left[ #1 \right]} 
\newcommand{\funcond}[3]{\fun{#1}{#2 \middle| #3}} 
\newcommand{\fun}[2]{#1 \rbk{#2}} 
\newcommand{\sqfuncond}[3]{\sqfun{#1}{#2 \middle| #3}} 
\newcommand{\sqfun}[2]{#1 \sqbk{#2}} 
\newcommand{\closedinterval}[2]{\sqbk{#1,\,#2}} 
\newcommand{\openinterval}[2]{\rbk{#1,\,#2}} 
\newcommand{\rightopeninterval}[2]{\left[#1, \, #2 \right)} 
\newcommand{\commutator}[2]{\sqbk{#1,\,#2}} 
\newcommand{\bra}[1]{\left\langle #1 \right|} 
\newcommand{\ketbra}[2]{\ket{#1} \!\! \bra{#2}} 
\newcommand{\ket}[1]{\left| #1 \right\rangle} 

\NewDocumentCommand{\imunit}{O{\mathsf{i}}}{#1} 
\NewDocumentCommand{\placeholder}{O{\bullet}}{#1} 
\NewDocumentCommand{\trace}{O{\operatorname{Tr}}}{#1} 
\newcommand{\Ran}{\operatorname{Ran}} 
\newcommand{\eqcsq}[1]{\sqbk{#1}} 
\newcommand{\kroneckerdelta}{\delta} 
\newcommand{\napiernum}{\mathsf{e}} 
\newcommand{\net}[2]{\rbk{#1}_{#2}} 
\newcommand{\od}[2]{\frac{d #1}{d #2}} 
\newcommand{\opchern}[1]{\operatorname{ch}} 
\newcommand{\opimag}{\operatorname{Im}} 
\newcommand{\opod}[1]{\frac{d}{d #1}} 
\newcommand{\setSymbolDownLeft}[2]{{\vphantom{#2}}_{#1}{#2}} 
\newcommand{\setSymbolUpLeft}[2]{{\vphantom{#2}}^{#1}{#2}} 
\NewDocumentCommand{\agvariety}{O{\mathcal}}{#1} 
\NewDocumentCommand{\cmdrel}{O{\omega}}{#1} 
\NewDocumentCommand{\dfsp}{O{A}}{#1} 
\NewDocumentCommand{\eqcpointed}{O{\eqcsq} m}{#1{#2}_{\ast}} 
\NewDocumentCommand{\fnheaviside}{O{H}}{#1} 
\NewDocumentCommand{\fthol}{O{\mathcal{O}}}{#1} 
\NewDocumentCommand{\ftmero}{O{\mathcal{M}}}{#1} 
\NewDocumentCommand{\grcentralizer}{O{Z}}{#1} 
\NewDocumentCommand{\grmetform}{O{2} m m}{\grmet[#1] \! \rbkt{#2}{#3}} 
\NewDocumentCommand{\grmet}{O{2}}{\setSymbolDownLeft{#1}{g}} 
\NewDocumentCommand{\grnormalizer}{O{N}}{#1} 
\NewDocumentCommand{\gropasym}{O{A}}{#1} 
\NewDocumentCommand{\gropsym}{O{S}}{#1} 
\NewDocumentCommand{\grpermorderedpair}{O{\mathcal{P}}}{#1} 
\NewDocumentCommand{\grsym}{O{\mathfrak{S}} m}{#1_{#2}} 
\NewDocumentCommand{\gtbase}{O{\mathcal}}{#1} 
\NewDocumentCommand{\gtfilter}{O{\mathcal}}{#1} 
\NewDocumentCommand{\gtfmlclosed}{O{\mathcal}}{#1} 
\NewDocumentCommand{\gtfmlopen}{O{\mathcal}}{#1} 
\NewDocumentCommand{\gtopenball}{O{U}}{#1} 
\NewDocumentCommand{\gtopencover}{O{\mathcal}}{#1} 
\NewDocumentCommand{\gtopennbh}{O{\mathcal}}{#1} 
\NewDocumentCommand{\gtpreopencover}{O{\mathcal}}{#1} 
\NewDocumentCommand{\gtsubbase}{O{\mathcal}}{#1} 
\NewDocumentCommand{\gtvicinity}{O{\mathcal}}{#1} 
\NewDocumentCommand{\lasp}{O{\mathcal}}{#1} 
\NewDocumentCommand{\latprightrbk}{O{\top} m}{\rbk{#2}^{#1}} 
\NewDocumentCommand{\latpright}{O{\top} m}{#2^{#1}} 
\NewDocumentCommand{\latp}{O{t} m}{\setSymbolUpLeft{#1}{#2}} 
\NewDocumentCommand{\lpdistribution}{O{\mu} m}{#2_{\ast,#1}} 
\NewDocumentCommand{\lpmollifier}{O{\rho}}{#1} 
\NewDocumentCommand{\lpofpositive}{O{\chi}}{#1} 
\NewDocumentCommand{\manliederiv}{O{L}}{#1} 
\NewDocumentCommand{\mansmoothnbh}{O{\mathcal}}{#1} 
\NewDocumentCommand{\mblfmldsysgenerated}{O{d} m}{\fun{#1}{#2}} 
\NewDocumentCommand{\mblfmlgenerated}{O{\sigma} m}{\fun{#1}{#2}} 
\NewDocumentCommand{\oacorrfn}{O{\Gamma}}{#1} 
\NewDocumentCommand{\oagnsvector}{O{\Omega}}{#1} 
\NewDocumentCommand{\oaideal}{O{\mathcal}}{#1} 
\NewDocumentCommand{\oanumberoperator}{O{A}}{#1} 
\NewDocumentCommand{\oaposcone}{O{\mathcal{P}}}{#1} 
\NewDocumentCommand{\oapressure}{O{P}}{#1} 
\NewDocumentCommand{\oarepn}{O{\pi}}{#1} 
\NewDocumentCommand{\oaspnormalstate}{O{N}}{#1} 
\NewDocumentCommand{\oasppurestate}{O{P}}{#1} 
\NewDocumentCommand{\oaspstate}{O{E}}{#1} 
\NewDocumentCommand{\oastatevector}{O{\Omega}}{#1} 
\NewDocumentCommand{\oastate}{O{\omega}}{#1} 
\NewDocumentCommand{\opdilation}{O{\delta}}{#1} 
\NewDocumentCommand{\opdmat}{O{\rho}}{#1} 
\NewDocumentCommand{\opfockan}{O{a}}{#1} 
\NewDocumentCommand{\opfockcran}{O{a}}{#1^{\#}} 
\NewDocumentCommand{\opfockcrdagger}{O{a}}{#1^{\dagger}} 
\NewDocumentCommand{\opfockcr}{O{a}}{#1^{\ast}} 
\NewDocumentCommand{\opfocknumber}{O{N}}{#1} 
\NewDocumentCommand{\opfocksegalconj}{O{\pi}}{#1} 
\NewDocumentCommand{\opfocksegal}{O{\phi}}{#1} 
\NewDocumentCommand{\opspecmeas}{O{E}}{#1} 
\NewDocumentCommand{\opspec}{O{} m}{\fun{\sigma_{#1}}{#2}} 
\NewDocumentCommand{\optransl}{O{\tau}}{#1} 
\NewDocumentCommand{\opvarspec}{O{} m}{\operatorname{spec}_{#1} #2} 
\NewDocumentCommand{\physaction}{O{\mathcal{A}}}{#1} 
\NewDocumentCommand{\physcharge}{O{e}}{\mathrm{#1}} 
\NewDocumentCommand{\physcplconst}{O{\mathsf{g}}}{#1} 
\NewDocumentCommand{\physelectrostaticcapasity}{O{\mathrm{Cap}}}{#1} 
\NewDocumentCommand{\physenergy}{O{E}}{#1} 
\NewDocumentCommand{\physgse}{O{E}}{#1_{0}} 
\NewDocumentCommand{\physham}{O{H}}{#1} 
\NewDocumentCommand{\physlagdensity}{O{\mathcal{L}}}{#1} 
\NewDocumentCommand{\physlag}{O{L}}{#1} 
\NewDocumentCommand{\physliouvilean}{O{L}}{#1} 
\NewDocumentCommand{\physmass}{O{m}}{#1} 
\NewDocumentCommand{\prbbrownmv}{O{B}}{#1} 
\NewDocumentCommand{\prbcharfun}{O{\chi}}{#1} 
\NewDocumentCommand{\prbdist}{O{\mathcal{P}}}{#1} 
\NewDocumentCommand{\prbgaussianmeasure}{O{\msrcal{N}}}{#1} 
\NewDocumentCommand{\prbnormaldist}{O{N}}{#1} 
\NewDocumentCommand{\prbpoissonprocess}{O{N}}{#1} 
\NewDocumentCommand{\prbprocess}{O{X}}{#1} 
\NewDocumentCommand{\prbqspace}{O{\mathcal{Q}}}{#1} 
\NewDocumentCommand{\prbspsample}{O{\Omega}}{#1} 
\NewDocumentCommand{\psh}{O{\mathfrak}}{#1} 
\NewDocumentCommand{\qtquantumchannel}{O{\mathcal{L}}}{#1} 
\NewDocumentCommand{\repn}{O{\pi}}{#1} 
\NewDocumentCommand{\schattencls}{O{\mathbb{K}}}{#1} 
\NewDocumentCommand{\setfmlcylinder}{O{\mathcal{C}}}{#1} 
\NewDocumentCommand{\setfml}{O{\mathcal}}{#1} 
\NewDocumentCommand{\setindex}{O{\mathcal} m}{#1{#2}} 
\NewDocumentCommand{\setlattice}{O{\Gamma}}{#1} 
\NewDocumentCommand{\setspecial}{O{\mathcal} m}{#1{#2}} 
\NewDocumentCommand{\shdiffform}{O{\sheaf{A}}}{#1} 
\NewDocumentCommand{\sheaf}{O{\mathfrak}}{#1} 
\NewDocumentCommand{\smchemicalpotential}{O{\mu}}{#1} 
\NewDocumentCommand{\smenergydensity}{O{\varrho}}{#1} 
\NewDocumentCommand{\smfluctuationwithdmat}{O{\beta} m}{\smuncertaintywithdmat[#1]{#2}^2} 
\NewDocumentCommand{\sminvtemperature}{O{\beta}}{#1} 
\NewDocumentCommand{\smlocaldensityoperator}{O{\rho}}{#1} 
\NewDocumentCommand{\smmicrocanonicalstate}{O{\beta} m}{\physmean{#2}_{#1}} 
\NewDocumentCommand{\smnumberdensity}{O{\rho}}{#1} 
\NewDocumentCommand{\smooth}{O{\mathcal{E}}}{#1} 
\NewDocumentCommand{\smparticlenumber}{O{N}}{#1} 
\NewDocumentCommand{\smpressure}{O{p}}{#1} 
\NewDocumentCommand{\smspecificfreeenergy}{O{\bar{f}}}{#1} 
\NewDocumentCommand{\smthermalvac}{O{\beta}}{\Omega_{#1}} 
\NewDocumentCommand{\smuncertaintywithdmat}{O{\beta} m}{\rbk{\triangle #2}_{#1}} 
\NewDocumentCommand{\sphilbfrak}{O{\mathfrak}}{#1} 
\NewDocumentCommand{\sphilb}{O{\mathcal}}{#1} 
\NewDocumentCommand{\splowerhalf}{O{\mathbb{H}}}{#1_{\txtneg}} 
\NewDocumentCommand{\spupperhalf}{O{\mathbb{H}}}{#1_{\txtnonneg}} 
\NewDocumentCommand{\topmetric}{O{d}}{#1} 
\NewDocumentCommand{\vaoutnormal}{O{\widehat}}{#1} 
\newcommand{\category}[1]{\mathop{\mathsf{#1}}} 

\newcommand{\catpresheaf}[1]{\category{PSh}} 
\newcommand{\faadj}[1]{#1^{\ast}} 
\newcommand{\fldreal}{\fld{R}} 
\newcommand{\fld}[1]{\mathbb{#1}} 
\newcommand{\fndef}[1]{\boldsymbol{1}_{#1}} 
\newcommand{\fnexp}[1]{\fun{\exp}{#1}} 
\newcommand{\fnrestr}[2]{\left. #1 \right|_{#2}} 
\newcommand{\lpseq}{\ell} 
\newcommand{\mblfml}[1]{\mathcal{#1}} 
\newcommand{\monnat}{\mathbb{N}} 
\newcommand{\msrcal}[1]{\mathcal{#1}} 
\newcommand{\msrprb}{\mathrm{Pr}} 

\newcommand{\oacstar}{C^{\ast}} 
\newcommand{\oaresolventalgebra}{\oa{R}} 
\newcommand{\oa}[1]{\mathcal{#1}} 
\newcommand{\opdmsr}[1]{\mathop{d #1}} 
\newcommand{\opspbddlin}[1]{\fun{\mathbb{B}}{#1}} 
\newcommand{\opspecint}[1]{\mathcal{E}} 
\newcommand{\physmean}[1]{\dbk{#1}} 
\newcommand{\prbexp}{\mathbb{E}} 
\newcommand{\ringratint}{\mathbb{Z}} 
\newcommand{\semigrposint}{\monnat_1} 

\newcommand{\seq}[2]{\if\relax\detokenize{#1}\relax \rbk{#1} \else \rbk{#1}_{#2} \fi} 
\newcommand{\setisomorphism}[1]{\operatorname{Iso}} 
\newcommand{\setone}[1]{\cbk{#1}}
\newcommand{\set}[2]{\left\{#1 \, \middle| \, #2\right\}}
\newcommand{\topdist}{\operatorname{dist}} 
\newcommand{\txtfin}{\mathrm{fin}} 
\newcommand{\txtfr}{\mathrm{fr}} 
\newcommand{\txtneg}{\mathrm{-}} 
\newcommand{\txtnonneg}{\mathrm{+}} 
\newcommand{\txtsym}{\mathrm{s}} 

\title{Interacting Lattice Bosons on the Concrete Buchholz Algebra\\\vspace{0.5em}{\large All-Temperature Equilibrium States from Functional-Integral Local-Number Bounds}}

\author{%
Yoshitsugu Sekine\\{\small\texttt{4429sekine@gmail.com}}%
}

\date{\today}

\begin{document}

\maketitle

\begin{abstract}
Interacting Bose--Hubbard dynamics is formulated on the concrete Buchholz algebra in Fock space. Under repulsive finite-range interactions and a low-activity gap, a discrete functional integral gives volume-uniform exponential local-number bounds at all temperatures. These bounds remove the local-moment hypothesis and yield finite-volume approximation and KMS accumulation states in \cite{DeuchertLampartLemm001}.

\end{abstract}

\setcounter{tocdepth}{3}
\tableofcontents

\section{Introduction}\label{introduction}

The scarcity of physically relevant automorphism groups on the Weyl algebra has long limited its use for interacting bosonic systems. The natural automorphisms of the Weyl algebra are induced by symplectic transformations and describe quadratic Hamiltonians, whereas even simple nonquadratic Hamiltonians may fail to preserve the algebra. The resolvent algebra was introduced to overcome this restriction and to incorporate resolvents of Hamiltonians and other unbounded observables into a representation-independent \(\oacstar\)-algebraic framework \cite[Sections 1 and 5]{BuchholzGrundling002}. Its rich ideal structure is not an incidental complication but a source of physical information \cite{DetlevBuchholz001}. This structure has been used in general criteria for Bose--Einstein condensation and in the analysis of trapped and interacting Bose systems \cite{BahnsBuchholz001,DetlevBuchholz002,DetlevBuchholz004,DetlevBuchholz005,DetlevBuchholz012}. These applications also show why the resolvent algebra must be evaluated through concrete models, of which only a limited range has been treated to date. One principal aim of this paper is to extend that range to the equilibrium theory of the Bose--Hubbard model.

Like the Weyl algebra, the resolvent algebra is defined before a representation is selected. This representation independence is essential to its role as an algebra of observables. The contrast with a von Neumann algebra makes the point precise. After an automorphism group and a ground, equilibrium, or vacuum state have been chosen, the corresponding von Neumann algebra is obtained as the weak closure of the \(\oacstar\)-algebra in the GNS representation of that state. It retains the folium and the thermodynamic sector of the reference state. A nonfactorial representation may have a nontrivial center, and its central decomposition records classical order parameters or the central-spectrum variables that distinguish sectors. This is the micro--macro interpretation of the center in sector theory \cite{IzumiOjima002}. In symmetry-breaking situations, the center of the mixed representation records the classical distribution of phases, while its factor components describe the individual quantum phases. The abstract resolvent algebra does not build such a classical component into the kinematics. Its center is trivial, even though its proper ideals remain physically informative \cite[Theorem 4.10]{BuchholzGrundling002}. The restriction to one reference folium is also visible in thermal representations. For example, factorial KMS representations of type III at distinct inverse temperatures are disjoint \cite[Theorem 5.3.35]{BratteliRobinson004}. The normal states of a von Neumann algebra selected by one state remain in that folium. Such an algebra cannot serve as a representation-independent container for all thermodynamic sectors or for unrelated dynamics. An abstract \(\oacstar\)-algebra is instead intended to retain the kinematics before any such choice. Developing as much of the dynamics and equilibrium theory as possible at this abstract level is a basic requirement rather than a matter of formal generality.

The first operator-algebraic problem for a concrete Hamiltonian is the construction of its automorphism group. This is the algebraic counterpart of self-adjointness in the operator formulation, but it involves two additional obstructions. The first is continuity. Even free bosonic dynamics need not be point-norm continuous in the natural \(\oacstar\)-norm \cite[Section 7.1]{BuchholzGrundling002}. A continuous automorphism group may emerge only after a suitable representation has been selected and the weak closure has been taken, provided that the dynamics is continuously implementable there. Such a passage to a von Neumann algebra resolves continuity within the chosen folium but reintroduces the representation dependence that the abstract resolvent algebra was designed to avoid.

The second obstruction is invariance of the algebra. Free bosonic dynamics acts by automorphisms of the resolvent algebra, although the action need not be point-norm continuous, whereas interacting dynamics can send elements of the gauge-invariant resolvent algebra into a larger algebra. An interacting evolution that leaves the resolvent algebra requires a larger ambient \(\oacstar\)-algebra on which it acts by automorphisms. For continuous Bose systems, particle-sector consistency in the Fock representation produces such an extension of the gauge-invariant resolvent algebra and accommodates broad classes of interacting dynamics \cite{DetlevBuchholz002,DetlevBuchholz003}. The same structure supports the analysis of proper condensates, for which infinite particle accumulation is detected by particle-number resolvents \cite{DetlevBuchholz004,DetlevBuchholz005}. The corresponding lattice extension is called the Buchholz algebra in \cite[Definition 2.2]{DeuchertLampartLemm001}. Within the Fock representation, it contains the gauge-invariant resolvent algebra and is preserved by the Bose--Hubbard dynamics \cite[Proposition 3.1]{DeuchertLampartLemm001}. In this sense, the Buchholz algebra is an ambient algebra for the gauge-invariant resolvent algebra. The continuous-space construction also contains nonzero invariant subalgebras on which the automorphism group is point-norm continuous, obtained by taking the continuous elements of the action \cite[Theorem 4.6]{DetlevBuchholz002}. Both the continuity and invariance problems have partial solutions once the Fock representation is fixed.

The Bose--Hubbard model is a stringent test of this framework because every lattice site has an infinite-dimensional local Hilbert space. An interaction-picture cluster expansion gives volume-uniform local particle-number estimates for its finite-volume Gibbs states at sufficiently high temperature and for arbitrary fixed chemical potential \cite[Theorem 1]{GongKuwaharaTong001}. Combined with the Deuchert--Lampart--Lemm automorphism and KMS construction, these estimates yield equilibrium states in a high-temperature region without a separate low-activity restriction on the chemical potential. The infinite-volume Bose--Hubbard dynamics on the concrete Buchholz algebra, its finite-volume approximation, and its KMS accumulation states are established in \cite[Proposition 3.1 and Theorems 4.1 and 5.2]{DeuchertLampartLemm001}. The finite-volume approximation and equilibrium-state conclusions, however, assume uniform polynomial moments of every required order for the local particle number. These are the hypotheses in \cite[equations (37) and (90)]{DeuchertLampartLemm001}. This state condition is technically severe. It can fail in condensed regimes with divergent local occupation, including proper-condensate limits. It also leaves the existence theorem conditional on an estimate that is not derived from the Hamiltonian. Superstability and a functional-integral representation are proposed as a route to the missing estimate in \cite[Remark 5.3]{DeuchertLampartLemm001}.

The present paper implements that proposal. For repulsive finite-range density interactions, we assume that the chemical potential makes the grand-canonical one-particle Hamiltonian bounded below by a strictly positive constant. A discrete functional integral yields a volume-uniform exponential bound on the local particle number at every positive inverse temperature. Theorem \ref{thm:fock-all-temperature-local-number} states the precise result. The exponential estimate is stronger than the polynomial condition used in the earlier approximation argument. It permits the occupation cutoff to be removed directly and gives the all-temperature finite-volume approximation in Theorem \ref{thm:fock-all-temperature-local-approximation} without imposing a local-number hypothesis on the state. The same argument produces weak-\(\ast\) accumulation states that are invariant on the gauge-invariant resolvent subalgebra and satisfy the KMS boundary relation there for the infinite-volume automorphism group, as stated in Corollary \ref{cor:fock-all-temperature-kms-states}. The high-temperature result and the present theorem cover complementary regimes. The high-temperature result allows arbitrary fixed chemical potential at sufficiently high temperature, whereas the present theorem allows every positive temperature under a low-activity gap condition. The gap assumption does not assert the absence of Bose--Einstein condensation in the Bose--Hubbard model. It excludes the threshold regime in which the chemical potential closes the one-particle gap, and that regime lies outside the present theorem.

The present construction nevertheless remains tied to the bosonic Fock representation. From the viewpoint of algebraic quantum statistical mechanics, the Buchholz algebra should ultimately be defined without a distinguished reference representation and should serve as an abstract ambient algebra for the gauge-invariant resolvent algebra. The Bose--Hubbard dynamics should then be defined as an automorphism group of this abstract algebra, while the ordinary \(\oacstar\)-KMS condition should be imposed on an appropriate invariant subalgebra on which the action is point-norm continuous. This would preserve the ability of the abstract algebra to accommodate distinct folia, temperatures, and phases before a GNS representation is selected.

The functional integral used here should also admit an intrinsic formulation at that level. The correspondence between stochastically positive KMS systems, periodic stochastic processes, and positive semigroups was developed in a formally \(\oacstar\)-algebraic setting by Klein--Landau \cite[Sections 1 and 6--8]{KleinLandau001}. Its reconstruction nevertheless passes through a GNS Hilbert space and the associated von Neumann algebra. Standard-form perturbation theory likewise constructs perturbed Liouvilleans and KMS states after a representation has been selected \cite{DerezinskiJaksicPillet001}. For bosonic fields, the Araki--Woods representation can be identified with a Gaussian periodic path space, and this equivalence yields interacting dynamics and equilibrium states through functional integration \cite[Sections 17.1.5 and 21.4--21.5]{DerezinskiGerard001}. These results are sharp at the von Neumann-algebraic level, but a deeper relation between probability, stochastic processes, and abstract \(\oacstar\)-algebras is needed for the resolvent and Buchholz algebras. Establishing the present all-temperature result in the fixed Fock representation is a first step toward that representation-independent theory.

The paper first fixes the Fock-space Bose--Hubbard model and the concrete Buchholz algebra. It then isolates the missing local-number input in the existing dynamical and equilibrium arguments, constructs the discrete functional integral, proves the exponential bound, and derives the infinite-volume dynamics and KMS states in the low-activity regime.

\section{Fock-Space Bose--Hubbard Model}\label{fock-space-bosehubbard-model}

The concrete model fixes the spatial and particle-number structures before either dynamics or states are taken to an infinite-volume limit.

\subsection{Graph, Fock spaces, and local observables}\label{graph-fock-spaces-and-local-observables}

Let \(\Gamma\) be a connected countable graph of uniformly bounded degree. The graph distance is denoted by \(\fun{\topdist}{x,y}\). For a finite \(X \Subset \Gamma\) and \(R \geq 0\), the closed graph neighborhood is \[\begin{aligned}
X[R]
& =
\set{y \in \Gamma}{\fun{\topdist}{y,X} \leq R}.
\end{aligned}\] The interior boundary is \[\begin{aligned}
\partial X
& =
\set{x \in X}{\fun{\topdist}{x,\Gamma \setminus X} = 1}.
\end{aligned}\] Assume that \(\Gamma\) is \(d\)-dimensional with surface parameter \(\sigma > 0\) in the sense that \[\begin{aligned}
\sup_{x \in \Gamma}
\sup_{R \geq 1}
\frac{\abscard{\partial\rbk{x[R]}}}{R^{d - 1}}
& \leq
\sigma.
\end{aligned}\] This is Definition 2.1 of \cite{DeuchertLampartLemm001} and implies the polynomial ball-volume bound \begin{equation}
\begin{aligned}
\sup_{x\in\Gamma}\abscard{x[R]}
& \leq
\sigma R^d,
\quad
R\geq1.
\end{aligned}
\label{eq:fock-polynomial-ball-volume-growth}
\end{equation} The phrase polynomial ball-volume growth refers to the estimate \eqref{eq:fock-polynomial-ball-volume-growth}.

For \(X \Subset \Gamma\) or \(X = \Gamma\), the one-particle and bosonic Fock spaces are \[\begin{aligned}
\sphilbfrak{h}_{X}
& =
\fun{\lpseq^{2}}{X},
&
\sphilb{F}_{X}
& =
\bigoplus_{n \in \monnat}\sphilbfrak{h}_{X}^{\otimes_{\txtsym}^{n}}.
\end{aligned}\] For \(n \in \monnat\), the orthogonal projection \(P_{X,n}\) onto the \(n\)-particle summand is defined by \begin{equation}
\begin{aligned}
\Ran P_{X,n}
& =
\sphilbfrak{h}_{X}^{\otimes_{\txtsym}^{n}}.
\end{aligned}
\label{eq:fock-particle-sector-projection}
\end{equation} The total number operator is \(\opfocknumber_X = \sum_{n \in \monnat}n P_{X,n}\). For \(m\in\monnat\), define \begin{equation}
\begin{aligned}
P_{X,\leq m}
& =
\sum_{n=0}^{m}P_{X,n},
&
\sphilb{D}_{X,\txtfin}
& =
\bigcup_{m\in\monnat}P_{X,\leq m}\sphilb{F}_X.
\end{aligned}
\label{eq:fock-finite-particle-subspace}
\end{equation} Thus \(\sphilb{D}_{X,\txtfin}\) consists precisely of the vectors having nonzero components in only finitely many particle-number summands. For \(x \in X\), the creation, annihilation, and local occupation operators on \(\sphilb{D}_{X,\txtfin}\) are \(\faadj{a_x}\), \(a_x\), and \(N_x = \faadj{a_x}a_x\). For \(x,y\in X\), their canonical commutation relations and the resulting local-number commutators on \(\sphilb{D}_{X,\txtfin}\) are \begin{equation}
\begin{aligned}
\commutator{a_x}{\faadj{a_y}}
& =
\kroneckerdelta_{x,y},
&
\commutator{a_x}{a_y}
& =
0,
\\
\commutator{\faadj{a_x}}{a_y}
& =
-\kroneckerdelta_{x,y},
&
\commutator{\faadj{a_x}}{\faadj{a_y}}
& =
0,
\\
\commutator{N_x}{\faadj{a_y}}
& =
\faadj{a_x}\commutator{a_x}{\faadj{a_y}}
+
\commutator{\faadj{a_x}}{\faadj{a_y}}a_x
=
\kroneckerdelta_{x,y}\faadj{a_x}
=
\kroneckerdelta_{x,y}\faadj{a_y},
\\
\commutator{N_x}{a_y}
& =
\faadj{a_x}\commutator{a_x}{a_y}
+
\commutator{\faadj{a_x}}{a_y}a_x
=
-\kroneckerdelta_{x,y}a_x
=
-\kroneckerdelta_{x,y}a_y.
\end{aligned}
\label{eq:fock-creation-annihilation-number-commutators}
\end{equation} For finite \(X \Subset \Gamma\), let \(\Omega_X\) be the Fock vacuum. For each multi-index \(\boldsymbol{\nu} = \seq{\nu_x}{x \in X} \in \monnat^X\), define \begin{equation}
\begin{aligned}
\Psi_{X,\boldsymbol{\nu}}
& =
\prod_{x \in X}
\frac{\rbk{\faadj{a_x}}^{\nu_x}}{\sqrt{\nu_x!}}
\Omega_X,
\\
N_x\Psi_{X,\boldsymbol{\nu}}
& =
\nu_x\Psi_{X,\boldsymbol{\nu}},
&
x
& \in
X,
\\
\opfocknumber_X\Psi_{X,\boldsymbol{\nu}}
& =
\rbk{\sum_{x \in X}\nu_x}\Psi_{X,\boldsymbol{\nu}}.
\end{aligned}
\label{eq:fock-occupation-number-basis}
\end{equation} The creation operators at distinct sites commute, and hence the product does not depend on the ordering of \(X\). The vectors \(\Psi_{X,\boldsymbol{\nu}}\), \(\boldsymbol{\nu} \in \monnat^X\), form the orthonormal occupation-number basis of \(\sphilb{F}_X\).

Let \(\physham[h]_{X}\) be the bounded hopping operator on \(\sphilbfrak{h}_{X}\) with matrix elements \[\begin{aligned}
\rbk{\physham[h]_{X}}_{x,y}
& =
\begin{cases}
-1,
& x \sim y,
\\
0,
& x \not\sim y.
\end{cases}
\end{aligned}\] Assume that the uniform degree of the graph is strictly positive, and set \begin{equation}
\begin{aligned}
0<C_h
& =
\sup_{x \in \Gamma}\abscard{\set{y \in \Gamma}{y \sim x}}.
\end{aligned}
\label{eq:fock-uniform-degree-bound}
\end{equation} Then \(\norm{\physham[h]_{\Lambda}} \leq C_h\) for every finite \(\Lambda\), and \(-\physham[h]_{\Lambda} - C_h\) has nonnegative off-diagonal entries and nonpositive row sums. Fix \(U > 0\) and \(r\in\semigrposint\), and let \(w \colon \Gamma \times \Gamma \to \fldreal\) be a bounded symmetric function satisfying \begin{equation}
\begin{aligned}
\fun{w}{x,y}
& \geq
0,
&
\fun{\topdist}{x,y}>r
& \implies
\fun{w}{x,y}=0.
\end{aligned}
\label{eq:fock-interaction-range}
\end{equation} The density interaction is defined by \[\begin{aligned}
\fun{v}{x,y}
& =
U\fun{\fndef{\setone{x}}}{y} + \fun{w}{x,y}.
\end{aligned}\] Thus \(v\) is pointwise nonnegative, and its interaction range is at most the integer \(r\) fixed in \eqref{eq:fock-interaction-range}. For finite \(\Lambda \Subset \Gamma\), the number-conserving Hamiltonian and its grand-canonical shift are defined on \(\sphilb{D}_{\Lambda,\txtfin}\) from \eqref{eq:fock-finite-particle-subspace} by \begin{equation}
\begin{aligned}
\physham_{\txtfr,\Lambda}
& =
\sum_{x,y \in \Lambda}
\rbk{\physham[h]_{\Lambda}}_{x,y}\faadj{a_x} a_y,
\\
V_{\Lambda}
& =
\sum_{x,y \in \Lambda}
\fun{v}{x,y}\faadj{a_x}\faadj{a_y} a_x a_y,
\\
H_{\Lambda}
& =
\physham_{\txtfr,\Lambda} + V_{\Lambda},
\\
K_{\Lambda,\smchemicalpotential_{\Lambda}}
& =
H_{\Lambda} - \smchemicalpotential_{\Lambda}\opfocknumber_{\Lambda}.
\end{aligned}
\label{eq:fock-bose-hubbard-hamiltonian}
\end{equation} The ordered-pair convention in \eqref{eq:fock-bose-hubbard-hamiltonian} agrees with \cite[equation (10)]{DeuchertLampartLemm001} and introduces no additional factor \(1/2\).

The Hamiltonian commutes with \(\opfocknumber_{\Lambda}\). Its restriction \(H_{\Lambda,n}\) to \(P_{\Lambda,n}\sphilb{F}_{\Lambda}\) is bounded for every \(n \in \monnat\). For \(X = \Gamma\), the same formula on the \(n\)-particle space defines a bounded operator \(H_{\Gamma,n}\) because the hopping has bounded degree and the interaction is bounded and finite-range. The symbol \(H_{\Gamma}\) denotes the self-adjoint direct sum \(\bigoplus_{n\in\monnat}H_{\Gamma,n}\) on \(\sphilb{F}_{\Gamma}\).

\subsection{Concrete Buchholz algebra}\label{concrete-buchholz-algebra}

The gauge-invariant resolvent algebra on \(\sphilbfrak{h}_{X}\) is denoted by \(\fun{\oaresolventalgebra}{X}^{\gamma}\). The concrete Buchholz algebra of \cite[Definition 2.2]{DeuchertLampartLemm001} is \[\begin{aligned}
\oa{B}_{X}
& =
\set{A \in \opspbddlin{\sphilb{F}_{X}}}
{\text{for every $n \in \monnat$ there is $R \in
\fun{\oaresolventalgebra}{X}^{\gamma}$ with $A P_{X,\leq n} = R P_{X,\leq n}$}}.
\end{aligned}\] For \(X = \Gamma\), write \(\oa{B} = \oa{B}_{\Gamma}\). Proposition 2.4 of the cited paper identifies \(\oa{B}\) with the bounded inverse limit of its particle-sector coordinate algebras. This characterization explains why the algebraic dynamics can exist without a uniform continuity estimate in the particle-number index \(n\).

\section{Existing Dynamics on the Concrete Buchholz Algebra and the Missing Estimate}\label{existing-dynamics-on-the-concrete-buchholz-algebra-and-the-missing-estimate}

The full automorphism group on the concrete Buchholz algebra exists before any state or moment condition is selected. The moment condition enters only when this group is approximated by spatially finite Hamiltonians and implemented in an infinite-particle representation.

\subsection{Automorphisms on the full Buchholz algebra}\label{automorphisms-on-the-full-buchholz-algebra}

For \(X \Subset \Gamma\) or \(X = \Gamma\), let \(\alpha_X\) denote the one-parameter family whose time-\(t\) map is defined by \[\begin{aligned}
\fun{\alpha_{X,t}}{A}
& =
\napiernum^{\imunit tH_X}A\napiernum^{-\imunit tH_X},
\quad
A \in \oa{B}_X,
\quad
t \in \fldreal.
\end{aligned}\]

\begin{fact}[Deuchert--Lampart--Lemm, Proposition 3.1]\label{fact:dll-buchholz-automorphism}
If $\fun{v}{x,y}$ tends to $0$ as $\fun{\topdist}{x,y} \to \infty$, then
for every finite $X\Subset\Gamma$ and for $X=\Gamma$, $\alpha_X$ is a group of $\ast$-automorphisms of
$\oa{B}_X$, whose time-$t$ member is $\alpha_{X,t}$ for every $t \in \fldreal$.
No continuity in $t$ is asserted.
\end{fact}

The particle-sector Dyson expansion and the consistency maps of the inverse limit prove Fact \ref{fact:dll-buchholz-automorphism}; see \cite[Proposition 3.1 and equations (21)--(29)]{DeuchertLampartLemm001}. Example 3.2 of \cite{DeuchertLampartLemm001} proves that the group need not be weakly continuous. The problem addressed here is not existence of \(\alpha_{\Gamma}\) on \(\oa{B}\) but approximation of its time-\(t\) map \(\alpha_{\Gamma,t}\) by \(\alpha_{\Lambda,t}\) on local gauge-invariant observables.

\subsection{The obstruction in the DLL approximation argument}\label{the-obstruction-in-the-dll-approximation-argument}

Theorem 4.1 of \cite{DeuchertLampartLemm001} derives finite-volume approximation only after assuming a volume-uniform local particle-number estimate. Its proof propagates the assumed estimate, uses it to remove a local occupation cutoff, and then applies a state-independent finite-range propagation bound to the cutoff dynamics. The equilibrium argument in Theorem 5.2 of the same paper inherits this dependence.

The present proof does not use either conditional theorem. The functional-integral estimate proved in this paper controls the cutoff tails directly for the Gibbs family. We then combine that control with the finite-range shell expansion and pass the KMS boundary functions to the accumulation state.

\section{Finite-Volume and Infinite-Volume Equilibrium States}\label{finite-volume-and-infinite-volume-equilibrium-states}

The equilibrium analysis begins with the trace-class Gibbs operator on a finite Fock space and passes its KMS boundary relation to a state of the infinite Buchholz algebra.

\subsection{Superstability and finite-volume Gibbs states}\label{superstability-and-finite-volume-gibbs-states}

Fix \(\beta > 0\) and a bounded family \(\seq{\smchemicalpotential_{\Lambda}}{\Lambda \Subset \Gamma}\) of chemical potentials. The finite-volume partition function and Gibbs state are \begin{equation}
\begin{aligned}
Z_{\Lambda,\beta}
& =
\sqfun{\trace_{\sphilb{F}_{\Lambda}}}{
\fnexp{-\beta K_{\Lambda,\smchemicalpotential_{\Lambda}}}},
\\
\fun{\oastate[\psi_{\Lambda}]}{A}
& =
\frac{
\sqfun{\trace_{\sphilb{F}_{\Lambda}}}{
A\fnexp{-\beta K_{\Lambda,\smchemicalpotential_{\Lambda}}}}}
{Z_{\Lambda,\beta}}.
\end{aligned}
\label{eq:fock-finite-volume-gibbs-state}
\end{equation}

\begin{prop}[Finite-volume Gibbs existence]\label{prop:fock-finite-volume-gibbs-existence}
For every finite $\Lambda \Subset \Gamma$ and $\beta > 0$, the operator
$\fnexp{-\beta K_{\Lambda,\smchemicalpotential_{\Lambda}}}$ is trace class and
\eqref{eq:fock-finite-volume-gibbs-state} defines a normal state on $\opspbddlin{\sphilb{F}_{\Lambda}}$.
Its restriction to $\fun{\oaresolventalgebra}{\Lambda}^{\gamma}$ satisfies the finite-volume $\beta$-KMS
boundary relation for the number-conserving dynamics.
\end{prop}

\begin{proof}
Let
$c_U = U + C_h + \sup_{\Lambda \Subset \Gamma}\abs{\smchemicalpotential_{\Lambda}}$, where
$C_h$ is the uniform degree bound in \eqref{eq:fock-uniform-degree-bound}.
Since $w \geq 0$ and the second quantization of $\physham[h]_{\Lambda}$ is bounded below by
$-C_h \opfocknumber_{\Lambda}$, one has the quadratic-form bound
$$\begin{aligned}
K_{\Lambda,\smchemicalpotential_{\Lambda}}
& \geq
U \sum_{x \in \Lambda}N_x^2 - c_U \sum_{x \in \Lambda}N_x
\end{aligned}$$
on $\sphilb{D}_{\Lambda,\txtfin}$ defined in
\eqref{eq:fock-finite-particle-subspace}.
For every occupation configuration
$\seq{\nu_x}{x \in \Lambda} \in \monnat^{\Lambda}$, the scalar estimate is
$$\begin{aligned}
U \sum_{x \in \Lambda}\nu_x^2 - c_U \sum_{x \in \Lambda}\nu_x
& \geq
\frac{U}{2}\sum_{x \in \Lambda}\nu_x^2
-
\frac{c_U^2}{2U}\abscard{\Lambda}.
\end{aligned}$$
The min--max principle and the occupation-number basis
\eqref{eq:fock-occupation-number-basis} give the trace majorant
$$\begin{aligned}
Z_{\Lambda,\beta}
& \leq
\fnexp{\frac{\beta c_U^2}{2U}\abscard{\Lambda}}
\prod_{x \in \Lambda}
\sum_{n \in \monnat}\fnexp{-\frac{\beta U}{2}n^2}
<
\infty.
\end{aligned}$$
The Gibbs density matrix is consequently well-defined.
Trace cyclicity for analytic elements proves the finite-volume KMS boundary relation.
On gauge-invariant observables, adding $-\smchemicalpotential_{\Lambda}\opfocknumber_{\Lambda}$ to the generator
does not change the real-time action because $\opfocknumber_{\Lambda}$ commutes with those observables.
\end{proof}

\subsection{Infinite-volume equilibrium problem}\label{infinite-volume-equilibrium-problem}

The Fock factorization \(\sphilb{F}_{\Gamma} \cong \sphilb{F}_{\Lambda} \otimes \sphilb{F}_{\Lambda^{\mathrm{c}}}\) and the vacuum vector \(\Omega_{\Lambda^{\mathrm{c}}}\) define the state from equation (89) of \cite{DeuchertLampartLemm001} by \begin{equation}
\begin{aligned}
\overline{\rho}_{\Lambda}
& =
\frac{\fnexp{-\beta K_{\Lambda,\smchemicalpotential_{\Lambda}}}}
{Z_{\Lambda,\beta}}
\otimes
\ketbra{\Omega_{\Lambda^{\mathrm{c}}}}{\Omega_{\Lambda^{\mathrm{c}}}},
\\
\fun{\oastate[\overline{\psi}_{\Lambda}]}{A}
& =
\sqfun{\trace_{\sphilb{F}_{\Gamma}}}{A\overline{\rho}_{\Lambda}},
\quad
A \in \oa{B}.
\end{aligned}
\label{eq:fock-vacuum-extended-gibbs-state}
\end{equation} The state space of \(\oa{B}_{\Gamma}\) is weak-\(\ast\) compact, and the net \(\seq{\oastate[\overline{\psi}_{\Lambda}]}{\Lambda \Subset \Gamma}\) has weak-\(\ast\) accumulation points. Weak-\(\ast\) compactness alone does not pass the KMS boundary condition because the dynamics is controlled through unbounded local occupations before the cutoffs are removed.

No continuity of \(\alpha_{\Gamma}\) on \(\oa{B}\) is assumed. The equilibrium conclusion to be proved is a state \(\oastate[\psi]\) on \(\oa{B}\) that is invariant under \(\alpha_{\Gamma}\) on \(\fun{\oaresolventalgebra}{\Gamma}^{\gamma}\) and satisfies the corresponding KMS boundary relation there. For \(A,B \in \fun{\oaresolventalgebra}{\Gamma}^{\gamma}\), this means the existence of a bounded scalar function \(F_{A,B}\) that is holomorphic on \(-\beta < \opimag z < 0\), continuous on the closed strip, and satisfies \begin{equation}
\begin{aligned}
\fun{F_{A,B}}{t}
& =
\fun{\oastate[\psi]}{\fun{\alpha_{\Gamma,t}}{A}B},
&
\fun{F_{A,B}}{t-\imunit\beta}
& =
\fun{\oastate[\psi]}{B\fun{\alpha_{\Gamma,t}}{A}}.
\end{aligned}
\label{eq:fock-infinite-volume-kms-boundary}
\end{equation} Continuity in this statement concerns only the scalar function \(F_{A,B}\) and is part of the conclusion. It does not assert point-norm, weak, or strong continuity of the automorphism group on \(\oa{B}\).

\section{All-Temperature Particle-Number Bounds at Low Activity}\label{all-temperature-particle-number-bounds-at-low-activity}

The discrete functional integral compares marked open paths with an interacting gas of unmarked loops. The argument uses nonnegativity of \(v\) and a uniform gap below the one-particle band.

\subsection{Sub-Markov kernels and the chemical-potential gap}\label{sub-markov-kernels-and-the-chemical-potential-gap}

Let \(C_h\) be the uniform degree bound defined by \eqref{eq:fock-uniform-degree-bound}. Assume that there is \(\varepsilon > 0\) such that \begin{equation}
\begin{aligned}
\smchemicalpotential_{\Lambda}
& \leq
-C_h - \varepsilon
\end{aligned}
\label{eq:fock-path-chemical-potential-gap}
\end{equation} for every finite \(\Lambda \Subset \Gamma\). The elementary probabilistic definitions are included because the bridge measure and its normalization enter the particle-number estimate directly.

\begin{defn}[Finite state space and sub-Markovian kernels]\label{def:fock-finite-state-space-submarkov}
A finite state space is a finite set $E$ whose elements are called states.
A matrix kernel $R$ on the finite state space $E$ is substochastic when
$$\begin{aligned}
\fun{R}{x,y}
& \geq
0,
&
\sum_{y\in E}\fun{R}{x,y}
& \leq
1
\end{aligned}$$
for every $x,y\in E$.
It is stochastic when every row sum is equal to $1$.
A stochastic kernel is a transition kernel that conserves mass.

A family of matrix kernels $P_t$, $t\geq0$, on $E$ is sub-Markovian when
$$\begin{aligned}
\fun{P_0}{x,y}
& =
\fun{\fndef{\setone{x}}}{y},
\\
\sum_{z\in E}
\fun{P_s}{x,z}\fun{P_t}{z,y}
& =
\fun{P_{s+t}}{x,y},
\\
\fun{P_t}{x,y}
& \geq
0,
\\
\sum_{y\in E}\fun{P_t}{x,y}
&\leq
1
\end{aligned}$$
for every $s,t\geq0$ and $x,y\in E$.
\end{defn}

\begin{defn}[Cemetery state and absorbing state]\label{def:fock-cemetery-absorbing-state}
Let $E$ be a finite state space in the sense of
Definition \ref{def:fock-finite-state-space-submarkov}, and let $\partial\notin E$.
For a substochastic kernel $R$ on $E$, define a stochastic kernel $\widehat{R}$ on
$E\cup\setone{\partial}$ by
$$\begin{aligned}
\fun{\widehat{R}}{x,y}
& =
\fun{R}{x,y},
&
x,y
& \in E,
\\
\fun{\widehat{R}}{x,\partial}
& =
1-\sum_{y\in E}\fun{R}{x,y},
&
x
& \in E,
\\
\fun{\widehat{R}}{\partial,\partial}
& =
1,
\\ 
\fun{\widehat{R}}{\partial,y}
& =
0,
\quad y\in E.
\end{aligned}$$
The extra point $\partial$ is called the cemetery state: the missing row mass of $R$ is the probability of a
transition from $x\in E$ to $\partial$.
The two identities in the last two lines mean that $\partial$ is an absorbing state for $\widehat{R}$; after a
transition to $\partial$, the kernel assigns probability $1$ to remaining at $\partial$ and probability $0$ to
returning to $E$.

For a process $Z=\seq{Z_t}{t\geq0}$ with values in $E\cup\setone{\partial}$, the state $\partial$ is absorbing
when, for every sample point and all $0\leq s\leq t$,
$$\begin{aligned}
Z_s=\partial
& \implies
Z_t=\partial.
\end{aligned}$$
\end{defn}

\begin{defn}[Discrete-time Markov chain]
Let $E$ be a finite state space in the sense of Definition \ref{def:fock-finite-state-space-submarkov}, let $R$
be a stochastic transition kernel on $E$, and let $\msrprb$ be a probability measure on a measurable space.
An $E$-valued sequence $\seq{Y_n}{n\in\monnat}$ is a discrete-time Markov chain with transition kernel $R$ when
$$\begin{aligned}
\funcond{\msrprb}{Y_{n+1}=y}{Y_0,\ldots,Y_n}
& =
\fun{R}{Y_n,y}
\end{aligned}$$
almost surely for every $n\in\monnat$ and $y\in E$.
The initial law is specified separately by the distribution of $Y_0$.
\end{defn}

\begin{defn}[Time-homogeneous Markov process]\label{def:fock-time-homogeneous-markov-process}
Let $E$ be a finite state space, let
$Z=\seq{Z_t}{t\geq0}$ be an $E$-valued process on a probability space
$\rbk{\Sigma,\mblfml{F},\msrprb}$, and let $P_t$, $t\geq0$, be stochastic kernels on $E$.
For $s\geq0$, let $\mblfml{F}_{Z,s}$ be the sigma-algebra generated by $Z_u$, $0\leq u\leq s$.
The process $Z$ is a time-homogeneous Markov process with transition kernels $P_t$ when, for every
$0\leq s\leq t$ and $y\in E$,
$$\begin{aligned}
\sqfuncond{\prbexp_{\msrprb}}{
\fndef{\setone{Z_t=y}}
}{\mblfml{F}_{Z,s}}
& =
\fun{P_{t-s}}{Z_s,y}
\end{aligned}$$
almost surely.
\end{defn}

\begin{defn}[Poisson process]
Let $\msrprb$ be a probability measure on a measurable space.
A process $N_t$, $t\geq0$, is a Poisson process of rate $c>0$ when $N_0=0$ almost surely, its increments over
pairwise disjoint time intervals are independent, and
$$\begin{aligned}
\fun{\msrprb}{N_t-N_s=k}
& =
\napiernum^{-c\rbk{t-s}}
\frac{\rbk{c\rbk{t-s}}^{k}}{k!}
\end{aligned}$$
for every $0\leq s<t$ and $k\in\monnat$.
Its paths are right-continuous, nondecreasing, integer-valued, and have jumps of size $1$.
\end{defn}

\begin{defn}[Exponential law]\label{def:fock-exponential-law}
Let $\msrprb$ be a probability measure on a measurable space.
A nonnegative random variable $T$ has the exponential law of rate $r>0$ when
$$\begin{aligned}
\fun{\msrprb}{T>t}
& =
\napiernum^{-rt}
\end{aligned}$$
for every $t\geq0$.
\end{defn}

\begin{defn}[Càdlàg path]\label{def:fock-cadlag-path}
Let $S$ be a metric space and let $I=\closedinterval{0}{a}$ for some $a>0$, or
$I=\rightopeninterval{0}{\infty}$.
A path $z\colon I\to S$ is càdlàg when it is right-continuous and has a left limit at every positive time in
$I$.
A process $Z=\seq{Z_t}{t\in I}$ is càdlàg when the path
$t\mapsto\fun{Z_t}{\omega}$ is càdlàg for every sample point $\omega$.
\end{defn}

\begin{defn}[Process killed at a nonnegative random variable]\label{def:fock-killed-process}
Let $E$ be a finite state space in the sense of
Definition \ref{def:fock-finite-state-space-submarkov}, and let $\partial$ be an absorbing state in the sense of
Definition \ref{def:fock-cemetery-absorbing-state}.
Let $X=\seq{X_t}{t\geq0}$ be a process with values in $E\cup\setone{\partial}$, and let $T$ be a nonnegative
random variable on the same probability space.
A process $\xi=\seq{\xi_t}{t\geq0}$ is called $X$ killed at $T$, or a killed process, when, for every sample point
$\omega$ and every $t\geq0$,
$$\begin{aligned}
\fun{\xi_t}{\omega}
& =
\begin{cases}
\fun{X_t}{\omega},
&
0\leq t<\fun{T}{\omega},
\\
\partial,
&
t\geq\fun{T}{\omega}.
\end{cases}
\end{aligned}$$
\end{defn}

\begin{lem}[Uniformization of the hopping kernel]\label{lem:fock-submarkov-uniformization}
For a finite $\Lambda\Subset\Gamma$, define
\begin{equation}
\begin{aligned}
Q_{\Lambda}
& =
-\physham[h]_{\Lambda}-C_h.
\end{aligned}
\label{eq:fock-killed-generator}
\end{equation}
The family $\napiernum^{tQ_{\Lambda}}$, $t\geq0$, is sub-Markovian.
For $R_{\Lambda}=-C_h^{-1}\physham[h]_{\Lambda}$, one has
\begin{equation}
\begin{aligned}
\napiernum^{tQ_{\Lambda}}
& =
\napiernum^{-tC_h}
\sum_{n\in\monnat}
\frac{\rbk{tC_h}^{n}}{n!}R_{\Lambda}^{n}.
\end{aligned}
\label{eq:fock-submarkov-uniformization}
\end{equation}
\end{lem}

\begin{proof}
The off-diagonal entries of $Q_{\Lambda}$ are nonnegative, and its row sums satisfy
$$\begin{aligned}
\sum_{y\in\Lambda}\rbk{Q_{\Lambda}}_{x,y}
& =
\abscard{\set{y\in\Lambda}{y\sim x}}-C_h
\leq
0.
\end{aligned}$$
The matrix $R_{\Lambda}$ has nonnegative entries and row sums at most $1$.
Every power $R_{\Lambda}^{n}$ has the same two properties.
Expanding $\napiernum^{tC_hR_{\Lambda}}$ proves
\eqref{eq:fock-submarkov-uniformization}.
The Poisson coefficients in that formula are nonnegative and sum to $1$.
Their convex combination has nonnegative entries and row sums at most $1$ for every $t\geq0$.
The matrix-exponential semigroup law completes the sub-Markovianity assertion.
\end{proof}

\begin{lem}[Probabilistic representation of $p_{\Lambda,t}$]\label{lem:fock-killed-walk-kernel}
Assume \eqref{eq:fock-path-chemical-potential-gap}, and let $\Lambda\Subset\Gamma$ be finite.
Let $Q_{\Lambda}$ be the operator defined by \eqref{eq:fock-killed-generator}.
Define $r_{\Lambda}$ and $p_{\Lambda,t}$, $t\geq0$, by
\begin{equation}
\begin{aligned}
r_{\Lambda}
& =
-\rbk{\smchemicalpotential_{\Lambda}+C_h},
\\
p_{\Lambda,t}
& =
\napiernum^{-t\rbk{\physham[h]_{\Lambda}-\smchemicalpotential_{\Lambda}}}
=
\napiernum^{-tr_{\Lambda}}\napiernum^{tQ_{\Lambda}}.
\end{aligned}
\label{eq:fock-killed-kernel-definition}
\end{equation}
Then, for every $x\in\Lambda$, there are a probability space
$\rbk{\Sigma_{x},\mblfml{F}_{x},\msrprb_x}$, a nonnegative random variable
$\tau_{\Lambda}$ with exponential law of rate $r_{\Lambda}$, and a càdlàg Markov process
$\xi=\seq{\xi_t}{t\geq0}$ on $\Lambda\cup\setone{\partial}$ in the sense of
Definitions \ref{def:fock-cadlag-path} and \ref{def:fock-time-homogeneous-markov-process} such that
$$\begin{aligned}
\fun{\msrprb_x}{\xi_0=x}
& =
1,
&
\fun{\msrprb_x}{\tau_{\Lambda}>t}
& =
\napiernum^{-r_{\Lambda}t},
\quad
t\geq0,
\end{aligned}$$
and $\partial$ is an absorbing state for $\xi$ in the sense of
Definition \ref{def:fock-cemetery-absorbing-state}.
The subscript $x$ in $\msrprb_x$ records the initial state $x$.
Its transition probabilities satisfy
\begin{equation}
\begin{aligned}
\fun{\msrprb_x}{\xi_t=y}
& =
\fun{p_{\Lambda,t}}{x,y},
&
x,y
& \in
\Lambda,
\\
\fun{\msrprb_x}{\xi_t=\partial}
& =
1-
\sum_{y\in\Lambda}\fun{p_{\Lambda,t}}{x,y}.
\end{aligned}
\label{eq:fock-killed-walk-transition}
\end{equation}
For every $t\geq0$, the transition probabilities satisfy the uniform row estimate
\begin{equation}
\begin{aligned}
\sum_{y \in \Lambda}\fun{p_{\Lambda,t}}{x,y}
& \leq
\napiernum^{-t\varepsilon}.
\end{aligned}
\label{eq:fock-killed-walk-row-bound}
\end{equation}
\end{lem}

\begin{proof}
Let $R_{\Lambda}$ be the substochastic kernel in
Lemma \ref{lem:fock-submarkov-uniformization}.
Its stochastic extension to $\Lambda\cup\setone{\partial}$ is
$$\begin{aligned}
\fun{\widehat{R}_{\Lambda}}{x,y}
& =
\fun{R_{\Lambda}}{x,y},
&
x,y
& \in
\Lambda,
\\
\fun{\widehat{R}_{\Lambda}}{x,\partial}
& =
1-
\sum_{y\in\Lambda}\fun{R_{\Lambda}}{x,y},
&
x
& \in
\Lambda,
\\
\fun{\widehat{R}_{\Lambda}}{\partial,\partial}
& =
1,
\\
\fun{\widehat{R}_{\Lambda}}{\partial,y}
& =
0,
&
y
& \in
\Lambda.
\end{aligned}$$
Thus $\partial$ is the cemetery state of $R_{\Lambda}$ and is absorbing for
$\widehat{R}_{\Lambda}$ in the sense of Definition \ref{def:fock-cemetery-absorbing-state}.
Let $\msrprb_x^{Y}$ be the law of the discrete-time Markov chain with transition kernel
$\widehat{R}_{\Lambda}$ and initial state $Y_0=x$.
Let $\msrprb^{N}$ be the law of a Poisson process of rate $C_h$.
On the product space with law $\msrprb_x^{X}=\msrprb_x^{Y}\otimes\msrprb^{N}$, define
$X_t=Y_{N_t}$.
For $x,y\in\Lambda$, the restriction of every $n$-step kernel of $\widehat{R}_{\Lambda}$ to
$\Lambda\times\Lambda$ is $R_{\Lambda}^{n}$.
The assertion holds for $n=0$ because both kernels are the identity on $\Lambda$.
If it holds for $n$, the absorbing-state identity
$\fun{\widehat{R}_{\Lambda}}{\partial,y}=0$ from
Definition \ref{def:fock-cemetery-absorbing-state} gives the induction step
$$\begin{aligned}
&
\fun{\widehat{R}_{\Lambda}^{n+1}}{x,y}
=
\sum_{z\in\Lambda}
\fun{\widehat{R}_{\Lambda}^{n}}{x,z}
\fun{\widehat{R}_{\Lambda}}{z,y}
+
\fun{\widehat{R}_{\Lambda}^{n}}{x,\partial}
\fun{\widehat{R}_{\Lambda}}{\partial,y}
\\ 
&=
\sum_{z\in\Lambda}
\fun{R_{\Lambda}^{n}}{x,z}
\fun{R_{\Lambda}}{z,y}
=
\fun{R_{\Lambda}^{n+1}}{x,y}.
\end{aligned}$$
The Markov property of $Y$ yields
$$\begin{aligned}
\fun{\msrprb_x^{Y}}{Y_n=y}
=
\fun{\widehat{R}_{\Lambda}^{n}}{x,y}
=
\fun{R_{\Lambda}^{n}}{x,y}
\end{aligned}$$
for every $n\in\monnat$ and $x,y\in\Lambda$.
The events $\setone{N_t=n}$, $n\in\monnat$, form a disjoint partition, and the product measure makes the
entire chain $Y$ independent of the Poisson process $N$.
The law of total probability and the Poisson distribution give
$$\begin{aligned}
&\fun{\msrprb_x^{X}}{X_t=y}
=
\sum_{n\in\monnat}
\fun{\msrprb_x^{X}}{X_t=y,\ N_t=n}
\\
& =
\sum_{n\in\monnat}
\fun{\msrprb_x^{X}}{Y_n=y,\ N_t=n}
=
\sum_{n\in\monnat}
\fun{\msrprb_x^{Y}}{Y_n=y}
\fun{\msrprb^{N}}{N_t=n}
\\
& =
\sum_{n\in\monnat}
\fun{R_{\Lambda}^{n}}{x,y}
\napiernum^{-tC_h}
\frac{\rbk{tC_h}^{n}}{n!}
=
\fun{
\napiernum^{-tC_h}
\sum_{n\in\monnat}
\frac{\rbk{tC_h}^{n}}{n!}R_{\Lambda}^{n}
}{x,y}
=
\fun{\napiernum^{tQ_{\Lambda}}}{x,y}.
\end{aligned}$$
The last equality is the uniformization formula
\eqref{eq:fock-submarkov-uniformization}.

\begin{samepage}
Let $\rbk{\Sigma_x^{X},\mblfml{F}_x^{X},\msrprb_x^{X}}$ denote the probability space carrying $X$ constructed
here, and let
$\rbk{\Sigma_{\tau_{\Lambda}},\mblfml{F}_{\tau_{\Lambda}},
\msrprb_{r_{\Lambda}}^{\tau_{\Lambda}}}$ carry a nonnegative random variable
$\tau_{\Lambda}$ with exponential law of rate $r_{\Lambda}$ in the sense of
Definition \ref{def:fock-exponential-law}.
The inequality $r_{\Lambda}\geq\varepsilon>0$ guarantees that this probability law is defined.
The probability space in the assertion is the product triple
$$\begin{aligned}
\pairbk{\Sigma_x,
\mblfml{F}_x,
\msrprb_x}
=
\pairbk{\Sigma_x^{X}\times\Sigma_{\tau_{\Lambda}},
\mblfml{F}_x^{X}\otimes\mblfml{F}_{\tau_{\Lambda}},
\msrprb_x^{X}\otimes\msrprb_{r_{\Lambda}}^{\tau_{\Lambda}}},
\end{aligned}$$
\end{samepage}
where $X$ and $\tau_{\Lambda}$ are lifted to $\Sigma_x$ by the coordinate projections.
For $\omega_X\in\Sigma_x^{X}$ and
$\omega_{\tau_{\Lambda}}\in\Sigma_{\tau_{\Lambda}}$, define $\xi = \net{\xi_t}{t \geq 0}$ as $X$ killed at $\tau_{\Lambda}$ in the sense
of Definition \ref{def:fock-killed-process} by
\begin{equation}
\begin{aligned}
\fun{\xi_t}{\omega_X,\omega_{\tau_{\Lambda}}}
& =
\begin{cases}
\fun{X_t}{\omega_X},
&
0\leq t<\fun{\tau_{\Lambda}}{\omega_{\tau_{\Lambda}}},
\\
\partial,
&
t\geq\fun{\tau_{\Lambda}}{\omega_{\tau_{\Lambda}}}.
\end{cases}
\end{aligned}
\label{eq:fock-killed-process-definition}
\end{equation}
Because $\tau_{\Lambda}$ is a random variable on $\Sigma_x$, it is not a deterministic parameter indexing the
process $\xi$.
For $y\in\Lambda$, the endpoint event is
$$\begin{aligned}
\setone{\xi_t=y}
& =
\setone{\tau_{\Lambda}>t}\cap\setone{X_t=y}.
\end{aligned}$$
The product law makes $\tau_{\Lambda}$ independent of the entire process $X$.
The definition of the exponential law and the transition identity for $X$ give
$$\begin{aligned}
\fun{\msrprb_x}{\xi_t=y}
& =
\fun{\msrprb_x}{\tau_{\Lambda}>t}
\fun{\msrprb_x}{X_t=y}
=
\napiernum^{-r_{\Lambda}t}
\fun{\napiernum^{tQ_{\Lambda}}}{x,y}
=
\fun{p_{\Lambda,t}}{x,y}.
\end{aligned}$$
The sub-Markov row sum and $r_{\Lambda}\geq\varepsilon$ give the claimed uniform estimate through
$$\begin{aligned}
\sum_{y\in\Lambda}\fun{p_{\Lambda,t}}{x,y}
=
\napiernum^{-r_{\Lambda}t}
\sum_{y\in\Lambda}
\fun{\napiernum^{tQ_{\Lambda}}}{x,y}
\leq
\napiernum^{-r_{\Lambda}t}
\leq
\napiernum^{-\varepsilon t}.
\end{aligned}$$
Subtracting this row sum from $1$ gives the cemetery-state formula in
\eqref{eq:fock-killed-walk-transition}.
The Poisson time change preserves the Markov property of the chain $Y$.
For $s,t\geq0$, Definition \ref{def:fock-exponential-law} gives
$$\begin{aligned}
\funcond{\msrprb_x}{\tau_{\Lambda}>s+t}{\tau_{\Lambda}>s}
& =
\frac{\napiernum^{-r_{\Lambda}\rbk{s+t}}}{\napiernum^{-r_{\Lambda}s}}
=
\napiernum^{-r_{\Lambda}t}.
\end{aligned}$$
Together with the Markov property of $X$ and the independence of $X$ and $\tau_{\Lambda}$, this identity proves
that $\xi$ is Markov.
The identities $\fun{\widehat{R}_{\Lambda}}{\partial,\partial}=1$ and
\eqref{eq:fock-killed-process-definition} show directly that $\xi_s=\partial$ implies $\xi_t=\partial$ for
every $t\geq s$; hence $\partial$ is absorbing for $\xi$ according to
Definition \ref{def:fock-cemetery-absorbing-state}.
\end{proof}

The continuous-time bridge path space is defined by \begin{equation}
\begin{aligned}
\Omega_{\Lambda;x,y}^{t}
& =
\set{\omega\colon\closedinterval{0}{t}\to\Lambda}
{\text{$\omega$ is càdlàg, has finitely many jumps,
$\omega_0=x $ and $\omega_t=y$}}.
\end{aligned}
\label{eq:fock-continuous-bridge-path-space}
\end{equation} For \(\omega\in\Omega_{\Lambda;x,y}^{t}\) and \(0\leq s\leq t\), the notation \(\omega_s\) denotes the value of \(\omega\) at time \(s\). This path space carries the sigma-algebra generated by its coordinate evaluation maps.

\begin{lem}[Unnormalized bridge measure]\label{lem:fock-unnormalized-bridge-measure}
Assume \eqref{eq:fock-path-chemical-potential-gap}, and let $\Lambda\Subset\Gamma$ be finite.
Fix $x,y\in\Lambda$ and $t>0$.
Let $p_{\Lambda,t}$ be defined by \eqref{eq:fock-killed-kernel-definition}, and let
$\rbk{\Sigma_x,\mblfml{F}_x,\msrprb_x}$ and $\xi$ be the probability space and process whose existence is
asserted in Lemma \ref{lem:fock-killed-walk-kernel} for the initial state $x$.
Let $\Omega_{\Lambda;x,y}^{t}$ be the measurable path space defined by
\eqref{eq:fock-continuous-bridge-path-space}.
There is a finite positive measure $\msrcal{W}_{\Lambda;x,y}^{t}$ on this space such that, for every
$n\in\semigrposint$, $0=s_0<s_1<\cdots<s_n=t$, and
$z_0=x,z_1,\ldots,z_{n-1},z_n=y$ in $\Lambda$,
\begin{equation}
\begin{aligned}
\fun{\msrcal{W}_{\Lambda;x,y}^{t}}{
\set{\omega\in\Omega_{\Lambda;x,y}^{t}}
{\omega_{s_j}=z_j,\ 1\leq j\leq n-1}}
& =
\prod_{j=0}^{n-1}
\fun{p_{\Lambda,s_{j+1}-s_j}}{z_j,z_{j+1}}.
\end{aligned}
\label{eq:fock-bridge-cylinder-distribution}
\end{equation}
Its total mass is
\begin{equation}
\begin{aligned}
\fun{\msrcal{W}_{\Lambda;x,y}^{t}}{\Omega_{\Lambda;x,y}^{t}}
& =
\fun{p_{\Lambda,t}}{x,y}.
\end{aligned}
\label{eq:fock-bridge-total-mass}
\end{equation}
For $s,t>0$, $x,y,z\in\Lambda$,
$\omega\in\Omega_{\Lambda;x,z}^{s}$, and
$\eta\in\Omega_{\Lambda;z,y}^{t}$, define their concatenation by
$$\begin{aligned}
\rbk{\omega\diamond_s\eta}_u
& =
\begin{cases}
\omega_u, & 0\leq u\leq s,\\
\eta_{u-s}, & s<u\leq s+t.
\end{cases}
\end{aligned}$$
For every nonnegative measurable function $F$ on
$\bigsqcup_{x,y\in\Lambda}\Omega_{\Lambda;x,y}^{s+t}$, the bridge measures satisfy
\begin{equation}
\begin{aligned}
&
\sum_{z\in\Lambda}
\int_{\Omega_{\Lambda;x,z}^{s}}
\int_{\Omega_{\Lambda;z,y}^{t}}
\fun{F}{\omega\diamond_s\eta}
\opdmsr{\msrcal{W}_{\Lambda;z,y}^{t}}(\eta)
\opdmsr{\msrcal{W}_{\Lambda;x,z}^{s}}(\omega)
=
\int_{\Omega_{\Lambda;x,y}^{s+t}}
\fun{F}{\zeta}
\opdmsr{\msrcal{W}_{\Lambda;x,y}^{s+t}}(\zeta).
\end{aligned}
\label{eq:fock-bridge-convolution-law}
\end{equation}
\end{lem}

\begin{proof}
For a measurable set $A\subset\Omega_{\Lambda;x,y}^{t}$, define
$$\begin{aligned}
\fun{\msrcal{W}_{\Lambda;x,y}^{t}}{A}
& =
\fun{\msrprb_x}{
\setone{\fnrestr{\xi}{\closedinterval{0}{t}}\in A}}.
\end{aligned}$$
This restriction of the path law is a finite positive measure.
Fix $0=s_0<s_1<\cdots<s_n=t$ and
$z_0=x,z_1,\ldots,z_{n-1},z_n=y$ in $\Lambda$.
For $0\leq u\leq t$, let $\mblfml{F}_{\xi,u}$ be the sigma-algebra generated by
$\xi_v$, $0\leq v\leq u$.
Let $G_0$ be the whole sample space on which $\xi$ is defined.
For $1\leq j\leq n$, define the event
$$\begin{aligned}
G_j
& =
\bigcap_{\ell=1}^{j}
\setone{\xi_{s_{\ell}}=z_{\ell}}.
\end{aligned}$$
On $G_{j-1}$, the process satisfies
$\xi_{s_{j-1}}=z_{j-1}\in\Lambda$.
The time-homogeneous Markov property in
Definition \ref{def:fock-time-homogeneous-markov-process} and the transition formula
\eqref{eq:fock-killed-walk-transition} identify the conditional probability on this event as
$$\begin{aligned}
\fndef{G_{j-1}}
\sqfuncond{\prbexp_{\msrprb_x}}{
\fndef{\setone{\xi_{s_j}=z_j}}
}{\mblfml{F}_{\xi,s_{j-1}}}
& =
\fndef{G_{j-1}}
\fun{p_{\Lambda,s_j-s_{j-1}}}{z_{j-1},z_j}
\end{aligned}$$
for $1\leq j\leq n$.
Since $G_{j-1}\in\mblfml{F}_{\xi,s_{j-1}}$, the defining property of conditional expectation gives
$$\begin{aligned}
\fun{\msrprb_x}{G_j}
& =
\sqfun{\prbexp_{\msrprb_x}}{
\fndef{G_{j-1}}
\fndef{\setone{\xi_{s_j}=z_j}}
}
=
\sqfun{\prbexp_{\msrprb_x}}{
\fndef{G_{j-1}}
\sqfuncond{\prbexp_{\msrprb_x}}{
\fndef{\setone{\xi_{s_j}=z_j}}
}{\mblfml{F}_{\xi,s_{j-1}}}
}
\\
& =
\sqfun{\prbexp_{\msrprb_x}}{
\fndef{G_{j-1}}
\fun{p_{\Lambda,s_j-s_{j-1}}}{z_{j-1},z_j}
}
=
\fun{p_{\Lambda,s_j-s_{j-1}}}{z_{j-1},z_j}
\fun{\msrprb_x}{G_{j-1}}.
\end{aligned}$$
Since $G_0$ is the whole sample space, $\fun{\msrprb_x}{G_0}=1$.
Iterating the recursion from this initial value yields
$$\begin{aligned}
\fun{\msrprb_x}{G_n}
=
\prod_{j=1}^{n}
\fun{p_{\Lambda,s_j-s_{j-1}}}{z_{j-1},z_j}
=
\prod_{j=0}^{n-1}
\fun{p_{\Lambda,s_{j+1}-s_j}}{z_j,z_{j+1}}.
\end{aligned}$$
Because $\partial$ is absorbing according to Definition \ref{def:fock-cemetery-absorbing-state}, the identity
$z_n=y\in\Lambda$ shows that the event $G_n$ implies
$\xi_s\in\Lambda$ for every $0\leq s\leq t$.
Since $\xi$ is càdlàg and takes values in the finite state space $\Lambda\cup\setone{\partial}$, its paths have
finitely many jumps on every bounded time interval.
The inverse image of the cylinder set in
\eqref{eq:fock-bridge-cylinder-distribution} under
$\fnrestr{\xi}{\closedinterval{0}{t}}$ is $G_n$ up to a $\msrprb_x$-null set.
The definition of $\msrcal{W}_{\Lambda;x,y}^{t}$ and the formula for
$\fun{\msrprb_x}{G_n}$ prove
\eqref{eq:fock-bridge-cylinder-distribution}.
Taking no intermediate times in that formula gives \eqref{eq:fock-bridge-total-mass}.
For an indicator function of a cylinder set, the left-hand side of
\eqref{eq:fock-bridge-convolution-law} is evaluated by
\eqref{eq:fock-bridge-cylinder-distribution}; summation over the common value $z$ and the semigroup identity
in Definition \ref{def:fock-finite-state-space-submarkov} give the right-hand side.
Finite linear combinations and monotone convergence extend the identity to every nonnegative measurable
function $F$.
\end{proof}

\subsection{Open paths, loops, and the interacting normalization}\label{open-paths-loops-and-the-interacting-normalization}

The kernels from Lemma \ref{lem:fock-killed-walk-kernel} and the bridge measures from Lemma \ref{lem:fock-unnormalized-bridge-measure} provide the positive path-space reference measure for the interaction.

\begin{defn}[Finite open-path tuple]\label{def:fock-finite-open-path-tuple}
Let $\Lambda\Subset\Gamma$ be finite, let $\beta>0$, and let $m\in\monnat$.
For $1\leq i\leq m$, let $q_i\in\semigrposint$ and $x_i,y_i\in\Lambda$.
A finite open-path tuple is a tuple
$\boldsymbol{\omega}=\rbk{\omega_1,\ldots,\omega_m}$ satisfying
$$\begin{aligned}
\omega_i
\in
\Omega_{\Lambda;x_i,y_i}^{q_i\beta},
\quad
1\leq i\leq m,
\end{aligned}$$
where the path spaces are defined by \eqref{eq:fock-continuous-bridge-path-space}.
For $m=0$, the tuple is empty.
\end{defn}

\begin{defn}[Occupation field]\label{def:fock-path-occupation-field}
For $q\in\semigrposint$, $x,y\in\Lambda$, and
$\omega\in\Omega_{\Lambda;x,y}^{q\beta}$, define the occupation field of $\omega$ by
$$\begin{aligned}
\fun{n_{\omega}}{s,z}
& =
\sum_{j = 0}^{q - 1}
\fun{\fndef{\setone{z}}}{\omega_{s + j\beta}},
\quad
0 \leq s \leq \beta,
\quad
z \in \Lambda.
\end{aligned}$$
Let $\boldsymbol{\xi}=\rbk{\xi_1,\ldots,\xi_m}$ be a finite open-path tuple in the sense of
Definition \ref{def:fock-finite-open-path-tuple}, with
$\xi_i\in\Omega_{\Lambda;x_i,y_i}^{q_i\beta}$ for $1\leq i\leq m$.
For $0\leq s\leq q_i\beta$, the symbol $\xi_{i,s}$ denotes the value of the $i$-th path at time $s$.
The occupation field of $\boldsymbol{\xi}$ on $\closedinterval{0}{\beta} \times \Lambda$ is defined by
\begin{equation}
\begin{aligned}
\fun{n_{\boldsymbol{\xi}}}{s,z}
=
\sum_{i = 1}^{m}\fun{n_{\xi_i}}{s,z}
=
\sum_{i = 1}^{m}
\sum_{j = 0}^{q_i - 1}
\fun{\fndef{\setone{z}}}{\xi_{i,s + j\beta}},
\quad
\rbk{s,z}
\in
\closedinterval{0}{\beta} \times \Lambda.
\end{aligned}
\label{eq:fock-path-tuple-occupation-field}
\end{equation}
For $m = 0$, both sums in \eqref{eq:fock-path-tuple-occupation-field} are zero.
\end{defn}

For fixed \(s\in\closedinterval{0}{\beta}\) and \(z\in\Lambda\), the single-path definition is equivalently \[\begin{aligned}
\fun{n_{\omega}}{s,z}
& =
\abscard{
\set{j\in\monnat}{0\leq j\leq q-1,\ \omega_{s+j\beta}=z}}.
\end{aligned}\] It counts how many of the \(q\) values \(\omega_s,\omega_{s+\beta},\ldots,\omega_{s+\rbk{q-1}\beta}\) are equal to \(z\), and hence \[\begin{aligned}
\sum_{z\in\Lambda}\fun{n_{\omega}}{s,z}
& =
q.
\end{aligned}\] For fixed \(\rbk{s,z}\in\closedinterval{0}{\beta}\times\Lambda\), this definition is the counting identity \[\begin{aligned}
\fun{n_{\boldsymbol{\xi}}}{s,z}
& =
\abscard{
\set{\rbk{i,j}}{
1\leq i\leq m,\ 0\leq j\leq q_i-1,\ \xi_{i,s+j\beta}=z}}.
\end{aligned}\] It therefore records the total number of path values at the lattice site \(z\) for the functional-integral time parameter \(s\), and \[\begin{aligned}
\sum_{z\in\Lambda}\fun{n_{\boldsymbol{\xi}}}{s,z}
& =
\sum_{i=1}^{m}q_i.
\end{aligned}\]

The interaction energy determined by Definition \ref{def:fock-path-occupation-field} and \eqref{eq:fock-path-tuple-occupation-field} is \begin{equation}
\begin{aligned}
\fun{\Phi_{\Lambda}}{\boldsymbol{\xi}}
& =
\int_{\closedinterval{0}{\beta}}
\sum_{x,y \in \Lambda}
\fun{v}{x,y}
\rbk{
\fun{n_{\boldsymbol{\xi}}}{s,x}\fun{n_{\boldsymbol{\xi}}}{s,y}
-
\fun{\fndef{\setone{x}}}{y}\fun{n_{\boldsymbol{\xi}}}{s,x}}
\opdmsr{s}.
\end{aligned}
\label{eq:fock-path-interaction-energy}
\end{equation} The subtraction in \eqref{eq:fock-path-interaction-energy} is the normal-ordering term. For finite open-path tuples in the sense of Definition \ref{def:fock-finite-open-path-tuple}, \(\boldsymbol{\omega}=\rbk{\omega_1,\ldots,\omega_k}\) and \(\boldsymbol{\zeta}=\rbk{\zeta_1,\ldots,\zeta_{\ell}}\), define the interaction energy of their concatenated tuple by \begin{equation}
\begin{aligned}
\fun{\Phi_{\Lambda}}{\boldsymbol{\omega},\boldsymbol{\zeta}}
& =
\fun{\Phi_{\Lambda}}{\omega_1,\ldots,\omega_k,\zeta_1,\ldots,\zeta_{\ell}},
&
k,\ell
& \in
\monnat.
\end{aligned}
\label{eq:fock-concatenated-path-interaction-energy}
\end{equation}

Using the bridge path spaces defined in \eqref{eq:fock-continuous-bridge-path-space}, define the loop path space by \begin{equation}
\begin{aligned}
\Omega_{\Lambda}^{\mathrm{loop}}
& =
\bigsqcup_{\substack{q \in \semigrposint \\ x \in \Lambda}}
\Omega_{\Lambda;x,x}^{q\beta}.
\end{aligned}
\label{eq:fock-loop-path-space}
\end{equation} The positive ideal loop measure \(\msrcal{L}_{\Lambda}\) on the space in \eqref{eq:fock-loop-path-space} is defined by \begin{equation}
\begin{aligned}
\int_{\Omega_{\Lambda}^{\mathrm{loop}}}
\fun{F}{\zeta}\opdmsr{\msrcal{L}_{\Lambda}}(\zeta)
& =
\sum_{q \in \semigrposint}\frac{1}{q}
\sum_{x \in \Lambda}
\int_{\Omega_{\Lambda;x,x}^{q\beta}}
\fun{F}{\omega}\opdmsr{\msrcal{W}_{\Lambda;x,x}^{q\beta}}(\omega)
\end{aligned}
\label{eq:fock-ideal-loop-measure}
\end{equation} for every nonnegative measurable \(F\). For a finite open-path tuple \(\boldsymbol{\omega}\), define \begin{equation}
\begin{aligned}
\fun{Z_{\Lambda}}{\boldsymbol{\omega}}
=
\sum_{\ell \in \monnat}\frac{1}{\ell!}
\int_{\rbk{\Omega_{\Lambda}^{\mathrm{loop}}}^{\ell}}
\napiernum^{-\fun{\Phi_{\Lambda}}{\boldsymbol{\omega},\boldsymbol{\zeta}}}
\opdmsr{\msrcal{L}_{\Lambda}^{\otimes\ell}}(\boldsymbol{\zeta}),
\quad
Z_{\Lambda}
=
\sum_{\ell \in \monnat}\frac{1}{\ell!}
\int_{\rbk{\Omega_{\Lambda}^{\mathrm{loop}}}^{\ell}}
\napiernum^{-\fun{\Phi_{\Lambda}}{\boldsymbol{\zeta}}}
\opdmsr{\msrcal{L}_{\Lambda}^{\otimes\ell}}(\boldsymbol{\zeta}).
\end{aligned}
\label{eq:fock-open-path-loop-partition-functions}
\end{equation} Set \(F=1\) in the loop-measure definition \eqref{eq:fock-ideal-loop-measure}. The bridge-mass identity \eqref{eq:fock-bridge-total-mass}, nonnegativity of \(p_{\Lambda,q\beta}\), and the row bound \eqref{eq:fock-killed-walk-row-bound} give \[\begin{aligned}
& \int_{\Omega_{\Lambda}^{\mathrm{loop}}}\opdmsr{\msrcal{L}_{\Lambda}}(\zeta)
=
\sum_{q \in \semigrposint}\frac{1}{q}
\sum_{x \in \Lambda}
\fun{\msrcal{W}_{\Lambda;x,x}^{q\beta}}{\Omega_{\Lambda;x,x}^{q\beta}}
=
\sum_{q \in \semigrposint}\frac{1}{q}
\sum_{x \in \Lambda}
\fun{p_{\Lambda,q\beta}}{x,x}
\\
& \leq
\sum_{q \in \semigrposint}\frac{1}{q}
\sum_{x \in \Lambda}
\sum_{y \in \Lambda}
\fun{p_{\Lambda,q\beta}}{x,y}
\leq
\sum_{q \in \semigrposint}\frac{1}{q}
\sum_{x \in \Lambda}
\napiernum^{-q\beta\varepsilon}
\\ 
& =
\abscard{\Lambda}
\sum_{q \in \semigrposint}\frac{\napiernum^{-q\beta\varepsilon}}{q}
=
-\abscard{\Lambda}
\log\rbk{1-\napiernum^{-\beta\varepsilon}}
<
\infty.
\end{aligned}\] For every \(q\in\semigrposint\) and \(x,y\in\Lambda\), nonnegativity and \eqref{eq:fock-killed-walk-row-bound} give the pointwise estimate \[\begin{aligned}
\fun{p_{\Lambda,q\beta}}{x,y}
\leq
\sum_{z \in \Lambda}
\fun{p_{\Lambda,q\beta}}{x,z}
\leq
\napiernum^{-q\beta\varepsilon}.
\end{aligned}\] Summing this estimate over \(q\) and evaluating the resulting geometric series give the open-path bound \[\begin{aligned}
\sum_{q \in \semigrposint}\fun{p_{\Lambda,q\beta}}{x,y}
\leq
\sum_{q \in \semigrposint}\napiernum^{-q\beta\varepsilon}
=
\napiernum^{-\beta\varepsilon}
\sum_{m \in \monnat}
\rbk{\napiernum^{-\beta\varepsilon}}^{m}
=
\frac{\napiernum^{-\beta\varepsilon}}{1-\napiernum^{-\beta\varepsilon}}
=
\frac{1}{\napiernum^{\beta\varepsilon} - 1}.
\end{aligned}\] The loop series in \eqref{eq:fock-open-path-loop-partition-functions} is consequently absolutely convergent and \(1 \leq Z_{\Lambda} < \infty\).

For \(k \in \semigrposint\) and \(\boldsymbol{x},\boldsymbol{y} \in \Lambda^k\), the reduced density kernel is \begin{equation}
\begin{aligned}
\fun{\widetilde{\rho}_{\Lambda,k}}{\boldsymbol{x};\boldsymbol{y}}
& =
\fun{\oastate[\psi_{\Lambda}]}{
\faadj{a_{y_1}}\cdots\faadj{a_{y_k}}a_{x_k}\cdots a_{x_1}}.
\end{aligned}
\label{eq:fock-reduced-density-kernel}
\end{equation} The superstable trace majorant in Proposition \ref{prop:fock-finite-volume-gibbs-existence} makes the trace defining \eqref{eq:fock-reduced-density-kernel} finite.

The proof of the open-path expansion has nine parts.

\begin{enumerate}
\def\labelenumi{\arabic{enumi}.}

\item
  Lemma \ref{lem:fock-interaction-occupation-eigenvalue} computes the interaction in the occupation-number basis.
\item
  Lemma \ref{lem:fock-finite-slice-weight-properties} records the positivity and row bound of the finite-slice weights.
\item
  Lemma \ref{lem:fock-finite-slice-trace-formula} expands the finite-slice trace in particle coordinates.
\item
  Lemma \ref{lem:fock-creation-annihilation-sector-formula} inserts the creation--annihilation strings.
\item
  Lemma \ref{lem:fock-marked-index-permutation-decomposition} reorganizes the permutation sum by the index sequences generated by repeated application of a coordinate-index permutation.
\item
  Lemma \ref{lem:fock-finite-slice-open-path-expansion} converts the coordinate products into bridge and loop weights.
\item
  Lemma \ref{lem:fock-finite-slice-domination} supplies uniform bounds, and Lemma \ref{lem:fock-continuous-time-limit} uses them to take the continuous-time limit.
\item
  Lemma \ref{lem:fock-continuous-bosonic-cycle-decomposition} identifies the operator and loop partition functions.
\item
  Lemma \ref{lem:fock-loop-marked-chain-expansion} combines these results.
\end{enumerate}

For a finite \(\Lambda\Subset\Gamma\) and \(\boldsymbol{\nu}=\seq{\nu_x}{x\in\Lambda}\in\monnat^{\Lambda}\), define \begin{equation}
\begin{aligned}
\fun{E_{\Lambda}}{\boldsymbol{\nu}}
& =
\sum_{x,y\in\Lambda}
\fun{v}{x,y}
\rbk{\nu_x\nu_y-\fun{\fndef{\setone{x}}}{y}\nu_x}.
\end{aligned}
\label{eq:fock-finite-slice-occupation-energy}
\end{equation}

\begin{lem}[Interaction in the occupation-number basis]\label{lem:fock-interaction-occupation-eigenvalue}
Let $\Lambda\Subset\Gamma$ be finite and let
$\boldsymbol{\nu}=\seq{\nu_x}{x\in\Lambda}\in\monnat^{\Lambda}$.
The occupation-number vector $\Psi_{\Lambda,\boldsymbol{\nu}}$ in
\eqref{eq:fock-occupation-number-basis} satisfies
$$\begin{aligned}
V_{\Lambda}\Psi_{\Lambda,\boldsymbol{\nu}}
& =
\fun{E_{\Lambda}}{\boldsymbol{\nu}}\Psi_{\Lambda,\boldsymbol{\nu}}.
\end{aligned}$$
Moreover, $\fun{E_{\Lambda}}{\boldsymbol{\nu}}\geq0$.
\end{lem}

\begin{proof}
If $x\neq y$, the operators $a_xa_y$ contribute
$\sqrt{\nu_x\nu_y}$ and the operators $\faadj{a_x}\faadj{a_y}$ contribute the same factor.  Hence
$$\begin{aligned}
\faadj{a_x}\faadj{a_y}a_xa_y\Psi_{\Lambda,\boldsymbol{\nu}}
& =
\nu_x\nu_y\Psi_{\Lambda,\boldsymbol{\nu}}.
\end{aligned}$$
If $x=y$, the two annihilation operators contribute $\sqrt{\nu_x}\sqrt{\nu_x-1}$ and the two creation
operators contribute the same factor, so
$$\begin{aligned}
\rbk{\faadj{a_x}}^2a_x^2\Psi_{\Lambda,\boldsymbol{\nu}}
& =
\nu_x\rbk{\nu_x-1}\Psi_{\Lambda,\boldsymbol{\nu}}.
\end{aligned}$$
These two identities give \eqref{eq:fock-finite-slice-occupation-energy} after summation over the ordered pairs
$\rbk{x,y}\in\Lambda^2$.  Every summand is nonnegative because $v$ is pointwise nonnegative and
$\nu_x\in\monnat$.
\end{proof}

To make the limiting step explicit, fix the particle sector \(N\) and an integer \(m \geq 1\), split \(\closedinterval{0}{\beta}\) into \(m\) intervals, and insert the occupation-number basis \eqref{eq:fock-occupation-number-basis} between the factors in \begin{equation}
\begin{aligned}
T_{\Lambda,m}
=
\rbk{\fnexp{-\frac{\beta}{m}
\rbk{\physham_{\txtfr,\Lambda} -
\smchemicalpotential_{\Lambda}\opfocknumber_{\Lambda}}}
\cdot
\fnexp{-\frac{\beta}{m} V_{\Lambda}}}^{m}.
\end{aligned}
\label{eq:fock-finite-slice-trotter-product}
\end{equation} The operators \(\physham_{\txtfr,\Lambda}\) and \(V_{\Lambda}\) are defined in \eqref{eq:fock-bose-hubbard-hamiltonian}. Define the one-step kernel \(b_m\) by \begin{equation}
\begin{aligned}
\fun{b_m}{x,y}
& =
\fun{\fnexp{-\frac{\sminvtemperature}{m}
\rbk{\physham[h]_{\Lambda} - \smchemicalpotential_{\Lambda}}}}{x,y}.
\end{aligned}
\label{eq:fock-finite-slice-one-step-kernel}
\end{equation} For \(\boldsymbol{x}^{(r)} \in \Lambda^N\), write \(\nu_z^{(r)} = \sum_{i = 1}^{N}\fun{\fndef{\setone{z}}}{x_i^{(r)}}\). For \(\varpi \in \grsym{N}\), let \(\Omega_{\varpi,m}^{(N)}\) be the set of tuples \(\rbk{\boldsymbol{x}^{(0)},\ldots,\boldsymbol{x}^{(m)}}\) in \(\rbk{\Lambda^N}^{m+1}\) satisfying \(x_i^{(m)} = x_{\fun{\varpi}{i}}^{(0)}\) for every \(i\), and define \begin{equation}
\begin{aligned}
\fun{B_{N,m}}{\boldsymbol{x}^{(0)},\ldots,\boldsymbol{x}^{(m)}}
& =
\prod_{r = 0}^{m - 1}
\rbk{
\prod_{i = 1}^{N}
\fun{b_m}{x_i^{(r)},x_i^{(r + 1)}}}
\napiernum^{-\frac{\sminvtemperature}{m}\fun{E_{\Lambda}}{\nu^{(r)}}}.
\end{aligned}
\label{eq:fock-finite-slice-sector-weight}
\end{equation} For each time index \(r\), the product of the kernels \(b_m\) is the one-particle contribution from \(\boldsymbol{x}^{(r)}\) to \(\boldsymbol{x}^{(r+1)}\), and the exponential is the interaction contribution of the occupation numbers \(\nu^{(r)}\). Thus \(B_{N,m}\) assigns a number to each tuple \(\rbk{\boldsymbol{x}^{(0)},\ldots,\boldsymbol{x}^{(m)}}\) in \(\Omega_{\varpi,m}^{(N)}\).

\begin{lem}[Properties of the finite-slice weights]\label{lem:fock-finite-slice-weight-properties}
Assume \eqref{eq:fock-path-chemical-potential-gap}, let $\Lambda\Subset\Gamma$ be finite, and let $\beta>0$.
For every $m\in\semigrposint$ and $x\in\Lambda$, the kernel in
\eqref{eq:fock-finite-slice-one-step-kernel} satisfies
\begin{equation}
\begin{aligned}
\fun{b_m}{x,y}
& \geq0,
&& x,y\in\Lambda,
&
\sum_{y\in\Lambda}\fun{b_m}{x,y}
& \leq
\napiernum^{-\frac{\beta\varepsilon}{m}}.
\end{aligned}
\label{eq:fock-finite-slice-one-step-row-bound}
\end{equation}
For every $N\in\monnat$, $\varpi\in\grsym{N}$, and
$\boldsymbol{x}\in\Omega_{\varpi,m}^{(N)}$, the weight in
\eqref{eq:fock-finite-slice-sector-weight} satisfies
\begin{equation}
\begin{aligned}
0
\leq
\fun{B_{N,m}}{\boldsymbol{x}}
\leq
\prod_{r=0}^{m-1}\prod_{i=1}^{N}
\fun{b_m}{x_i^{(r)},x_i^{(r+1)}}.
\end{aligned}
\label{eq:fock-finite-slice-sector-weight-bound}
\end{equation}
\end{lem}

\begin{proof}
Definitions \eqref{eq:fock-killed-kernel-definition} and
\eqref{eq:fock-finite-slice-one-step-kernel} give
$b_m=p_{\Lambda,\beta/m}$.  Nonnegativity of $p_{\Lambda,t}$ and the row estimate
\eqref{eq:fock-killed-walk-row-bound}, evaluated at $t=\beta/m$, give
\eqref{eq:fock-finite-slice-one-step-row-bound}.
Lemma \ref{lem:fock-interaction-occupation-eigenvalue} gives
$\fun{E_{\Lambda}}{\nu^{(r)}}\geq0$ for every $r$.
Consequently every factor in \eqref{eq:fock-finite-slice-sector-weight} is nonnegative, and each exponential
factor is at most $1$.  This proves \eqref{eq:fock-finite-slice-sector-weight-bound}.
\end{proof}

\begin{lem}[Finite-slice trace formula]\label{lem:fock-finite-slice-trace-formula}
Assume \eqref{eq:fock-path-chemical-potential-gap}, let $\Lambda\Subset\Gamma$ be finite, and let $\beta>0$.
For every $N\in\monnat$ and $m\in\semigrposint$, the Trotter product
\eqref{eq:fock-finite-slice-trotter-product} satisfies
\begin{equation}
\begin{aligned}
\sqfun{\trace_{P_{\Lambda,N}\sphilb{F}_{\Lambda}}}{T_{\Lambda,m}}
& =
\frac{1}{N!}
\sum_{\varpi \in \grsym{N}}
\sum_{\boldsymbol{x} \in \Omega_{\varpi,m}^{(N)}}
\fun{B_{N,m}}{\boldsymbol{x}}.
\end{aligned}
\label{eq:fock-finite-slice-trace-formula}
\end{equation}
\end{lem}

\begin{proof}
Lemma \ref{lem:fock-interaction-occupation-eigenvalue} gives the diagonal matrix element of the interaction
factor, and \eqref{eq:fock-finite-slice-one-step-kernel} gives the one-particle matrix element of the free
factor.  Identify $P_{\Lambda,N}\sphilb{F}_{\Lambda}$ with the symmetric subspace of
$\lpseq^2\rbk{\Lambda^N}$.  For $\boldsymbol{x}\in\Lambda^N$, let $\delta_{\boldsymbol{x}}$ be the function on
$\Lambda^N$ that equals $1$ at $\boldsymbol{x}$ and $0$ at every other tuple.  Insertion of the basis
$\set{\delta_{\boldsymbol{x}}}{\boldsymbol{x}\in\Lambda^N}$ between consecutive factors of
\eqref{eq:fock-finite-slice-trotter-product} assigns the product
\eqref{eq:fock-finite-slice-sector-weight} to each sequence of intermediate coordinates.
The bosonic symmetrizer $S_N$ on $\lpseq^2\rbk{\Lambda^N}$ is defined by
$$\begin{aligned}
S_N\delta_{\boldsymbol{x}}
& =
\frac{1}{N!}
\sum_{\varpi\in\grsym{N}}
\delta_{\rbk{x_{\fun{\varpi}{1}},\ldots,x_{\fun{\varpi}{N}}}}.
\end{aligned}$$
Therefore its endpoint condition is
$x_i^{(m)}=x_{\fun{\varpi}{i}}^{(0)}$.  Summing the intermediate coordinates therefore gives
\eqref{eq:fock-finite-slice-trace-formula}.  The summands are nonnegative by
Lemma \ref{lem:fock-finite-slice-weight-properties}.
\end{proof}

For \(q \in \semigrposint\), let \(\Omega_{\Lambda;x,y}^{q,m}\) consist of sequences \(\omega = \seq{\omega_r}{r = 0}^{q m}\) with \(\omega_0 = x\) and \(\omega_{q m} = y\), and define \begin{equation}
\begin{aligned}
\fun{\msrcal{W}_{\Lambda;x,y}^{q\beta,(m)}}{\setone{\omega}}
& =
\prod_{r = 0}^{q m - 1}\fun{b_m}{\omega_r,\omega_{r + 1}},
\\
\fun{n_{\omega}^{(m)}}{r,z}
& =
\sum_{j = 0}^{q - 1}
\fun{\fndef{\setone{z}}}{\omega_{r + j m}},
\quad
0 \leq r \leq m - 1,
\\
\fun{\Phi_{\Lambda}^{(m)}}{\boldsymbol{\xi}}
& =
\frac{\sminvtemperature}{m}
\sum_{r = 0}^{m - 1}
\sum_{x,y \in \Lambda}
\fun{v}{x,y}
\rbk{
\fun{n_{\boldsymbol{\xi}}^{(m)}}{r,x}
\fun{n_{\boldsymbol{\xi}}^{(m)}}{r,y}
-
\fun{\fndef{\setone{x}}}{y}
\fun{n_{\boldsymbol{\xi}}^{(m)}}{r,x}}.
\end{aligned}
\label{eq:fock-discrete-bridge-path-space}
\end{equation} The discrete loop path space and its ideal loop measure are defined by \begin{equation}
\begin{aligned}
\Omega_{\Lambda}^{\mathrm{loop},m}
& =
\bigsqcup_{\substack{q \in \semigrposint \\ z \in \Lambda}}
\Omega_{\Lambda;z,z}^{q,m},
\\
\sum_{\zeta \in \Omega_{\Lambda}^{\mathrm{loop},m}}
\fun{F}{\zeta}
\fun{\msrcal{L}_{\Lambda}^{(m)}}{\setone{\zeta}}
& =
\sum_{q \in \semigrposint}\frac{1}{q}
\sum_{z \in \Lambda}
\sum_{\zeta \in \Omega_{\Lambda;z,z}^{q,m}}
\fun{F}{\zeta}
\fun{\msrcal{W}_{\Lambda;z,z}^{q\beta,(m)}}{\setone{\zeta}}
\end{aligned}
\label{eq:fock-discrete-ideal-loop-measure}
\end{equation} for every nonnegative function \(F\) on \(\Omega_{\Lambda}^{\mathrm{loop},m}\). For finite tuples \(\boldsymbol{\omega}\) and \(\boldsymbol{\zeta}\) whose entries belong to the discrete bridge path spaces in \eqref{eq:fock-discrete-bridge-path-space}, the notation \(\fun{\Phi_{\Lambda}^{(m)}}{\boldsymbol{\omega},\boldsymbol{\zeta}}\) means that the occupation numbers of all paths in the concatenated tuple are inserted into the definition of \(\Phi_{\Lambda}^{(m)}\). The discrete open-path and loop sums are \begin{equation}
\begin{aligned}
\fun{Z_{\Lambda}^{(m)}}{\boldsymbol{\omega}}
& =
\sum_{\ell \in \monnat}\frac{1}{\ell!}
\sum_{\boldsymbol{\zeta} \in \rbk{\Omega_{\Lambda}^{\mathrm{loop},m}}^{\ell}}
\napiernum^{-\fun{\Phi_{\Lambda}^{(m)}}{\boldsymbol{\omega},\boldsymbol{\zeta}}}
\prod_{j = 1}^{\ell}
\fun{\msrcal{L}_{\Lambda}^{(m)}}{\setone{\zeta_j}},
\\
Z_{\Lambda}^{(m)}
& =
\sum_{\ell \in \monnat}\frac{1}{\ell!}
\sum_{\boldsymbol{\zeta} \in \rbk{\Omega_{\Lambda}^{\mathrm{loop},m}}^{\ell}}
\napiernum^{-\fun{\Phi_{\Lambda}^{(m)}}{\boldsymbol{\zeta}}}
\prod_{j = 1}^{\ell}
\fun{\msrcal{L}_{\Lambda}^{(m)}}{\setone{\zeta_j}}.
\end{aligned}
\label{eq:fock-discrete-open-path-loop-partition-functions}
\end{equation} Let \(\widetilde{\rho}_{\Lambda,k}^{(m)}\) denote the reduced kernel obtained from \eqref{eq:fock-reduced-density-kernel} by replacing the Gibbs density matrix with the normalized \(m\)-slice Trotter product and its trace with \(Z_{\Lambda}^{(m)}\). For \(M\in\monnat\), \(\boldsymbol{z}\in\Lambda^M\), and \(\widehat{\varpi} \in \grsym{k+M}\), its action on the terminal tuple \(\rbk{\boldsymbol{y},\boldsymbol{z}}\) is defined componentwise by \begin{equation}
\begin{aligned}
\rbk{\fun{\widehat{\varpi}}{\boldsymbol{y},\boldsymbol{z}}}_a
& =
\begin{cases}
y_{\fun{\widehat{\varpi}}{a}}, & 1 \leq \fun{\widehat{\varpi}}{a} \leq k,\\
z_{\fun{\widehat{\varpi}}{a}-k}, & k+1 \leq \fun{\widehat{\varpi}}{a} \leq k+M,
\end{cases}
\quad
1 \leq a \leq k+M.
\end{aligned}
\label{eq:fock-permuted-terminal-tuple}
\end{equation}

\begin{lem}[Creation--annihilation sector formula]\label{lem:fock-creation-annihilation-sector-formula}
Assume \eqref{eq:fock-path-chemical-potential-gap}, let $\Lambda\Subset\Gamma$ be finite, and let $\beta>0$.
For $k,m\in\semigrposint$ and $\boldsymbol{x},\boldsymbol{y}\in\Lambda^k$, the finite-slice reduced kernel is
\begin{equation}
\begin{aligned}
\fun{\widetilde{\rho}_{\Lambda,k}^{(m)}}{\boldsymbol{x};\boldsymbol{y}}
& =
\frac{1}{Z_{\Lambda}^{(m)}}
\sum_{M \in \monnat}\frac{1}{M!}
\sum_{\boldsymbol{z} \in \Lambda^M}
\sum_{\widehat{\varpi} \in \grsym{k+M}}
\sum_{\substack{
\boldsymbol{u}^{(1)},\ldots,\boldsymbol{u}^{(m-1)} \in \Lambda^{k+M}\\
\boldsymbol{u}^{(0)} = \rbk{\boldsymbol{x},\boldsymbol{z}}\\
\boldsymbol{u}^{(m)} = \fun{\widehat{\varpi}}{\boldsymbol{y},\boldsymbol{z}}}}
\fun{B_{k+M,m}}{\boldsymbol{u}^{(0)},\ldots,\boldsymbol{u}^{(m)}}.
\end{aligned}
\label{eq:fock-creation-annihilation-sector-formula}
\end{equation}
\end{lem}

\begin{proof}
If $N<k$, then $a_{x_k}\cdots a_{x_1}$ vanishes on
$P_{\Lambda,N}\sphilb{F}_{\Lambda}$.
Fix $N\geq k$ and $\boldsymbol{x}=\rbk{x_1,\ldots,x_k}$, and set
$$\begin{aligned}
\fun{\kappa_u}{\boldsymbol{x}}
& =
\abscard{\set{j\in\setone{1,\ldots,k}}{x_j=u}},
\quad u\in\Lambda.
\end{aligned}$$
Choose an occupation-number basis vector $\Psi_N=\Psi_{\Lambda,\boldsymbol{\nu}}$ from
\eqref{eq:fock-occupation-number-basis} with $\sum_{u\in\Lambda}\nu_u=N$.
If $\nu_u<\fun{\kappa_u}{\boldsymbol{x}}$ for some $u\in\Lambda$, the product
$a_{x_k}\cdots a_{x_1}$ contains more than $\nu_u$ annihilation operators at $u$.
The occupation-number formula \eqref{eq:fock-occupation-number-basis} gives
$a_u^{\nu_u+1}\Psi_N=0$.
Annihilation operators at distinct sites commute, and the factor $a_u^{\fun{\kappa_u}{\boldsymbol{x}}}$
therefore gives
$$\begin{aligned}
a_{x_k}\cdots a_{x_1}\Psi_N
& =
0.
\end{aligned}$$
It remains to consider the occupation-number basis vectors satisfying
$$\begin{aligned}
\nu_u
& \geq
\fun{\kappa_u}{\boldsymbol{x}},
\quad u\in\Lambda.
\end{aligned}$$
Identify $P_{\Lambda,N}\sphilb{F}_{\Lambda}$ with the symmetric functions in
$\lpseq^2\rbk{\Lambda^N}$.
Let $\psi_N$ be the symmetric coordinate function of $\Psi_N$, and let
$\psi_{N-k}^{\boldsymbol{x}}$ be that of $a_{x_k}\cdots a_{x_1}\Psi_N$.
For $\boldsymbol{z}=\rbk{z_1,\ldots,z_{N-k}}$, the $x_j$ are the annihilated coordinates and the $z_j$ are
the remaining coordinates; the $x_j$ need not occur among the $z_j$.
Definition \eqref{eq:fock-occupation-number-basis} shows that the two values
$\fun{\psi_{N-k}^{\boldsymbol{x}}}{\boldsymbol{z}}$ and
$\fun{\psi_N}{x_1,\ldots,x_k,z_1,\ldots,z_{N-k}}$ vanish unless
$$\begin{aligned}
\fun{\kappa_u}{\boldsymbol{x}}
+
\sum_{j=1}^{N-k}\fun{\fndef{\setone{u}}}{z_j}
& =
\nu_u,
\quad u\in\Lambda.
\end{aligned}$$
For every $\boldsymbol{z}\in\Lambda^{N-k}$, the coordinate representation of the annihilation product is
$$\begin{aligned}
\fun{\psi_{N-k}^{\boldsymbol{x}}}{\boldsymbol{z}}
& =
\sqrt{\frac{N!}{\rbk{N-k}!}}
\fun{\psi_N}{x_1,\ldots,x_k,z_1,\ldots,z_{N-k}}.
\end{aligned}$$
The adjoint string $\faadj{a_{y_1}}\cdots\faadj{a_{y_k}}$ supplies the same square-root coefficient in the
matrix element defining \eqref{eq:fock-reduced-density-kernel}.
Consequently the creation-annihilation strings contribute the factor
$N!/\rbk{N-k}!$.
Multiplication by the coefficient $1/N!$ of the bosonic symmetrizer gives
$$\begin{aligned}
\frac{1}{N!}\frac{N!}{\rbk{N-k}!}
& =
\frac{1}{\rbk{N-k}!}.
\end{aligned}$$
Set $M = N-k$, and let $\boldsymbol{z} = \rbk{z_1,\ldots,z_M} \in \Lambda^M$ denote the remaining particle
coordinates.
For $\widehat{\varpi} \in \grsym{k+M}$, the initial endpoint tuple is
$\rbk{x_1,\ldots,x_k,z_1,\ldots,z_M}$, while the terminal endpoint tuple is the
$\widehat{\varpi}$-permutation of
$\rbk{y_1,\ldots,y_k,z_1,\ldots,z_M}$.
The permutation $\widehat{\varpi}$ acts on the coordinate-index set $\setone{1,\ldots,k+M}$ and permutes the
positions in these tuples.
It does not act on the lattice $\Lambda$.
Its terminal action is \eqref{eq:fock-permuted-terminal-tuple}.
Insertion of the intermediate coordinates in the weight
\eqref{eq:fock-finite-slice-sector-weight} proves
\eqref{eq:fock-creation-annihilation-sector-formula}.
\end{proof}

Separate the coordinate-index set into \[\begin{aligned}
D_k
& =
\setone{1,\ldots,k},
&
U_M
& =
\setone{k+1,\ldots,k+M}.
\end{aligned}\] The set \(D_k\) contains the distinguished coordinate indices, and \(U_M\) contains the remaining coordinate indices.

\begin{lem}[Marked-index permutation decomposition]\label{lem:fock-marked-index-permutation-decomposition}
Let $k \in \semigrposint$, $M \in \monnat$, and $\widehat{\varpi} \in \grsym{k+M}$.
For every $i \in D_k$, define
$$\begin{aligned}
q_i
& =
\min\set{q \in \semigrposint}{\fun{\widehat{\varpi}^{q}}{i} \in D_k},
\\
a_{i,r}
& =
\fun{\widehat{\varpi}^{r}}{i},
\quad 0 \leq r \leq q_i.
\end{aligned}$$
Define the map $\varpi \colon D_k \to D_k$ by
$$\begin{aligned}
\fun{\varpi}{i}
& =
a_{i,q_i},
\quad i \in D_k.
\end{aligned}$$
The map $\varpi$ is a permutation of $D_k$, and
$$\begin{aligned}
a_{i,0}
& =
i,
&
a_{i,r}
& \in
U_M
\quad
\rbk{1 \leq r \leq q_i-1},
&
a_{i,q_i}
& =
\fun{\varpi}{i}
\in
D_k.
\end{aligned}$$
The sequence $\rbk{a_{i,0},\ldots,a_{i,q_i}}$ is the marked chain from $i$ to $\fun{\varpi}{i}$.
The restriction of $\widehat{\varpi}$ to
$U_M \setminus \bigcup_{i=1}^{k}\set{a_{i,r}}{1\leq r\leq q_i-1}$ is a permutation.
There are $\ell \in \monnat$ and points
$c_{j,0} \in U_M \setminus \bigcup_{i=1}^{k}\set{a_{i,r}}{1\leq r\leq q_i-1}$,
$1 \leq j \leq \ell$, such that the sets
$$\begin{aligned}
C_j
& =
\set{\fun{\widehat{\varpi}^{r}}{c_{j,0}}}{r \in \monnat},
\quad 1 \leq j \leq \ell,
\end{aligned}$$
are pairwise disjoint and satisfy
$$\begin{aligned}
U_M \setminus \bigcup_{i=1}^{k}\set{a_{i,r}}{1\leq r\leq q_i-1}
& =
\bigsqcup_{j=1}^{\ell}C_j.
\end{aligned}$$
The sets $C_1,\ldots,C_{\ell}$ are called the unmarked cycles.
For $1 \leq j \leq \ell$ and $0 \leq r \leq \abscard{C_j}$, define
$c_{j,r}=\fun{\widehat{\varpi}^{r}}{c_{j,0}}$.
These definitions give
$$\begin{aligned}
\setone{c_{j,0},\ldots,c_{j,\abscard{C_j}-1}}
& =
C_j,
&
c_{j,\abscard{C_j}}
& =
c_{j,0},
\\
\fun{\widehat{\varpi}}{c_{j,r}}
& =
c_{j,r+1},
&
0 \leq r
& \leq
\abscard{C_j}-1.
\end{aligned}$$
The resulting partition of $U_M$ and its cardinality are
$$\begin{aligned}
U_M
=
\rbk{\bigsqcup_{i=1}^{k}\set{a_{i,r}}{1\leq r\leq q_i-1}}
\sqcup
\rbk{\bigsqcup_{j=1}^{\ell}C_j},
\quad
M
=
\sum_{i = 1}^{k}\rbk{q_i-1}
+
\sum_{j = 1}^{\ell}\abscard{C_j}.
\end{aligned}$$
For fixed $\varpi$, $q_1,\ldots,q_k$, and $C_1,\ldots,C_{\ell}$, reindexing the sum over
$\widehat{\varpi}$ in \eqref{eq:fock-creation-annihilation-sector-formula}, including its coefficient
$1/M!$, produces
\begin{equation}
\begin{aligned}
\frac{1}{\ell!}
\prod_{j = 1}^{\ell}\frac{1}{\abscard{C_j}}.
\end{aligned}
\label{eq:fock-marked-index-permutation-coefficient}
\end{equation}
\end{lem}

\begin{proof}
For $i \in D_k$, finiteness of $\setone{1,\ldots,k+M}$ and bijectivity of $\widehat{\varpi}$ give a
$q \in \semigrposint$ such that $\fun{\widehat{\varpi}^{q}}{i}=i$.
Hence the set defining $q_i$ is nonempty.
For each $a \in \setone{1,\ldots,k+M}$, consider the set
$\set{\fun{\widehat{\varpi}^{r}}{a}}{r \in \monnat}$.
This set either intersects $D_k$ or is contained in $U_M$.
The points in its intersection with $D_k$, listed in the order in which positive powers of
$\widehat{\varpi}$ reach them, divide it into the sequences
$\rbk{a_{i,0},\ldots,a_{i,q_i}}$.
On this list, $\varpi$ sends each point to the next point and sends the last point to the first point.
Thus $\varpi$ is bijective on each intersection and is a permutation of $D_k$.
Each set generated by an index in
$U_M \setminus \bigcup_{i=1}^{k}\set{a_{i,r}}{1\leq r\leq q_i-1}$ is one of the sets $C_j$.
These alternatives prove the partition of $U_M$ in the statement.
For fixed $\varpi$, $q_1,\ldots,q_k$, and $C_1,\ldots,C_{\ell}$, the $M$ indices in $U_M$ can be assigned to
the positions $a_{i,r}$, $1\leq r\leq q_i-1$, and $c_{j,r}$,
$0\leq r\leq\abscard{C_j}-1$, in $M!$ ways.
This factor cancels the coefficient $1/M!$ in
\eqref{eq:fock-creation-annihilation-sector-formula}.
The marked chains have fixed initial indices.  The unmarked cycles can be ordered in $\ell!$ ways, and the
$j$-th cycle has $\abscard{C_j}$ choices of $c_{j,0}$.  Removing these repetitions leaves
$$\begin{aligned}
\frac{1}{\ell!}
\prod_{j = 1}^{\ell}\frac{1}{\abscard{C_j}}.
\end{aligned}$$
The factor $1/\abscard{C_j}$ agrees with the length-$\abscard{C_j}$ factor in the discrete loop measure
$\msrcal{L}_{\Lambda}^{(m)}$ defined by \eqref{eq:fock-discrete-ideal-loop-measure}.
\end{proof}

\begin{lem}[Finite-slice open-path expansion]\label{lem:fock-finite-slice-open-path-expansion}
Assume \eqref{eq:fock-path-chemical-potential-gap}, let $\Lambda\Subset\Gamma$ be finite, and let $\beta>0$.
For every $m\in\semigrposint$, the finite-slice loop series is the trace normalization:
\begin{equation}
\begin{aligned}
Z_{\Lambda}^{(m)}
& =
\sum_{N\in\monnat}
\sqfun{\trace_{P_{\Lambda,N}\sphilb{F}_{\Lambda}}}{T_{\Lambda,m}}.
\end{aligned}
\label{eq:fock-finite-slice-partition-trace-identity}
\end{equation}
For every $k,m\in\semigrposint$ and $\boldsymbol{x},\boldsymbol{y}\in\Lambda^k$, the finite-slice reduced kernel
satisfies
\begin{equation}
\begin{aligned}
\fun{\widetilde{\rho}_{\Lambda,k}^{(m)}}{\boldsymbol{x};\boldsymbol{y}}
& =
\sum_{\varpi \in \grsym{k}}
\sum_{q_1,\ldots,q_k \in \semigrposint}
\sum_{\substack{
\omega_i \in \Omega_{\Lambda;x_i,y_{\fun{\varpi}{i}}}^{q_i,m}\\
1 \leq i \leq k}}
\frac{\fun{Z_{\Lambda}^{(m)}}{\boldsymbol{\omega}}}{Z_{\Lambda}^{(m)}}
\prod_{i = 1}^{k}
\fun{\msrcal{W}_{\Lambda;x_i,y_{\fun{\varpi}{i}}}^{q_i\beta,(m)}}{\setone{\omega_i}}.
\end{aligned}
\label{eq:fock-finite-slice-open-path-expansion}
\end{equation}
\end{lem}

\begin{proof}
Lemma \ref{lem:fock-creation-annihilation-sector-formula} supplies the sector formula.
Apply Lemma \ref{lem:fock-marked-index-permutation-decomposition} to each permutation in that formula.
The reindexed coefficient is \eqref{eq:fock-marked-index-permutation-coefficient}.
For each coordinate index $a$, the initial and terminal sites are
$$\begin{aligned}
u_a^{(0)}
& =
\begin{cases}
x_a, & a \in D_k,\\
z_{a-k}, & a \in U_M,
\end{cases}
&
u_a^{(m)}
& =
\begin{cases}
y_{\fun{\widehat{\varpi}}{a}}, & \fun{\widehat{\varpi}}{a} \in D_k,\\
z_{\fun{\widehat{\varpi}}{a}-k}, & \fun{\widehat{\varpi}}{a} \in U_M.
\end{cases}
\end{aligned}$$
Concatenating the coordinate sequences along the marked chains and unmarked cycles defines
$$\begin{aligned}
\omega_{i,rm+s}
& =
u_{a_{i,r}}^{(s)},
\quad
0 \leq r \leq q_i-1,
\quad
0 \leq s \leq m,
\\
\zeta_{j,rm+s}
& =
u_{c_{j,r}}^{(s)},
\quad
0 \leq r \leq \abscard{C_j}-1,
\quad
0 \leq s \leq m.
\end{aligned}$$
The endpoint equalities used to concatenate the coordinate sequences are
$$\begin{aligned}
u_{a_{i,r}}^{(m)}
& =
z_{a_{i,r+1}-k}
=
u_{a_{i,r+1}}^{(0)},
\quad
0 \leq r \leq q_i-2,
\\
u_{c_{j,r}}^{(m)}
& =
u_{c_{j,r+1}}^{(0)},
\quad
0 \leq r \leq \abscard{C_j}-2,
\\
u_{c_{j,\abscard{C_j}-1}}^{(m)}
& =
u_{c_{j,0}}^{(0)}.
\end{aligned}$$
The endpoints of the path $\omega_i$ are
$$\begin{aligned}
\omega_{i,0}
& =
x_i,
&
\omega_{i,q_i m}
& =
y_{\fun{\varpi}{i}}.
\end{aligned}$$
In particular, the $i$-th marked path has $q_i$ blocks of $m$ time slices, while the $j$-th unmarked path is a
loop with $\abscard{C_j}m$ time slices.

The one-step part of $B_{k+M,m}$ factorizes as
$$\begin{aligned}
&
\prod_{s = 0}^{m-1}
\prod_{a = 1}^{k+M}
\fun{b_m}{u_a^{(s)},u_a^{(s+1)}}
=
\prod_{i = 1}^{k}
\prod_{t = 0}^{q_i m-1}
\fun{b_m}{\omega_{i,t},\omega_{i,t+1}}
\prod_{j = 1}^{\ell}
\prod_{t = 0}^{\abscard{C_j}m-1}
\fun{b_m}{\zeta_{j,t},\zeta_{j,t+1}}
\\ 
& =
\prod_{i = 1}^{k}
\fun{\msrcal{W}_{\Lambda;x_i,y_{\fun{\varpi}{i}}}^{q_i\beta,(m)}}{\setone{\omega_i}}
\prod_{j = 1}^{\ell}
\fun{\msrcal{W}_{\Lambda;z_{c_{j,0}-k},z_{c_{j,0}-k}}^{\abscard{C_j}\beta,(m)}}{\setone{\zeta_j}}.
\end{aligned}$$
\begin{samepage}
At the $s$-th Trotter slice, the same partition of the coordinate indices gives the occupation identity
$$\begin{aligned}
\nu_x^{(s)}
& =
\sum_{a = 1}^{k+M}
\fun{\fndef{\setone{x}}}{u_a^{(s)}}
=
\sum_{i = 1}^{k}
\sum_{r = 0}^{q_i-1}
\fun{\fndef{\setone{x}}}{\omega_{i,rm+s}}
+
\sum_{j = 1}^{\ell}
\sum_{r = 0}^{\abscard{C_j}-1}
\fun{\fndef{\setone{x}}}{\zeta_{j,rm+s}}
\\
& =
\fun{n_{\boldsymbol{\omega}}^{(m)}}{s,x}
+
\fun{n_{\boldsymbol{\zeta}}^{(m)}}{s,x}.
\end{aligned}$$
\end{samepage}
Substitution of this identity into the occupation-energy factor of $B_{k+M,m}$ gives
$$\begin{aligned}
\prod_{s = 0}^{m-1}
\napiernum^{-\frac{\sminvtemperature}{m}\fun{E_{\Lambda}}{\nu^{(s)}}}
& =
\napiernum^{-\fun{\Phi_{\Lambda}^{(m)}}{\boldsymbol{\omega},\boldsymbol{\zeta}}}.
\end{aligned}$$
Combining the one-step and occupation-energy identities gives the complete factorization
$$\begin{aligned}
\fun{B_{k+M,m}}{\boldsymbol{u}^{(0)},\ldots,\boldsymbol{u}^{(m)}}
& =
\napiernum^{-\fun{\Phi_{\Lambda}^{(m)}}{\boldsymbol{\omega},\boldsymbol{\zeta}}}
\prod_{i = 1}^{k}
\fun{\msrcal{W}_{\Lambda;x_i,y_{\fun{\varpi}{i}}}^{q_i\beta,(m)}}{\setone{\omega_i}}
\prod_{j = 1}^{\ell}
\fun{\msrcal{W}_{\Lambda;z_{c_{j,0}-k},z_{c_{j,0}-k}}^{\abscard{C_j}\beta,(m)}}{\setone{\zeta_j}}.
\end{aligned}$$
For $k=0$, the sets generated by repeated application of the coordinate-index permutation are precisely the
unmarked cycles $C_1,\ldots,C_{\ell}$.
Ordering these sets and choosing one of the $\abscard{C_j}$ indices as $c_{j,0}$ produces the factor
$\ell!^{-1}\prod_{j=1}^{\ell}\abscard{C_j}^{-1}$.
This coefficient and the factorization of $B_{M,m}$ convert the sum of the sector traces in
\eqref{eq:fock-finite-slice-trace-formula} into the loop series
\eqref{eq:fock-discrete-open-path-loop-partition-functions}.
This proves \eqref{eq:fock-finite-slice-partition-trace-identity}.
The coefficient
$\frac{1}{\ell!}\prod_{j = 1}^{\ell}\frac{1}{\abscard{C_j}}$ and the path factorization rewrite the sector sum as
$$\begin{aligned}
\fun{\widetilde{\rho}_{\Lambda,k}^{(m)}}{\boldsymbol{x};\boldsymbol{y}}
& =
\frac{1}{Z_{\Lambda}^{(m)}}
\sum_{\varpi \in \grsym{k}}
\sum_{q_1,\ldots,q_k \in \semigrposint}
\sum_{\substack{
\omega_i \in \Omega_{\Lambda;x_i,y_{\fun{\varpi}{i}}}^{q_i,m}\\
1 \leq i \leq k}}
\prod_{i = 1}^{k}
\fun{\msrcal{W}_{\Lambda;x_i,y_{\fun{\varpi}{i}}}^{q_i\beta,(m)}}{\setone{\omega_i}}
\\
& \mathrel{\phantom{=}}
\times
\sum_{\ell \in \monnat}\frac{1}{\ell!}
\sum_{\boldsymbol{\zeta}\in\rbk{\Omega_{\Lambda}^{\mathrm{loop},m}}^{\ell}}
\napiernum^{-\fun{\Phi_{\Lambda}^{(m)}}{\boldsymbol{\omega},\boldsymbol{\zeta}}}
\prod_{j = 1}^{\ell}
\fun{\msrcal{L}_{\Lambda}^{(m)}}{\setone{\zeta_j}}.
\end{aligned}$$
The final two sums and the interaction factor are precisely
$\fun{Z_{\Lambda}^{(m)}}{\boldsymbol{\omega}}$ by
\eqref{eq:fock-discrete-open-path-loop-partition-functions}.
Substitution of this definition proves \eqref{eq:fock-finite-slice-open-path-expansion}.
The derivation identifies the marked endpoint condition
$\omega_{i,0} = x_i$ and $\omega_{i,q_i m} = y_{\fun{\varpi}{i}}$.
The unmarked cycles remain in the normalized loop factor.
\end{proof}

\begin{lem}[Finite-slice domination]\label{lem:fock-finite-slice-domination}
Under the assumptions of Lemma \ref{lem:fock-finite-slice-open-path-expansion}, suppose that $v$ is pointwise
nonnegative.
Every tuple $\boldsymbol{\omega}$ occurring in \eqref{eq:fock-finite-slice-open-path-expansion} satisfies
\begin{equation}
\begin{aligned}
0
\leq
\frac{\fun{Z_{\Lambda}^{(m)}}{\boldsymbol{\omega}}}{Z_{\Lambda}^{(m)}}
\leq
\napiernum^{-\fun{\Phi_{\Lambda}^{(m)}}{\boldsymbol{\omega}}}
\leq
1.
\end{aligned}
\label{eq:fock-finite-slice-normalized-loop-bound}
\end{equation}
The total loop mass, the finite-slice partition function, and the reduced kernel obey
\begin{equation}
\begin{aligned}
\sum_{\zeta\in\Omega_{\Lambda}^{\mathrm{loop},m}}
\fun{\msrcal{L}_{\Lambda}^{(m)}}{\setone{\zeta}}
& \leq
-\abscard{\Lambda}
\log\rbk{1-\napiernum^{-\beta\varepsilon}},
\\
Z_{\Lambda}^{(m)}
& \leq
\frac{1}{\rbk{1-\napiernum^{-\beta\varepsilon}}^{\abscard{\Lambda}}},
\\
0
\leq
\fun{\widetilde{\rho}_{\Lambda,k}^{(m)}}{\boldsymbol{x};\boldsymbol{y}}
& \leq
\frac{k!}{\rbk{\napiernum^{\beta\varepsilon}-1}^{k}}.
\end{aligned}
\label{eq:fock-finite-slice-uniform-bounds}
\end{equation}
\end{lem}

\begin{proof}
For the concatenated tuple $\rbk{\boldsymbol{\omega},\boldsymbol{\zeta}}$, the discrete occupation numbers
satisfy
$$\begin{aligned}
\fun{n_{\rbk{\boldsymbol{\omega},\boldsymbol{\zeta}}}^{(m)}}{r,x}
& =
\fun{n_{\boldsymbol{\omega}}^{(m)}}{r,x}
+
\fun{n_{\boldsymbol{\zeta}}^{(m)}}{r,x}.
\end{aligned}$$
Substitution into the quadratic and linear terms of $\Phi_{\Lambda}^{(m)}$ gives
$$\begin{aligned}
&
\fun{n_{\rbk{\boldsymbol{\omega},\boldsymbol{\zeta}}}^{(m)}}{r,x}
\fun{n_{\rbk{\boldsymbol{\omega},\boldsymbol{\zeta}}}^{(m)}}{r,y}
\\ 
& =
\fun{n_{\boldsymbol{\omega}}^{(m)}}{r,x}
\fun{n_{\boldsymbol{\omega}}^{(m)}}{r,y}
+
\fun{n_{\boldsymbol{\zeta}}^{(m)}}{r,x}
\fun{n_{\boldsymbol{\zeta}}^{(m)}}{r,y}
+
\fun{n_{\boldsymbol{\omega}}^{(m)}}{r,x}
\fun{n_{\boldsymbol{\zeta}}^{(m)}}{r,y}
+
\fun{n_{\boldsymbol{\zeta}}^{(m)}}{r,x}
\fun{n_{\boldsymbol{\omega}}^{(m)}}{r,y}.
\end{aligned}$$
The symmetry and pointwise nonnegativity of $v$ identify the two mixed quadratic terms and give
$$\begin{aligned}
\fun{I_{\Lambda}^{(m)}}{\boldsymbol{\omega},\boldsymbol{\zeta}}
& =
\frac{2\sminvtemperature}{m}
\sum_{r = 0}^{m - 1}
\sum_{x,y \in \Lambda}
\fun{v}{x,y}
\fun{n_{\boldsymbol{\omega}}^{(m)}}{r,x}
\fun{n_{\boldsymbol{\zeta}}^{(m)}}{r,y}
\geq
0,
\\
\fun{\Phi_{\Lambda}^{(m)}}{\boldsymbol{\omega},\boldsymbol{\zeta}}
& =
\fun{\Phi_{\Lambda}^{(m)}}{\boldsymbol{\omega}}
+
\fun{\Phi_{\Lambda}^{(m)}}{\boldsymbol{\zeta}}
+
\fun{I_{\Lambda}^{(m)}}{\boldsymbol{\omega},\boldsymbol{\zeta}}
\geq
\fun{\Phi_{\Lambda}^{(m)}}{\boldsymbol{\omega}}
+
\fun{\Phi_{\Lambda}^{(m)}}{\boldsymbol{\zeta}}.
\end{aligned}$$
For every tuple $\boldsymbol{\xi}$, the discrete self-energy is nonnegative because
$$\begin{aligned}
&
\sum_{x,y \in \Lambda}
\fun{v}{x,y}
\rbk{
\fun{n_{\boldsymbol{\xi}}^{(m)}}{r,x}
\fun{n_{\boldsymbol{\xi}}^{(m)}}{r,y}
-
\fun{\fndef{\setone{x}}}{y}
\fun{n_{\boldsymbol{\xi}}^{(m)}}{r,x}}
\\
& =
\sum_{x \in \Lambda}
\fun{v}{x,x}
\fun{n_{\boldsymbol{\xi}}^{(m)}}{r,x}
\rbk{\fun{n_{\boldsymbol{\xi}}^{(m)}}{r,x}-1}
+
\sum_{\substack{x,y \in \Lambda\\x\neq y}}
\fun{v}{x,y}
\fun{n_{\boldsymbol{\xi}}^{(m)}}{r,x}
\fun{n_{\boldsymbol{\xi}}^{(m)}}{r,y}
\geq
0.
\end{aligned}$$
The discrete loop series in
\eqref{eq:fock-discrete-open-path-loop-partition-functions} satisfies
$$\begin{aligned}
&
\fun{Z_{\Lambda}^{(m)}}{\boldsymbol{\omega}}
=
\sum_{\ell \in \monnat}\frac{1}{\ell!}
\sum_{\boldsymbol{\zeta}\in\rbk{\Omega_{\Lambda}^{\mathrm{loop},m}}^{\ell}}
\napiernum^{-\fun{\Phi_{\Lambda}^{(m)}}{\boldsymbol{\omega},\boldsymbol{\zeta}}}
\prod_{j = 1}^{\ell}
\fun{\msrcal{L}_{\Lambda}^{(m)}}{\setone{\zeta_j}}
\\
& \leq
\napiernum^{-\fun{\Phi_{\Lambda}^{(m)}}{\boldsymbol{\omega}}}
\sum_{\ell \in \monnat}\frac{1}{\ell!}
\sum_{\boldsymbol{\zeta}\in\rbk{\Omega_{\Lambda}^{\mathrm{loop},m}}^{\ell}}
\napiernum^{-\fun{\Phi_{\Lambda}^{(m)}}{\boldsymbol{\zeta}}}
\prod_{j = 1}^{\ell}
\fun{\msrcal{L}_{\Lambda}^{(m)}}{\setone{\zeta_j}}
=
\napiernum^{-\fun{\Phi_{\Lambda}^{(m)}}{\boldsymbol{\omega}}}
Z_{\Lambda}^{(m)}
\leq
Z_{\Lambda}^{(m)}.
\end{aligned}$$
This proves \eqref{eq:fock-finite-slice-normalized-loop-bound}.

The total mass of the discrete ideal loop measure can be computed from
\eqref{eq:fock-discrete-ideal-loop-measure}.
The semigroup property and the row estimate \eqref{eq:fock-killed-walk-row-bound} give
$$\begin{aligned}
&
\sum_{\zeta\in\Omega_{\Lambda}^{\mathrm{loop},m}}
\fun{\msrcal{L}_{\Lambda}^{(m)}}{\setone{\zeta}}
=
\sum_{q \in \semigrposint}\frac{1}{q}
\sum_{z \in \Lambda}
\sum_{\zeta\in\Omega_{\Lambda;z,z}^{q,m}}
\prod_{r = 0}^{qm-1}
\fun{b_m}{\zeta_r,\zeta_{r+1}}
=
\sum_{q \in \semigrposint}\frac{1}{q}
\sum_{z \in \Lambda}
\fun{b_m^{qm}}{z,z}
\\
& =
\sum_{q \in \semigrposint}\frac{1}{q}
\sum_{z \in \Lambda}
\fun{p_{\Lambda,q\beta}}{z,z}
\leq
\sum_{q \in \semigrposint}\frac{1}{q}
\sum_{z,w \in \Lambda}
\fun{p_{\Lambda,q\beta}}{z,w}
\leq
\abscard{\Lambda}
\sum_{q \in \semigrposint}
\frac{\napiernum^{-q\beta\varepsilon}}{q}
=
-\abscard{\Lambda}
\log\rbk{1-\napiernum^{-\beta\varepsilon}}.
\end{aligned}$$
The equality $b_m^{qm}=p_{\Lambda,q\beta}$ follows directly from the definitions of $b_m$ and
$p_{\Lambda,t}$.
Dropping the nonnegative interaction energy and summing the exponential series gives
$$\begin{aligned}
&Z_{\Lambda}^{(m)}
\leq
\sum_{\ell \in \monnat}\frac{1}{\ell!}
\rbk{
\sum_{\zeta\in\Omega_{\Lambda}^{\mathrm{loop},m}}
\fun{\msrcal{L}_{\Lambda}^{(m)}}{\setone{\zeta}}}^{\ell}
=
\fnexp{
\sum_{\zeta\in\Omega_{\Lambda}^{\mathrm{loop},m}}
\fun{\msrcal{L}_{\Lambda}^{(m)}}{\setone{\zeta}}}
\\ 
& \leq
\fnexp{
-\abscard{\Lambda}
\fun{\log}{1 - \napiernum^{-\sminvtemperature\omega \varepsilon}}}
=
\rbk{1-\napiernum^{-\beta\varepsilon}}^{-\abscard{\Lambda}}.
\end{aligned}$$

We apply the normalized-loop estimate to the finite-slice reduced kernel.
Positivity of the bridge weights gives
$$\begin{aligned}
&
0
\leq
\fun{\widetilde{\rho}_{\Lambda,k}^{(m)}}{\boldsymbol{x};\boldsymbol{y}}
\leq
\sum_{\varpi \in \grsym{k}}
\sum_{q_1,\ldots,q_k \in \semigrposint}
\sum_{\substack{
\omega_i \in \Omega_{\Lambda;x_i,y_{\fun{\varpi}{i}}}^{q_i,m}\\
1 \leq i \leq k}}
\prod_{i = 1}^{k}
\fun{\msrcal{W}_{\Lambda;x_i,y_{\fun{\varpi}{i}}}^{q_i\beta,(m)}}{\setone{\omega_i}}
\\ 
& =
\sum_{\varpi \in \grsym{k}}
\prod_{i = 1}^{k}
\rbk{
\sum_{q_i \in \semigrposint}
\sum_{\omega_i \in \Omega_{\Lambda;x_i,y_{\fun{\varpi}{i}}}^{q_i,m}}
\fun{\msrcal{W}_{\Lambda;x_i,y_{\fun{\varpi}{i}}}^{q_i\beta,(m)}}{\setone{\omega_i}}}
\\ 
& =
\sum_{\varpi \in \grsym{k}}
\prod_{i = 1}^{k}
\rbk{
\sum_{q_i \in \semigrposint}
\fun{b_m^{q_i m}}{x_i,y_{\fun{\varpi}{i}}}}
=
\sum_{\varpi \in \grsym{k}}
\prod_{i = 1}^{k}
\rbk{
\sum_{q_i \in \semigrposint}
\fun{p_{\Lambda,q_i\beta}}{x_i,y_{\fun{\varpi}{i}}}}
\\
& \leq
\sum_{\varpi \in \grsym{k}}
\prod_{i = 1}^{k}
\rbk{
\sum_{q_i \in \semigrposint}
\napiernum^{-q_i\beta\varepsilon}}
=
k!
\rbk{
\sum_{q \in \semigrposint}
\napiernum^{-q\beta\varepsilon}}^{k}
=
\frac{k!}{\rbk{\napiernum^{\beta\varepsilon}-1}^{k}}.
\end{aligned}$$
The marked-chain estimate is independent of $m$, the volume, and the endpoint sites.
The loop-mass calculation, the exponential-series estimate, and the marked-chain estimate prove
\eqref{eq:fock-finite-slice-uniform-bounds}.
\end{proof}

The passage from finite-slice paths to continuous time is isolated in the next lemma.

\begin{lem}[Continuous-time limit]\label{lem:fock-continuous-time-limit}
Under the assumptions of Lemma \ref{lem:fock-finite-slice-domination}, the discrete loop partition function
converges to the loop partition function in \eqref{eq:fock-open-path-loop-partition-functions}, and
\begin{equation}
\begin{aligned}
\lim_{m\to\infty}
\fun{\widetilde{\rho}_{\Lambda,k}^{(m)}}{\boldsymbol{x};\boldsymbol{y}}
& =
\sum_{\varpi \in \grsym{k}}
\sum_{q_1,\ldots,q_k \in \semigrposint}
\int_{\prod_{i = 1}^{k}\Omega_{\Lambda;x_i,y_{\fun{\varpi}{i}}}^{q_i\beta}}
\frac{\fun{Z_{\Lambda}}{\boldsymbol{\omega}}}{Z_{\Lambda}}
\prod_{i = 1}^{k}
\opdmsr{\msrcal{W}_{\Lambda;x_i,y_{\fun{\varpi}{i}}}^{q_i\beta}}(\omega_i).
\end{aligned}
\label{eq:fock-continuous-time-open-path-limit}
\end{equation}
\end{lem}

\begin{proof}
For a path $\omega\in\Omega_{\Lambda;x,y}^{q\beta}$, let its $m$-skeleton be
$\rbk{\omega_0,\ldots,\omega_{qm}}$ with $\omega_r=\omega_{r\beta/m}$.
The joint distribution of the path values at the times $r\beta/m$, $0\leq r\leq qm$, equals the finite-slice
measure because
$$\begin{aligned}
&
\fun{\msrcal{W}_{\Lambda;x,y}^{q\beta}}{
\set{\omega}{\omega_{r\beta/m}=\omega_r,\quad 0\leq r\leq qm}}
\\
& =
\prod_{r = 0}^{qm-1}
\fun{p_{\Lambda,\beta/m}}{\omega_r,\omega_{r+1}}
=
\prod_{r = 0}^{qm-1}
\fun{b_m}{\omega_r,\omega_{r+1}}
=
\fun{\msrcal{W}_{\Lambda;x,y}^{q\beta,(m)}}{
\setone{\rbk{\omega_0,\ldots,\omega_{qm}}}}.
\end{aligned}$$
For $\msrcal{W}_{\Lambda;x,y}^{q\beta}$-almost every path, the number of jumps in
$\closedinterval{0}{q\beta}$ is finite.
The left Riemann sums of the piecewise-constant occupation functions defined in
\eqref{eq:fock-path-tuple-occupation-field} satisfy
$$\begin{aligned}
&
\lim_{m \to \infty}
\frac{\sminvtemperature}{m}
\sum_{r = 0}^{m-1}
\sum_{x,y\in\Lambda}
\fun{v}{x,y}
\rbk{
\fun{n_{\boldsymbol{\xi}}}{\frac{r\beta}{m},x}
\fun{n_{\boldsymbol{\xi}}}{\frac{r\beta}{m},y}
-
\fun{\fndef{\setone{x}}}{y}
\fun{n_{\boldsymbol{\xi}}}{\frac{r\beta}{m},x}}
\\ 
& =
\int_{\closedinterval{0}{\beta}}
\sum_{x,y\in\Lambda}
\fun{v}{x,y}
\rbk{
\fun{n_{\boldsymbol{\xi}}}{s,x}
\fun{n_{\boldsymbol{\xi}}}{s,y}
-
\fun{\fndef{\setone{x}}}{y}
\fun{n_{\boldsymbol{\xi}}}{s,x}}
\opdmsr{s}
=
\fun{\Phi_{\Lambda}}{\boldsymbol{\xi}}.
\end{aligned}$$
Fix $s\in\semigrposint$, endpoints $x_a,y_a\in\Lambda$, and lengths
$q_a\in\semigrposint$ for $1\leq a\leq s$.
Since $0 \leq \napiernum^{-\Phi_{\Lambda}^{(m)}} \leq 1$,
dominated convergence for this fixed path tuple gives
\begin{equation}
\begin{aligned}
&
\lim_{m \to \infty}
\sum_{\substack{
\xi_a^{(m)}\in\Omega_{\Lambda;x_a,y_a}^{q_a,m}\\
1\leq a\leq s}}
\napiernum^{-\fun{\Phi_{\Lambda}^{(m)}}{\boldsymbol{\xi}^{(m)}}}
\prod_{a=1}^{s}
\fun{\msrcal{W}_{\Lambda;x_a,y_a}^{q_a\beta,(m)}}{
\setone{\xi_a^{(m)}}}
\\ 
& =
\int_{\prod_{a=1}^{s}\Omega_{\Lambda;x_a,y_a}^{q_a\beta}}
\napiernum^{-\fun{\Phi_{\Lambda}}{\boldsymbol{\xi}}}
\prod_{a=1}^{s}
\opdmsr{\msrcal{W}_{\Lambda;x_a,y_a}^{q_a\beta}}(\xi_a).
\end{aligned}
\label{eq:fock-fixed-tuple-skeleton-limit}
\end{equation}
The discrete sum ranges over skeletons with the prescribed endpoints.

The loop-count, loop-length, and marked-length tails are uniform in $m$.
The loop-count tail satisfies
$$\begin{aligned}
0
& \leq
\sum_{\ell=L+1}^{\infty}\frac{1}{\ell!}
\sum_{\boldsymbol{\zeta}\in\rbk{\Omega_{\Lambda}^{\mathrm{loop},m}}^{\ell}}
\napiernum^{-\fun{\Phi_{\Lambda}^{(m)}}{\boldsymbol{\zeta}}}
\prod_{j=1}^{\ell}
\fun{\msrcal{L}_{\Lambda}^{(m)}}{\setone{\zeta_j}}
\leq
\sum_{\ell=L+1}^{\infty}\frac{1}{\ell!}
\rbk{
-\abscard{\Lambda}
\log\rbk{1-\napiernum^{-\beta\varepsilon}}}^{\ell}
\xrightarrow{L \to \infty}
0.
\end{aligned}$$
The loop-length tail satisfies
$$\begin{aligned}
0
& \leq
\sum_{q=Q+1}^{\infty}\frac{1}{q}
\sum_{z\in\Lambda}
\sum_{\zeta\in\Omega_{\Lambda;z,z}^{q,m}}
\fun{\msrcal{W}_{\Lambda;z,z}^{q\beta,(m)}}{\setone{\zeta}}
=
\sum_{q=Q+1}^{\infty}\frac{1}{q}
\sum_{z\in\Lambda}
\fun{p_{\Lambda,q\beta}}{z,z}
\leq
\abscard{\Lambda}
\sum_{q=Q+1}^{\infty}
\frac{\napiernum^{-q\beta\varepsilon}}{q}
\xrightarrow{Q \to \infty}
0.
\end{aligned}$$
For loop tuples, summing over the position of a loop whose length exceeds $Q\beta$ gives
$$\begin{aligned}
&
\sum_{\ell=1}^{\infty}\frac{1}{\ell!}
\sum_{\substack{
\boldsymbol{\zeta}\in\rbk{\Omega_{\Lambda}^{\mathrm{loop},m}}^{\ell}\\
\zeta_j\in\bigsqcup_{\substack{q\in\semigrposint,\ q>Q\\z\in\Lambda}}\Omega_{\Lambda;z,z}^{q,m}
\text{ for some }1\leq j\leq\ell}}
\prod_{j=1}^{\ell}
\fun{\msrcal{L}_{\Lambda}^{(m)}}{\setone{\zeta_j}}
\\
& \leq
\rbk{
\abscard{\Lambda}
\sum_{q=Q+1}^{\infty}
\frac{\napiernum^{-q\beta\varepsilon}}{q}}
\fnexp{
\abscard{\Lambda}
\sum_{q\in\semigrposint}
\frac{\napiernum^{-q\beta\varepsilon}}{q}}
\xrightarrow{Q \to \infty}
0.
\end{aligned}$$
For each marked path, the length tail satisfies
$$\begin{aligned}
0
& \leq
\sum_{q=Q+1}^{\infty}
\sum_{\omega\in\Omega_{\Lambda;x,y}^{q,m}}
\fun{\msrcal{W}_{\Lambda;x,y}^{q\beta,(m)}}{\setone{\omega}}
=
\sum_{q=Q+1}^{\infty}
\fun{p_{\Lambda,q\beta}}{x,y}
\leq
\sum_{q=Q+1}^{\infty}\napiernum^{-q\beta\varepsilon}
\xrightarrow{Q \to \infty}
0.
\end{aligned}$$
The sum over marked lengths with at least one $q_i>Q$ is bounded uniformly in the endpoints by
$$\begin{aligned}
&
\sum_{\varpi\in\grsym{k}}
\sum_{\substack{q_1,\ldots,q_k\in\semigrposint\\q_i>Q\text{ for some }1\leq i\leq k}}
\prod_{i=1}^{k}
\fun{p_{\Lambda,q_i\beta}}{x_i,y_{\fun{\varpi}{i}}}
\leq
k!k
\rbk{
\sum_{q=Q+1}^{\infty}\napiernum^{-q\beta\varepsilon}}
\rbk{
\sum_{q\in\semigrposint}\napiernum^{-q\beta\varepsilon}}^{k-1}
\xrightarrow{Q \to \infty}
0.
\end{aligned}$$
The fixed-tuple convergence, the three tail estimates, and the finiteness of $\grsym{k}$ permit the limit
$m\to\infty$ to pass through the loop-count, loop-length, and marked-length sums in
\eqref{eq:fock-finite-slice-open-path-expansion}.
Applying the same argument with no marked paths gives $Z_{\Lambda}^{(m)}\to Z_{\Lambda}$.
Applying it with the $k$ marked paths gives \eqref{eq:fock-continuous-time-open-path-limit}.
\end{proof}

\begin{lem}[Continuous bosonic cycle decomposition]\label{lem:fock-continuous-bosonic-cycle-decomposition}
Assume \eqref{eq:fock-path-chemical-potential-gap}, suppose that $v$ is pointwise nonnegative, let
$\Lambda\Subset\Gamma$ be finite, and let $\beta>0$.
The operator partition function in \eqref{eq:fock-finite-volume-gibbs-state} equals the loop partition function:
\begin{equation}
\begin{aligned}
Z_{\Lambda,\beta}
& =
Z_{\Lambda}.
\end{aligned}
\label{eq:fock-continuous-cycle-partition-identity}
\end{equation}
For every $k\in\semigrposint$ and $\boldsymbol{x},\boldsymbol{y}\in\Lambda^k$, the reduced kernel satisfies
\begin{equation}
\begin{aligned}
\fun{\widetilde{\rho}_{\Lambda,k}}{\boldsymbol{x};\boldsymbol{y}}
& =
\sum_{\varpi \in \grsym{k}}
\sum_{q_1,\ldots,q_k \in \semigrposint}
\int_{\prod_{i = 1}^{k}\Omega_{\Lambda;x_i,y_{\fun{\varpi}{i}}}^{q_i\beta}}
\frac{\fun{Z_{\Lambda}}{\boldsymbol{\omega}}}{Z_{\Lambda}}
\prod_{i = 1}^{k}
\opdmsr{\msrcal{W}_{\Lambda;x_i,y_{\fun{\varpi}{i}}}^{q_i\beta}}(\omega_i).
\end{aligned}
\label{eq:fock-continuous-cycle-reduced-kernel}
\end{equation}
\end{lem}

\begin{proof}
Apply the fixed-tuple limit \eqref{eq:fock-fixed-tuple-skeleton-limit} to the finite-slice trace formula
\eqref{eq:fock-finite-slice-trace-formula}.
The trace majorant in Proposition \ref{prop:fock-finite-volume-gibbs-existence} permits summation over the
particle sectors and gives
\begin{equation}
\begin{aligned}
Z_{\Lambda,\beta}
& =
\sum_{N \in \monnat}\frac{1}{N!}
\sum_{\varpi \in \grsym{N}}
\sum_{\boldsymbol{x} \in \Lambda^N}
\int_{\prod_{i = 1}^{N}\Omega_{\Lambda;x_i,x_{\fun{\varpi}{i}}}^{\beta}}
\napiernum^{-\fun{\Phi_{\Lambda}}{\omega_1,\ldots,\omega_N}}
\prod_{i = 1}^{N}
\opdmsr{\msrcal{W}_{\Lambda;x_i,x_{\fun{\varpi}{i}}}^{\beta}}(\omega_i).
\end{aligned}
\label{eq:fock-precycle-partition-function}
\end{equation}
Let $n_q$ be the number of cycles of length $q$ in $\varpi$.
The cycle lengths exhaust all particle-coordinate indices and satisfy
\begin{equation}
\begin{aligned}
\sum_{q \in \semigrposint}q n_q
& =
N,
\end{aligned}
\label{eq:fock-bosonic-cycle-length-constraint}
\end{equation}
The number of permutations with these cycle counts is obtained by arranging the $N$ particle-coordinate indices
in a row.
There are $N!$ arrangements.
Each of the $n_q$ cycles of length $q$ has $q$ cyclic rotations representing the same cycle, and the cycles
of the same length have $n_q!$ orderings.
Division by these repetitions gives
\begin{equation}
\begin{aligned}
\abscard{\setone{\varpi \in \grsym{N}:\varpi\text{ has $n_q$ cycles of length $q$}}}
& =
\frac{N!}{\prod_{q \in \semigrposint}q^{n_q} n_q!},
\\
\frac{1}{N!}
\frac{N!}{\prod_{q \in \semigrposint}q^{n_q} n_q!}
& =
\prod_{q \in \semigrposint}
\frac{1}{n_q!}\rbk{\frac{1}{q}}^{n_q}.
\end{aligned}
\label{eq:fock-bosonic-cycle-count}
\end{equation}

We calculate how a single cycle becomes one loop.
Fix a cycle $\rbk{i_1\,i_2\,\cdots\,i_q}$ and use the convention $i_{q+1}=i_1$.
For paths $\omega_{i_r}\in\Omega_{\Lambda;x_{i_r},x_{i_{r+1}}}^{\beta}$, define their concatenated loop
$\zeta\in\Omega_{\Lambda;x_{i_1},x_{i_1}}^{q\beta}$ by
$$\begin{aligned}
\zeta_{s+\rbk{r-1}\beta}
& =
\omega_{i_r,s},
\quad
0\leq s<\beta,
\quad
1\leq r\leq q,
\\
\zeta_{q\beta}
& =
x_{i_1}.
\end{aligned}$$
Repeated application of the bridge convolution law
\eqref{eq:fock-bridge-convolution-law} gives, for every nonnegative function $F$ of the
concatenated loop,
$$\begin{aligned}
&
\sum_{x_{i_1},\ldots,x_{i_q}\in\Lambda}
\int_{\prod_{r=1}^{q}
\Omega_{\Lambda;x_{i_r},x_{i_{r+1}}}^{\beta}}
\fun{F}{\zeta}
\prod_{r=1}^{q}
\opdmsr{\msrcal{W}_{\Lambda;x_{i_r},x_{i_{r+1}}}^{\beta}}(\omega_{i_r})
=
\sum_{x_{i_1}\in\Lambda}
\int_{\Omega_{\Lambda;x_{i_1},x_{i_1}}^{q\beta}}
\fun{F}{\zeta}
\opdmsr{\msrcal{W}_{\Lambda;x_{i_1},x_{i_1}}^{q\beta}}(\zeta).
\end{aligned}$$
The loop occupation at a time $s\in\closedinterval{0}{\beta}$ is the sum of the occupations of its
constituent paths at that time:
$$\begin{aligned}
\fun{n_{\zeta}}{s,x}
& =
\sum_{r=1}^{q}
\fun{\fndef{\setone{x}}}{\omega_{i_r,s}}.
\end{aligned}$$
If $\boldsymbol{\zeta}$ is the tuple obtained by concatenating every cycle of $\varpi$, summation of this
identity over all cycles gives
$$\begin{aligned}
\fun{n_{\boldsymbol{\zeta}}}{s,x}
=
\sum_{i=1}^{N}
\fun{\fndef{\setone{x}}}{\omega_{i,s}},
\quad
\fun{\Phi_{\Lambda}}{\boldsymbol{\zeta}}
=
\fun{\Phi_{\Lambda}}{\omega_1,\ldots,\omega_N}.
\end{aligned}$$

For each $q\in\semigrposint$, define the length-$q\beta$ part of the ideal loop measure by
$$\begin{aligned}
\int_{\Omega_{\Lambda}^{\mathrm{loop}}}
\fun{F}{\zeta}
\opdmsr{\msrcal{L}_{\Lambda,q}}(\zeta)
& =
\frac{1}{q}
\sum_{z\in\Lambda}
\int_{\Omega_{\Lambda;z,z}^{q\beta}}
\fun{F}{\zeta}
\opdmsr{\msrcal{W}_{\Lambda;z,z}^{q\beta}}(\zeta).
\end{aligned}$$
The loop-measure definition \eqref{eq:fock-ideal-loop-measure} gives
$\msrcal{L}_{\Lambda}
=
\sum_{q\in\semigrposint}
\msrcal{L}_{\Lambda,q}$.
For fixed cycle counts satisfying \eqref{eq:fock-bosonic-cycle-length-constraint}, the coefficient
\eqref{eq:fock-bosonic-cycle-count}, the bridge convolution law
\eqref{eq:fock-bridge-convolution-law}, and the energy identity compute the
corresponding contribution to \eqref{eq:fock-precycle-partition-function} as
$$\begin{aligned}
&
\prod_{q\in\semigrposint}\frac{1}{n_q!}
\int_{\prod_{q\in\semigrposint}
\rbk{\Omega_{\Lambda}^{\mathrm{loop}}}^{n_q}}
\napiernum^{-\fun{\Phi_{\Lambda}}{\boldsymbol{\zeta}}}
\prod_{q\in\semigrposint}
\opdmsr{\msrcal{L}_{\Lambda,q}^{\otimes n_q}}(\boldsymbol{\zeta}_q).
\end{aligned}$$
Only finitely many $n_q$ are nonzero because their weighted sum is $N$.
Let $\ell=\sum_{q\in\semigrposint}n_q$.
The multinomial expansion of
$\msrcal{L}_{\Lambda}^{\otimes\ell}$ contains
$\ell!/\prod_{q\in\semigrposint}n_q!$ orderings with these fixed counts.
Multiplication by $1/\ell!$ gives
$$\begin{aligned}
\frac{1}{\ell!}
\frac{\ell!}{\prod_{q\in\semigrposint}n_q!}
& =
\prod_{q\in\semigrposint}\frac{1}{n_q!}.
\end{aligned}$$
Summing over all cycle counts and then over $\ell$ yields
\begin{equation}
\begin{aligned}
Z_{\Lambda,\beta}
& =
\sum_{\ell\in\monnat}\frac{1}{\ell!}
\int_{\rbk{\Omega_{\Lambda}^{\mathrm{loop}}}^{\ell}}
\napiernum^{-\fun{\Phi_{\Lambda}}{\boldsymbol{\zeta}}}
\opdmsr{\msrcal{L}_{\Lambda}^{\otimes\ell}}(\boldsymbol{\zeta})
=
Z_{\Lambda}.
\end{aligned}
\label{eq:fock-loop-grand-partition-identity}
\end{equation}
The quantity $Z_{\Lambda,\beta}$ in \eqref{eq:fock-loop-grand-partition-identity} is the operator trace defined
in \eqref{eq:fock-finite-volume-gibbs-state}.
The quantity $Z_{\Lambda}$ is the loop series defined in
\eqref{eq:fock-open-path-loop-partition-functions}.
Equation \eqref{eq:fock-loop-grand-partition-identity} proves
\eqref{eq:fock-continuous-cycle-partition-identity}.

For the reduced kernel, the $k$ creation-annihilation pairs cut the permutation cycles that contain the marked
particle-coordinate indices.
Fix $\varpi \in \grsym{k}$, marked lengths $q_1,\ldots,q_k$, and unmarked cycle counts
$\seq{n_q}{q \in \semigrposint}$ with
$$\begin{aligned}
N
& =
\sum_{i = 1}^{k}q_i
+
\sum_{q \in \semigrposint}q n_q.
\end{aligned}$$
The factor $N!/\rbk{N-k}!$ selects and orders the $k$ distinguished particle-coordinate indices.
After the marked endpoints and $\varpi$ are fixed, the remaining $N-k$ particle-coordinate indices fill the
$q_i-1$ intermediate positions of the $i$-th marked chain, $1\leq i\leq k$, and the positions in the unmarked
cycles.
There is no cyclic-rotation quotient for a marked chain because its initial endpoint is prescribed.
The number of arrangements of the remaining particle-coordinate indices is
$$\begin{aligned}
\frac{\rbk{N-k}!}
{\prod_{q\in\semigrposint}q^{n_q}n_q!}.
\end{aligned}$$
Multiplication by the selection factor and the bosonic symmetrizer gives
$$\begin{aligned}
\frac{N!}{\rbk{N - k}!}
\frac{1}{N!}
\frac{\rbk{N - k}!}{\prod_{q \in \semigrposint}q^{n_q} n_q!}
& =
\prod_{q \in \semigrposint}
\frac{1}{n_q!}\rbk{\frac{1}{q}}^{n_q}.
\end{aligned}$$
For fixed marked paths
$\boldsymbol{\omega}=\rbk{\omega_1,\ldots,\omega_k}$, summation over all unmarked cycle counts is
$$\begin{aligned}
\sum_{\ell\in\monnat}\frac{1}{\ell!}
\int_{\rbk{\Omega_{\Lambda}^{\mathrm{loop}}}^{\ell}}
\napiernum^{-\fun{\Phi_{\Lambda}}{\boldsymbol{\omega},\boldsymbol{\zeta}}}
\opdmsr{\msrcal{L}_{\Lambda}^{\otimes\ell}}(\boldsymbol{\zeta})
=
\fun{Z_{\Lambda}}{\boldsymbol{\omega}}.
\end{aligned}$$
Substitution of this identity into the marked-chain sum gives
$$\begin{aligned}
\fun{\widetilde{\rho}_{\Lambda,k}}{\boldsymbol{x};\boldsymbol{y}}
& =
\sum_{\varpi \in \grsym{k}}
\sum_{q_1,\ldots,q_k\in\semigrposint}
\int_{\prod_{i=1}^{k}
\Omega_{\Lambda;x_i,y_{\fun{\varpi}{i}}}^{q_i\beta}}
\frac{\fun{Z_{\Lambda}}{\boldsymbol{\omega}}}{Z_{\Lambda}}
\prod_{i=1}^{k}
\opdmsr{\msrcal{W}_{\Lambda;x_i,y_{\fun{\varpi}{i}}}^{q_i\beta}}(\omega_i).
\end{aligned}$$
This identity is \eqref{eq:fock-continuous-cycle-reduced-kernel}.
The loop and marked-chain tail bounds prove absolute convergence of every rearrangement in the calculation.
\end{proof}

\begin{lem}[Open-path and loop expansion]\label{lem:fock-loop-marked-chain-expansion}
Assume \eqref{eq:fock-path-chemical-potential-gap}, suppose that $v$ is pointwise nonnegative, let
$\Lambda\Subset\Gamma$ be finite, and let $\beta>0$.
For every $k\in\semigrposint$ and $\boldsymbol{x},\boldsymbol{y}\in\Lambda^k$, the reduced kernel satisfies
\begin{equation}
\begin{aligned}
\fun{\widetilde{\rho}_{\Lambda,k}}{\boldsymbol{x};\boldsymbol{y}}
& =
\sum_{\varpi \in \grsym{k}}
\sum_{q_1,\ldots,q_k \in \semigrposint}
\int_{\prod_{i = 1}^{k}\Omega_{\Lambda;x_i,y_{\fun{\varpi}{i}}}^{q_i\beta}}
\frac{\fun{Z_{\Lambda}}{\boldsymbol{\omega}}}{Z_{\Lambda}}
\prod_{i = 1}^{k}
\opdmsr{\msrcal{W}_{\Lambda;x_i,y_{\fun{\varpi}{i}}}^{q_i\beta}}(\omega_i).
\end{aligned}
\label{eq:fock-reduced-density-open-path-expansion}
\end{equation}
\end{lem}

\begin{proof}
Lemma \ref{lem:fock-finite-slice-trace-formula} and
Lemma \ref{lem:fock-creation-annihilation-sector-formula} provide the coordinate formulas used in the
finite-slice expansion.
Lemma \ref{lem:fock-finite-slice-open-path-expansion} gives the exact finite-slice formula.
Its trace normalization is \eqref{eq:fock-finite-slice-partition-trace-identity}.
Lemma \ref{lem:fock-finite-slice-domination} supplies bounds that are uniform in the slice number.
Lemma \ref{lem:fock-continuous-time-limit} identifies the limit of the finite-slice reduced kernels with the
right-hand side of \eqref{eq:fock-reduced-density-open-path-expansion}.
The Trotter product \eqref{eq:fock-finite-slice-trotter-product} and the operator
$K_{\Lambda,\smchemicalpotential_{\Lambda}}$ in \eqref{eq:fock-bose-hubbard-hamiltonian} satisfy
$$\begin{aligned}
\lim_{m\to\infty}
\fnrestr{T_{\Lambda,m}}{P_{\Lambda,N}\sphilb{F}_{\Lambda}}
& =
\fnrestr{
\napiernum^{-\beta K_{\Lambda,\smchemicalpotential_{\Lambda}}}
}{P_{\Lambda,N}\sphilb{F}_{\Lambda}}
\end{aligned}$$
in the operator norm for every $N\in\monnat$.
The trace majorant in Proposition \ref{prop:fock-finite-volume-gibbs-existence} permits the particle-sector sum
and the limit to be interchanged.
The resulting limit is
$$\begin{aligned}
\lim_{m\to\infty}
\fun{\widetilde{\rho}_{\Lambda,k}^{(m)}}{\boldsymbol{x};\boldsymbol{y}}
& =
\fun{\widetilde{\rho}_{\Lambda,k}}{\boldsymbol{x};\boldsymbol{y}}.
\end{aligned}$$
Comparison with \eqref{eq:fock-continuous-time-open-path-limit} proves
\eqref{eq:fock-reduced-density-open-path-expansion}.
Lemma \ref{lem:fock-continuous-bosonic-cycle-decomposition} gives the equality of the operator and loop
normalizations in \eqref{eq:fock-continuous-cycle-partition-identity} and independently verifies the reduced
kernel formula in \eqref{eq:fock-continuous-cycle-reduced-kernel}.
\end{proof}

Let \(\oastate[\psi_{\txtfr,\Lambda}]\) be the Gibbs state at inverse temperature \(\beta\) for the free grand-canonical Hamiltonian \(\physham_{\txtfr,\Lambda}-\smchemicalpotential_{\Lambda}\opfocknumber_{\Lambda}\). It has the same hopping and chemical potential as the interacting Gibbs state \(\oastate[\psi_{\Lambda}]\) in \eqref{eq:fock-finite-volume-gibbs-state}. Let \(\widetilde{\rho}_{\txtfr,\Lambda,k}\) be its reduced density kernel defined by \eqref{eq:fock-reduced-density-kernel}. Equation \eqref{eq:fock-path-interaction-energy} is identically zero when \(v=0\). Let \(Z_{\txtfr,\Lambda}\) and \(\fun{Z_{\txtfr,\Lambda}}{\boldsymbol{\omega}}\) be the two loop partition functions of this free gas, obtained from \eqref{eq:fock-open-path-loop-partition-functions} by setting \(v=0\). In particular, \begin{equation}
\begin{aligned}
\fun{\widetilde{\rho}_{\txtfr,\Lambda,k}}{\boldsymbol{x};\boldsymbol{y}}
& =
\fun{\oastate[\psi_{\txtfr,\Lambda}]}{
\faadj{a_{y_1}}\cdots\faadj{a_{y_k}}a_{x_k}\cdots a_{x_1}}.
\end{aligned}
\label{eq:fock-free-reduced-density-kernel}
\end{equation}

\begin{lem}[Repulsive loop domination]\label{lem:fock-repulsive-loop-domination}
Assume \eqref{eq:fock-path-chemical-potential-gap}, suppose that $v$ is pointwise nonnegative, let
$\Lambda\Subset\Gamma$ be finite, and let $\beta>0$.
For every finite open-path tuple $\boldsymbol{\omega}$ in the sense of
Definition \ref{def:fock-finite-open-path-tuple}, the normalized loop factor satisfies
\begin{equation}
\begin{aligned}
\frac{\fun{Z_{\Lambda}}{\boldsymbol{\omega}}}{Z_{\Lambda}}
& \leq
\napiernum^{-\fun{\Phi_{\Lambda}}{\boldsymbol{\omega}}}
\leq
1.
\end{aligned}
\label{eq:fock-interacting-loop-cancellation}
\end{equation}
For every $k\in\semigrposint$ and $\boldsymbol{x},\boldsymbol{y}\in\Lambda^k$, the interacting reduced density
kernel in \eqref{eq:fock-reduced-density-kernel} is bounded by the free kernel in
\eqref{eq:fock-free-reduced-density-kernel} with the same hopping and chemical potential:
\begin{equation}
\begin{aligned}
0
\leq
\fun{\widetilde{\rho}_{\Lambda,k}}{\boldsymbol{x};\boldsymbol{y}}
& \leq
\fun{\widetilde{\rho}_{\txtfr,\Lambda,k}}{\boldsymbol{x};\boldsymbol{y}}.
\end{aligned}
\label{eq:fock-reduced-density-domination}
\end{equation}
\end{lem}

\begin{proof}
Use the open-path representation in Lemma \ref{lem:fock-loop-marked-chain-expansion}.
For two finite open-path tuples in the sense of Definition \ref{def:fock-finite-open-path-tuple}, define their
nonnegative cross energy by
$$\begin{aligned}
\fun{I_{\Lambda}}{\boldsymbol{\omega},\boldsymbol{\zeta}}
& =
2\int_{\closedinterval{0}{\beta}}
\sum_{x,y \in \Lambda}
\fun{v}{x,y}
\fun{n_{\boldsymbol{\omega}}}{s,x}
\fun{n_{\boldsymbol{\zeta}}}{s,y}
\opdmsr{s}.
\end{aligned}$$
Expanding the concatenated-tuple definition
\eqref{eq:fock-concatenated-path-interaction-energy} with the occupation-field definition
\eqref{eq:fock-path-tuple-occupation-field} and the interaction energy
\eqref{eq:fock-path-interaction-energy} starts from
$$\begin{aligned}
\fun{n_{\rbk{\boldsymbol{\omega},\boldsymbol{\zeta}}}}{s,x}
& =
\fun{n_{\boldsymbol{\omega}}}{s,x}
+
\fun{n_{\boldsymbol{\zeta}}}{s,x},
\\
\fun{n_{\rbk{\boldsymbol{\omega},\boldsymbol{\zeta}}}}{s,x}
\fun{n_{\rbk{\boldsymbol{\omega},\boldsymbol{\zeta}}}}{s,y}
& =
\fun{n_{\boldsymbol{\omega}}}{s,x}
\fun{n_{\boldsymbol{\omega}}}{s,y}
+
\fun{n_{\boldsymbol{\zeta}}}{s,x}
\fun{n_{\boldsymbol{\zeta}}}{s,y}
\\
& +
\fun{n_{\boldsymbol{\omega}}}{s,x}
\fun{n_{\boldsymbol{\zeta}}}{s,y}
+
\fun{n_{\boldsymbol{\zeta}}}{s,x}
\fun{n_{\boldsymbol{\omega}}}{s,y}.
\end{aligned}$$
The linear normal-ordering term splits without a mixed contribution.
Symmetry of $v$ makes the two mixed quadratic terms equal.
Integration over $s$ gives
\begin{equation}
\begin{aligned}
\fun{\Phi_{\Lambda}}{\boldsymbol{\omega},\boldsymbol{\zeta}}
=
\fun{\Phi_{\Lambda}}{\boldsymbol{\omega}}
+
\fun{\Phi_{\Lambda}}{\boldsymbol{\zeta}}
+
\fun{I_{\Lambda}}{\boldsymbol{\omega},\boldsymbol{\zeta}}
\geq
\fun{\Phi_{\Lambda}}{\boldsymbol{\omega}}
+
\fun{\Phi_{\Lambda}}{\boldsymbol{\zeta}}.
\end{aligned}
\label{eq:fock-repulsive-path-splitting}
\end{equation}
The inequality uses pointwise nonnegativity of $v$.
Insert \eqref{eq:fock-repulsive-path-splitting} into the first member of
\eqref{eq:fock-open-path-loop-partition-functions}.
The normalized background contribution is bounded by
$$\begin{aligned}
\frac{\fun{Z_{\Lambda}}{\boldsymbol{\omega}}}{Z_{\Lambda}}
\leq
\napiernum^{-\fun{\Phi_{\Lambda}}{\boldsymbol{\omega}}}
\frac{1}{Z_{\Lambda}}
\sum_{\ell \in \monnat}\frac{1}{\ell!}
\int_{\rbk{\Omega_{\Lambda}^{\mathrm{loop}}}^{\ell}}
\napiernum^{-\fun{\Phi_{\Lambda}}{\boldsymbol{\zeta}}}
\opdmsr{\msrcal{L}_{\Lambda}^{\otimes\ell}}(\boldsymbol{\zeta})
=
\napiernum^{-\fun{\Phi_{\Lambda}}{\boldsymbol{\omega}}}.
\end{aligned}$$
The final loop series is exactly $Z_{\Lambda}$.
All bridge measures in \eqref{eq:fock-reduced-density-open-path-expansion} are positive.
The bound \eqref{eq:fock-interacting-loop-cancellation} gives
$$\begin{aligned}
0
\leq
\fun{\widetilde{\rho}_{\Lambda,k}}{\boldsymbol{x};\boldsymbol{y}}
& \leq
\sum_{\varpi \in \grsym{k}}
\sum_{q_1,\ldots,q_k\in\semigrposint}
\int_{\prod_{i=1}^{k}
\Omega_{\Lambda;x_i,y_{\fun{\varpi}{i}}}^{q_i\beta}}
\prod_{i=1}^{k}
\opdmsr{\msrcal{W}_{\Lambda;x_i,y_{\fun{\varpi}{i}}}^{q_i\beta}}(\omega_i).
\end{aligned}$$
For this free gas, $v=0$ implies
$$\begin{aligned}
\fun{Z_{\txtfr,\Lambda}}{\boldsymbol{\omega}}
=
Z_{\txtfr,\Lambda},
\quad
\frac{\fun{Z_{\txtfr,\Lambda}}{\boldsymbol{\omega}}}
{Z_{\txtfr,\Lambda}}
=
1.
\end{aligned}$$
The right-hand side of the reduced-kernel bound is consequently
$\fun{\widetilde{\rho}_{\txtfr,\Lambda,k}}{\boldsymbol{x};\boldsymbol{y}}$, which proves
\eqref{eq:fock-reduced-density-domination}.
\end{proof}

\subsection{Factorial and exponential local moments}\label{factorial-and-exponential-local-moments}

The free Bose gas in the comparison has one-particle covariance \[\begin{aligned}
C_{\txtfr,\Lambda}
& =
\rbk{\napiernum^{\beta\rbk{\physham[h]_{\Lambda} - \smchemicalpotential_{\Lambda}}} - 1}^{-1}.
\end{aligned}\] The gap \eqref{eq:fock-path-chemical-potential-gap} implies \(\physham[h]_{\Lambda} - \smchemicalpotential_{\Lambda} \geq \varepsilon\) and \[\begin{aligned}
\bkt{\delta_x}{C_{\txtfr,\Lambda}\delta_x}
& \leq
\rbk{\napiernum^{\beta\varepsilon} - 1}^{-1}.
\end{aligned}\]

\begin{thm}[Exponential replacement of the DLL moment hypothesis]\label{thm:fock-all-temperature-local-number}
Assume \eqref{eq:fock-path-chemical-potential-gap} and pointwise nonnegativity of $v$.
For every $\beta > 0$, $k \in \monnat$, and $0 < a < \beta\varepsilon$, the finite-volume Gibbs states defined
by \eqref{eq:fock-finite-volume-gibbs-state} satisfy
\begin{equation}
\begin{aligned}
\sup_{\Lambda \Subset \Gamma}
\sup_{x \in \Lambda}
\fun{\oastate[\psi_{\Lambda}]}{
\prod_{j = 0}^{k - 1}\rbk{N_x - j}}
& \leq
\frac{k!}{\rbk{\napiernum^{\beta\varepsilon} - 1}^{k}},
\end{aligned}
\label{eq:fock-factorial-local-number-bound}
\end{equation}
where the product is $1$ for $k = 0$, and
\begin{equation}
\begin{aligned}
\sup_{\Lambda \Subset \Gamma}
\sup_{x \in \Lambda}
\fun{\oastate[\psi_{\Lambda}]}{\napiernum^{aN_x}}
& \leq
\frac{\napiernum^{\beta\varepsilon} - 1}
{\napiernum^{\beta\varepsilon} - \napiernum^{a}}.
\end{aligned}
\label{eq:fock-exponential-local-number-bound}
\end{equation}
\end{thm}

\begin{proof}
For $k \in \semigrposint$, evaluate \eqref{eq:fock-reduced-density-domination} at
$\boldsymbol{x} = \boldsymbol{y} = \rbk{x,\ldots,x}$.
The domination is Lemma \ref{lem:fock-repulsive-loop-domination}.
The occupation identity and the free Bose Wick rule are
$$\begin{aligned}
\rbk{\faadj{a_x}}^{k}a_x^{k}
=
\prod_{j = 0}^{k - 1}\rbk{N_x - j},
\quad
\fun{\oastate[\psi_{\txtfr,\Lambda}]}{
\rbk{\faadj{a_x}}^{k}a_x^{k}}
=
k! \rbk{\bkt{\delta_x}{C_{\txtfr,\Lambda}\delta_x}}^{k}.
\end{aligned}$$
Indeed, fix the basis vector $\Psi_{\Lambda,\boldsymbol{\nu}}$ in
\eqref{eq:fock-occupation-number-basis} with $\nu_x=n$.
For $n\geq k$, let $\boldsymbol{\nu}^{(k)}$ be the multi-index obtained by replacing $\nu_x$ with $n-k$.
Repeated annihilation and creation give
$$\begin{aligned}
a_x^k\Psi_{\Lambda,\boldsymbol{\nu}}
& =
\sqrt{n\rbk{n-1}\cdots\rbk{n-k+1}}
\Psi_{\Lambda,\boldsymbol{\nu}^{(k)}},
\\
\rbk{\faadj{a_x}}^ka_x^k\Psi_{\Lambda,\boldsymbol{\nu}}
& =
n\rbk{n-1}\cdots\rbk{n-k+1}
\Psi_{\Lambda,\boldsymbol{\nu}},
\end{aligned}$$
while $a_x^k\Psi_{\Lambda,\boldsymbol{\nu}}=0$ when $n<k$.
The two diagonal operators consequently agree.
The free Bose Wick rule sums over all bijections between the $k$ creation and $k$ annihilation entries.
Every contraction equals $\bkt{\delta_x}{C_{\txtfr,\Lambda}\delta_x}$, and the number of bijections is
$\abscard{\grsym{k}} = k!$.
The domination and covariance estimates now give
$$\begin{aligned}
&
\fun{\oastate[\psi_{\Lambda}]}{
\prod_{j=0}^{k-1}\rbk{N_x-j}}
=
\fun{\widetilde{\rho}_{\Lambda,k}}{
\rbk{x,\ldots,x};\rbk{x,\ldots,x}}
\\
& \leq
\fun{\widetilde{\rho}_{\txtfr,\Lambda,k}}{
\rbk{x,\ldots,x};\rbk{x,\ldots,x}}
=
k!\rbk{\bkt{\delta_x}{C_{\txtfr,\Lambda}\delta_x}}^k
\leq
\frac{k!}{\rbk{\napiernum^{\beta\varepsilon}-1}^k}.
\end{aligned}$$
Taking the two suprema proves \eqref{eq:fock-factorial-local-number-bound}.
The case $k = 0$ is the normalization of the state.

For every $n \in \monnat$, the binomial theorem gives the factorial generating identity
\begin{equation}
\begin{aligned}
&
\napiernum^{a n}
=
\rbk{1+\napiernum^{a}-1}^{n}
=
\sum_{k=0}^{n}
\binom{n}{k}\rbk{\napiernum^{a}-1}^{k}
\\
& =
\sum_{k=0}^{n}
\frac{\rbk{\napiernum^{a}-1}^{k}}{k!}
\frac{n!}{\rbk{n-k}!}
=
\sum_{k \in \monnat}
\frac{\rbk{\napiernum^{a} - 1}^{k}}{k!}
\prod_{j = 0}^{k - 1}\rbk{n - j},
\end{aligned}
\label{eq:fock-factorial-generating-identity}
\end{equation}
where the final infinite sum equals the finite sum because the falling factorial vanishes for $k>n$.
Apply \eqref{eq:fock-factorial-generating-identity} to the spectral resolution of $N_x$.
Every summand is positive.
Monotone convergence and \eqref{eq:fock-factorial-local-number-bound} give
$$\begin{aligned}
\fun{\oastate[\psi_{\Lambda}]}{\napiernum^{aN_x}}
=
\sum_{k \in \monnat}
\frac{\rbk{\napiernum^{a}-1}^{k}}{k!}
\fun{\oastate[\psi_{\Lambda}]}{
\prod_{j=0}^{k-1}\rbk{N_x-j}}
\leq
\sum_{k \in \monnat}
\rbk{
\frac{\napiernum^{a}-1}
{\napiernum^{\beta\varepsilon}-1}}^{k}
=
\frac{1}
{1-\frac{\napiernum^{a}-1}{\napiernum^{\beta\varepsilon}-1}}
=
\frac{\napiernum^{\beta\varepsilon} - 1}
{\napiernum^{\beta\varepsilon} - \napiernum^{a}}.
\end{aligned}$$
The geometric series converges because $a < \beta\varepsilon$.
This proves \eqref{eq:fock-exponential-local-number-bound}.
\end{proof}

Theorem \ref{thm:fock-all-temperature-local-number} is uniform in the volume and the lattice site. It is an all-temperature statement but not an all-density statement. For \(\Gamma = \ringratint^{d}\) and nearest-neighbor hopping, one may take \(C_h = 2d\), and the hypothesis is a uniform version of \(\smchemicalpotential < -2d\) in \cite[Remark 5.3]{DeuchertLampartLemm001}.

\begin{defn}[Exponential spatial weights]\label{def:fock-exponential-spatial-weights}
Fix $b>0$.
For $x,y\in\Gamma$ and $\Delta\Subset\Gamma$ or $\Delta=\Gamma$, define
\begin{equation}
\begin{aligned}
\fun{q_x}{y}
& =
\napiernum^{-b\fun{\topdist}{x,y}},
&
Q_{x,\Delta}
& =
\sum_{y\in\Delta}\fun{q_x}{y}N_y,
&
S_b
& =
\sup_{x\in\Gamma}\sum_{y\in\Gamma}\fun{q_x}{y}.
\end{aligned}
\label{eq:fock-exponential-spatial-weights}
\end{equation}
\end{defn}

Let \(P_{\Delta,n}\) be the particle-sector projection defined by \eqref{eq:fock-particle-sector-projection} with \(X=\Delta\). The inequalities \(0<\fun{q_x}{y}\leq1\) give the sector bound \begin{equation}
\begin{aligned}
0
\leq
Q_{x,\Delta}P_{\Delta,n}
\leq
nP_{\Delta,n},
\quad
n\in\monnat.
\end{aligned}
\label{eq:fock-weighted-number-sector-bound}
\end{equation}

\begin{lem}[Initial weighted exponential moment]
For every $b>0$, the constant $S_b$ in \eqref{eq:fock-exponential-spatial-weights} is finite.
Assume \eqref{eq:fock-path-chemical-potential-gap} and pointwise nonnegativity of $v$.
Fix $\beta>0$ and $a\in\openinterval{0}{\beta\varepsilon}$, and define
\begin{equation}
\begin{aligned}
M_a
& =
\frac{\napiernum^{\beta\varepsilon}-1}
{\napiernum^{\beta\varepsilon}-\napiernum^a}.
\end{aligned}
\label{eq:fock-weighted-moment-constant}
\end{equation}
For every finite $\Lambda\Subset\Gamma$, $\Delta\in\setone{\Lambda,\Gamma}$, and $x\in\Gamma$,
the state $\oastate[\overline{\psi}_{\Lambda}]$ defined by
\eqref{eq:fock-vacuum-extended-gibbs-state} satisfies
\begin{equation}
\begin{aligned}
\fun{\oastate[\overline{\psi}_{\Lambda}]}{
\napiernum^{\frac{a}{S_b}Q_{x,\Delta}}}
& \leq
M_a.
\end{aligned}
\label{eq:fock-initial-weighted-exponential-moment}
\end{equation}
\end{lem}

\begin{proof}
For every $x\in\Gamma$, decomposition into the sets of points at fixed graph distance from $x$ and
\eqref{eq:fock-polynomial-ball-volume-growth} give
$$\begin{aligned}
\sum_{y\in\Gamma}\napiernum^{-b\fun{\topdist}{x,y}}
& =
1+
\sum_{n=1}^{\infty}
\napiernum^{-bn}
\abscard{\set{y\in\Gamma}{\fun{\topdist}{x,y}=n}}
\\
& \leq
1+
\sum_{n=1}^{\infty}
\napiernum^{-bn}\abscard{x[n]}
\leq
1+\sigma
\sum_{n=1}^{\infty}n^d\napiernum^{-bn}
<
\infty.
\end{aligned}$$
Taking the supremum over $x$ proves $S_b<\infty$.
The condition $0<a<\beta\varepsilon$ gives
$$\begin{aligned}
M_a-1
& =
\frac{\napiernum^a-1}
{\napiernum^{\beta\varepsilon}-\napiernum^a}
>
0.
\end{aligned}$$
The estimate \eqref{eq:fock-exponential-local-number-bound} states that
$\fun{\oastate[\overline{\psi}_{\Lambda}]}{\napiernum^{aN_y}}\leq M_a$ for $y\in\Lambda$.
For $y\in \Delta\setminus\Lambda$, definition
\eqref{eq:fock-vacuum-extended-gibbs-state} gives
$\fun{\oastate[\overline{\psi}_{\Lambda}]}{\napiernum^{aN_y}}=1\leq M_a$.
Consequently,
$$\begin{aligned}
\fun{\oastate[\overline{\psi}_{\Lambda}]}{\napiernum^{aN_y}}
& \leq
M_a,
\quad
y\in \Delta.
\end{aligned}$$
For $\boldsymbol{n}=\rbk{n_y}_{y\in \Delta}\in\monnat^{\Delta}$, define the joint probability mass by
$$\begin{aligned}
\fun{\nu_{\Lambda,\Delta}}{\boldsymbol{n}}
& =
\fun{\oastate[\overline{\psi}_{\Lambda}]}{
\prod_{y\in \Delta}\fun{\fndef{\setone{n_y}}}{N_y}}.
\end{aligned}$$
Here $\fun{\fndef{\setone{n_y}}}{N_y}$ is the spectral projection of $N_y$ for the eigenvalue $n_y$.
The commutativity of the number operators shows that the product is their joint spectral projection,
and normalization of the state gives
$\sum_{\boldsymbol{n}\in\monnat^{\Delta}}\fun{\nu_{\Lambda,\Delta}}{\boldsymbol{n}}=1$.
Since the definition of $S_b$ gives
$$\begin{aligned}
0
\leq
\sum_{y\in \Delta}\frac{\fun{q_x}{y}}{S_b}
\leq
1,
\end{aligned}$$
generalized Hölder on the probability space
$\rbk{\monnat^{\Delta},\nu_{\Lambda,\Delta}}$ uses the exponents
$S_b/\fun{q_x}{y}$, $y\in \Delta$.
If $\sum_{y\in \Delta}\fun{q_x}{y}/S_b<1$, include the constant function $1$ with exponent
$\rbk{1-\sum_{y\in \Delta}\fun{q_x}{y}/S_b}^{-1}$; if the sum is $1$, omit this factor.
Thus
$$\begin{aligned}
&
\fun{\oastate[\overline{\psi}_{\Lambda}]}{
\napiernum^{\frac{a}{S_b}Q_{x,\Delta}}}
=
\sum_{\boldsymbol{n}\in\monnat^{\Delta}}
\fun{\nu_{\Lambda,\Delta}}{\boldsymbol{n}}
\prod_{y\in \Delta}
\rbk{\napiernum^{a n_y}}^{\frac{\fun{q_x}{y}}{S_b}}
\\ 
& \leq
\rbk{
\sum_{\boldsymbol{n}\in\monnat^{\Delta}}
\fun{\nu_{\Lambda,\Delta}}{\boldsymbol{n}}}^{
1-\sum_{y\in \Delta}\fun{q_x}{y}/S_b}
\prod_{y \in \Delta}
\rbk{
\sum_{\boldsymbol{n}\in\monnat^{\Delta}}
\fun{\nu_{\Lambda,\Delta}}{\boldsymbol{n}}
\napiernum^{a n_y}}^{\frac{\fun{q_x}{y}}{S_b}}
\\ 
& =
\prod_{y \in \Delta}
\rbk{
\fun{\oastate[\overline{\psi}_{\Lambda}]}{
\napiernum^{aN_y}}}^{\frac{\fun{q_x}{y}}{S_b}}
\leq
M_a^{\sum_{y\in \Delta} \frac{\fun{q_x}{y}}{S_b}}
\leq
M_a.
\end{aligned}$$
This proves \eqref{eq:fock-initial-weighted-exponential-moment}.
\end{proof}

\begin{lem}[Exponential conjugation by weighted occupation]
Let $q_x$ and $Q_{x,\Delta}$ be defined by
\eqref{eq:fock-exponential-spatial-weights}.
For $\Delta\Subset\Gamma$ or $\Delta=\Gamma$, let
$\sphilb{D}_{\Delta,\txtfin}$ be defined by
\eqref{eq:fock-finite-particle-subspace} with $X=\Delta$.
For $x\in\Gamma$, $u,v\in\Delta$, and $s\in\fldreal$, the following identities hold on
$\sphilb{D}_{\Delta,\txtfin}$:
\begin{equation}
\begin{aligned}
\napiernum^{sQ_{x,\Delta}}\faadj{a_u}\napiernum^{-sQ_{x,\Delta}}
& =
\napiernum^{s\fun{q_x}{u}}\faadj{a_u},
\\
\napiernum^{sQ_{x,\Delta}}a_v\napiernum^{-sQ_{x,\Delta}}
& =
\napiernum^{-s\fun{q_x}{v}}a_v.
\end{aligned}
\label{eq:fock-exponential-creation-annihilation-conjugation}
\end{equation}
\end{lem}

\begin{proof}
The local-number commutators in
\eqref{eq:fock-creation-annihilation-number-commutators} give
$$\begin{aligned}
\commutator{Q_{x,\Delta}}{\faadj{a_u}}
& =
\sum_{y\in\Delta}
\fun{q_x}{y}\commutator{N_y}{\faadj{a_u}}
=
\fun{q_x}{u}\faadj{a_u},
\\
\commutator{Q_{x,\Delta}}{a_v}
& =
\sum_{y\in\Delta}
\fun{q_x}{y}\commutator{N_y}{a_v}
=
-\fun{q_x}{v}a_v.
\end{aligned}$$
Choose the pair $\rbk{A,\lambda}$ from the following two possibilities:
$$\begin{aligned}
\rbk{A,\lambda}
& \in
\setone{
\rbk{\faadj{a_u},\fun{q_x}{u}},
\rbk{a_v,-\fun{q_x}{v}}}.
\end{aligned}$$
The two commutator identities give $\commutator{Q_{x,\Delta}}{A}=\lambda A$.
The sector bound \eqref{eq:fock-weighted-number-sector-bound} shows that
$Q_{x,\Delta}$ is bounded on each particle-number summand and preserves
$\sphilb{D}_{\Delta,\txtfin}$.
For $n\in\monnat$, the first choice of $A$ maps the $n$-particle summand into the
$\rbk{n+1}$-particle summand.
For $n\in\semigrposint$, the second choice maps the $n$-particle summand into the
$\rbk{n-1}$-particle summand, and it vanishes on the vacuum summand.
Thus all operators in the following derivative act on $\sphilb{D}_{\Delta,\txtfin}$.
For $\Xi\in\sphilb{D}_{\Delta,\txtfin}$, the common calculation for both choices of
$\rbk{A,\lambda}$ is
$$\begin{aligned}
\opod{s}
\rbk{
\napiernum^{-s\lambda}
\napiernum^{sQ_{x,\Delta}}
A
\napiernum^{-sQ_{x,\Delta}}\Xi}
& =
\napiernum^{-s\lambda}
\napiernum^{sQ_{x,\Delta}}
\rbk{
-\lambda A+\commutator{Q_{x,\Delta}}{A}}
\napiernum^{-sQ_{x,\Delta}}\Xi
=
0.
\end{aligned}$$
At $s=0$, the vector inside the derivative equals $A\Xi$.
Hence the vector is independent of $s$.
Multiplication by $\napiernum^{s\lambda}$ gives
$$\begin{aligned}
\napiernum^{sQ_{x,\Delta}}
A
\napiernum^{-sQ_{x,\Delta}}\Xi
& =
\napiernum^{s\lambda}A\Xi.
\end{aligned}$$
The two choices of $\rbk{A,\lambda}$ prove
\eqref{eq:fock-exponential-creation-annihilation-conjugation}.
\end{proof}

\begin{lem}[Weighted hopping commutator]
Let $C_h$ be defined by \eqref{eq:fock-uniform-degree-bound}, and let
$q_x$, $Q_{x,\Delta}$, and $S_b$ be defined by
\eqref{eq:fock-exponential-spatial-weights}.
Fix $a,b>0$, and define
\begin{equation}
\begin{aligned}
c_{a,b}
& =
C_h\rbk{\napiernum^b-1}
\fnexp{\frac{a\rbk{\napiernum^b-1}}{2S_b}}.
\end{aligned}
\label{eq:fock-weighted-commutator-constant}
\end{equation}
For $\Delta\Subset\Gamma$ or $\Delta=\Gamma$, let
$\sphilb{D}_{\Delta,\txtfin}$ be the finite-particle subspace defined by
\eqref{eq:fock-finite-particle-subspace} with $X=\Delta$.
For $x\in\Gamma$, $0<\theta\leq a/S_b$, and
$\Xi\in\sphilb{D}_{\Delta,\txtfin}$, one has
\begin{equation}
\begin{aligned}
\mathord{\pm}\imunit
\bkt{\Xi}{
\commutator{H_{\Delta}}{\napiernum^{\theta Q_{x,\Delta}}}\Xi}
& \leq
c_{a,b}\theta
\bkt{\Xi}{Q_{x,\Delta}\napiernum^{\theta Q_{x,\Delta}}\Xi}.
\end{aligned}
\label{eq:fock-weighted-exponential-commutator}
\end{equation}
\end{lem}

\begin{proof}
The density interaction commutes with $Q_{x,\Delta}$.
Fix $u,v\in\Delta$ with $u\sim v$.
Equation \eqref{eq:fock-exponential-creation-annihilation-conjugation} and insertion of the identity
$\napiernum^{-\theta Q_{x,\Delta}}\napiernum^{\theta Q_{x,\Delta}}$
between $\faadj{a_u}$ and $a_v$ give the following conjugation formula on
$\sphilb{D}_{\Delta,\txtfin}$:
\begin{equation}
\begin{aligned}
\napiernum^{\theta Q_{x,\Delta}}
\faadj{a_u}a_v
\napiernum^{-\theta Q_{x,\Delta}}
=
\rbk{\napiernum^{\theta Q_{x,\Delta}}
\faadj{a_u}\napiernum^{-\theta Q_{x,\Delta}}}
\rbk{\napiernum^{\theta Q_{x,\Delta}}
a_v\napiernum^{-\theta Q_{x,\Delta}}}
=
\napiernum^{\theta\rbk{\fun{q_x}{u} - \fun{q_x}{v}}}\faadj{a_u}a_v.
\end{aligned}
\label{eq:fock-hopping-weight-conjugation}
\end{equation}
The triangle inequality for the graph distance gives
$\abs{\fun{\topdist}{x,u}-\fun{\topdist}{x,v}}\leq1$.
If $\fun{\topdist}{x,u}\geq\fun{\topdist}{x,v}$, then
$$\begin{aligned}
\abs{\fun{q_x}{u} - \fun{q_x}{v}}
& =
\fun{q_x}{u}
\rbk{
\napiernum^{b\rbk{
\fun{\topdist}{x,u}-\fun{\topdist}{x,v}}}
-1}
\leq
\rbk{\napiernum^b-1}\fun{q_x}{u}.
\end{aligned}$$
In this case $\fun{q_x}{u}\leq\fun{q_x}{v}$.
Interchanging $u$ and $v$ in the other case proves
\begin{equation}
\begin{aligned}
\abs{\fun{q_x}{u} - \fun{q_x}{v}}
& \leq
\rbk{\napiernum^{b} - 1}
\min\setone{\fun{q_x}{u},\fun{q_x}{v}},
\quad
u\sim v.
\end{aligned}
\label{eq:fock-q-edge-difference-bound}
\end{equation}
Multiplying \eqref{eq:fock-hopping-weight-conjugation} on the right by $\napiernum^{\theta Q_{x,\Delta}}$ and subtracting
the resulting equality from $\faadj{a_u}a_v\napiernum^{\theta Q_{x,\Delta}}$ gives
\begin{equation}
\begin{aligned}
\commutator{\faadj{a_u}a_v}{\napiernum^{\theta Q_{x,\Delta}}}
& =
\rbk{
1-
\napiernum^{\theta\rbk{\fun{q_x}{u}-\fun{q_x}{v}}}}
\faadj{a_u}a_v
\napiernum^{\theta Q_{x,\Delta}}.
\end{aligned}
\label{eq:fock-weighted-hopping-commutator}
\end{equation}
Fix $\Xi\in\sphilb{D}_{\Delta,\txtfin}$ and set
$\Phi=\fnexp{\theta Q_{x,\Delta}/2}\Xi$.
The sector bound \eqref{eq:fock-weighted-number-sector-bound} shows that $Q_{x,\Delta}$ preserves every
particle-number summand and is bounded on each such summand.
Hence $\Phi\in\sphilb{D}_{\Delta,\txtfin}$, so every inner product in the following calculation is defined.
For $s\in\fldreal$,
$$\begin{aligned}
\abs{1-\napiernum^s}\napiernum^{-s/2}
& =
2\abs{\sinh\rbk{s/2}}
\leq
\abs{s}\napiernum^{\abs{s}/2}.
\end{aligned}$$
The definition of $\Phi$ gives $\Xi=\fnexp{-\theta Q_{x,\Delta}/2}\Phi$.
Substitute this equality into \eqref{eq:fock-weighted-hopping-commutator} and apply the scalar inequality with
$s
=
\theta\rbk{\fun{q_x}{u}-\fun{q_x}{v}}$.
This gives
$$\begin{aligned}
&
\abs{
\bkt{\Xi}{
\commutator{\faadj{a_u}a_v}{\napiernum^{\theta Q_{x,\Delta}}}
\Xi}}
=
\abs{
1-\napiernum^{\theta\rbk{\fun{q_x}{u}-\fun{q_x}{v}}}}
\napiernum^{-\frac{\theta}{2}
\rbk{\fun{q_x}{u}-\fun{q_x}{v}}}
\abs{\bkt{\Phi}{\faadj{a_u}a_v\Phi}}
\\
& \leq
\frac{\theta}{2}
\abs{\fun{q_x}{u}-\fun{q_x}{v}}
\napiernum^{
\frac{\theta}{2}
\abs{\fun{q_x}{u}-\fun{q_x}{v}}}
\rbk{
\bkt{\Phi}{N_u\Phi}
+
\bkt{\Phi}{N_v\Phi}},
\end{aligned}$$
where the last line uses
$$\begin{aligned}
2\abs{\bkt{\Phi}{\faadj{a_u}a_v\Phi}}
=
2\abs{\bkt{a_u\Phi}{a_v\Phi}}
\leq
\norm{a_u\Phi}^{2}+\norm{a_v\Phi}^{2}
=
\bkt{\Phi}{N_u\Phi}+\bkt{\Phi}{N_v\Phi}.
\end{aligned}$$
Equation \eqref{eq:fock-q-edge-difference-bound} and
$\min\setone{\fun{q_x}{u},\fun{q_x}{v}}\leq\fun{q_x}{u}$ and
$\min\setone{\fun{q_x}{u},\fun{q_x}{v}}\leq\fun{q_x}{v}$ give
$$\begin{aligned}
&
\abs{\fun{q_x}{u}-\fun{q_x}{v}}
\rbk{
\bkt{\Phi}{N_u\Phi}
+
\bkt{\Phi}{N_v\Phi}}
\\ 
&\leq
\rbk{\napiernum^b-1}
\min\setone{\fun{q_x}{u},\fun{q_x}{v}}
\rbk{
\bkt{\Phi}{N_u\Phi}
+
\bkt{\Phi}{N_v\Phi}}
\\
& \leq
\rbk{\napiernum^{b}-1}
\rbk{
\fun{q_x}{u}\bkt{\Phi}{N_u\Phi}
+
\fun{q_x}{v}\bkt{\Phi}{N_v\Phi}}.
\end{aligned}$$
Equation \eqref{eq:fock-q-edge-difference-bound},
$\min\setone{\fun{q_x}{u},\fun{q_x}{v}}\leq1$, and $0<\theta\leq a/S_b$ also give
$$\begin{aligned}
\theta\abs{\fun{q_x}{u}-\fun{q_x}{v}}
& \leq
\frac{a\rbk{\napiernum^b-1}}{S_b}.
\end{aligned}$$
Summing the hopping terms and applying the uniform degree bound
\eqref{eq:fock-uniform-degree-bound} gives
$$\begin{aligned}
\abs{
\bkt{\Xi}{
\imunit\commutator{H_{\Delta}}{\napiernum^{\theta Q_{x,\Delta}}}
\Xi}}
\leq
c_{a,b}\theta
\sum_{y\in \Delta}
\fun{q_x}{y}
\bkt{\Phi}{N_y\Phi}
=
c_{a,b}\theta
\bkt{\Xi}{
Q_{x,\Delta}\napiernum^{\theta Q_{x,\Delta}}\Xi}.
\end{aligned}$$
The two choices of sign prove \eqref{eq:fock-weighted-exponential-commutator}.
\end{proof}

\begin{lem}[Inserted weighted exponential moment]
Assume \eqref{eq:fock-path-chemical-potential-gap} and pointwise nonnegativity of $v$.
Fix $b>0$, $\beta>0$, $a\in\openinterval{0}{\beta\varepsilon}$, a finite $\Lambda\Subset\Gamma$,
$\Delta\in\setone{\Lambda,\Gamma}$, a finite $\Lambda_0\Subset\Delta$, and $x\in\Gamma$.
Let $M_a$ be defined by \eqref{eq:fock-weighted-moment-constant}.
For every $A\in\fun{\oaresolventalgebra}{\Lambda_0}^{\gamma}$ and
$0<c\leq a/\rbk{S_b+\abscard{\Lambda_0}}$, one has
\begin{equation}
\begin{aligned}
\max\setone{
\fun{\oastate[\overline{\psi}_{\Lambda}]}{
A\napiernum^{cQ_{x,\Delta}}\faadj{A}},
\fun{\oastate[\overline{\psi}_{\Lambda}]}{
\faadj{A}\napiernum^{cQ_{x,\Delta}}A}}
& \leq
M_a\norm{A}^{2}.
\end{aligned}
\label{eq:fock-inserted-exponential-moment-bound}
\end{equation}
\end{lem}

\begin{proof}
Gauge invariance gives $\commutator{A}{N_{\Lambda_0}}=0$.
Write $Q_{x,\Delta}=Q_{x,\Lambda_0}+Q_{x,\Delta\setminus \Lambda_0}$.
For $0<c\leq a/(S_b+\abscard{\Lambda_0})$, the operator inequalities
$Q_{x,\Lambda_0}\leq N_{\Lambda_0}$ and
$$\begin{aligned}
A\napiernum^{cQ_{x,\Lambda_0}}\faadj{A}
\leq
A\napiernum^{cN_{\Lambda_0}}\faadj{A}
\leq
\norm{A}^{2}\napiernum^{cN_{\Lambda_0}}
\end{aligned}$$
hold because $A$ commutes with $N_{\Lambda_0}$.
Since $A$ also commutes with $Q_{x,\Delta\setminus \Lambda_0}$, generalized H\"older and
\eqref{eq:fock-exponential-local-number-bound} can be applied explicitly.
The operator bound is
$$\begin{aligned}
\fun{\oastate[\overline{\psi}_{\Lambda}]}{
A\napiernum^{cQ_{x,\Delta}}\faadj{A}}
& =
\fun{\oastate[\overline{\psi}_{\Lambda}]}{
A\napiernum^{cQ_{x,\Lambda_0}}\faadj{A}
\napiernum^{cQ_{x,\Delta\setminus \Lambda_0}}}
\\
& \leq
\norm{A}^{2}
\fun{\oastate[\overline{\psi}_{\Lambda}]}{
\fnexp{
c\sum_{y\in \Lambda_0}N_y
+
c\sum_{y\in \Delta\setminus \Lambda_0}\fun{q_x}{y}N_y}}.
\end{aligned}$$
The sum of the coefficients of the number operators in the last exponential is bounded by
$$\begin{aligned}
c\abscard{\Lambda_0}
+
c\sum_{y\in \Delta\setminus \Lambda_0}\fun{q_x}{y}
\leq
c\rbk{\abscard{\Lambda_0}+S_b}
\leq
a.
\end{aligned}$$
Generalized H\"older for the joint spectral measure of the commuting number operators gives
$$\begin{aligned}
&
\fun{\oastate[\overline{\psi}_{\Lambda}]}{
\fnexp{
c\sum_{y\in \Lambda_0}N_y
+
c\sum_{y\in \Delta\setminus \Lambda_0}\fun{q_x}{y}N_y}}
\leq
\prod_{y\in \Lambda_0}
\rbk{
\fun{\oastate[\overline{\psi}_{\Lambda}]}{\napiernum^{aN_y}}}^{c/a}
\prod_{y\in \Delta\setminus \Lambda_0}
\rbk{
\fun{\oastate[\overline{\psi}_{\Lambda}]}{\napiernum^{aN_y}}}
^{c\fun{q_x}{y}/a}
\\
& \leq
M_a^{
\frac{1}{a}
\rbk{c\abscard{\Lambda_0}
+
c\sum_{y\in \Delta\setminus \Lambda_0}\fun{q_x}{y}}}
\leq
M_a.
\end{aligned}$$
Combining the operator bound with the generalized Hölder estimate proves the first term in
\eqref{eq:fock-inserted-exponential-moment-bound}.
Replacing $A$ by $\faadj{A}$ proves the second term.
\end{proof}

\begin{lem}[Exponential propagation and cutoff removal]\label{lem:fock-exponential-cutoff-removal}
Assume \eqref{eq:fock-path-chemical-potential-gap}, and suppose that $v$ is pointwise nonnegative and has
finite range.
Fix $\beta>0$, $a \in \openinterval{0}{\beta\varepsilon}$, $T > 0$, and a finite $\Lambda_0 \Subset \Gamma$.
For a finite $\Lambda\Subset\Gamma$, choose $\Delta\in\setone{\Lambda,\Gamma}$ and a finite $\Delta_0$ satisfying
$\Lambda_0\subset \Delta_0\Subset \Delta$.
For $\lambda\geq1$, define the local particle-number cutoff
$$\begin{aligned}
P_{\Delta_0,\lambda}
& =
\prod_{y \in \Delta_0}\fun{\fndef{\closedinterval{0}{\lambda}}}{N_y}.
\end{aligned}$$
For $t\in\fldreal$ and $C\in\opspbddlin{\sphilb{F}_{\Delta}}$, define the corresponding cutoff dynamics by
$$\begin{aligned}
\fun{\widetilde{\alpha}_{\Delta,t}^{\Delta_0,\lambda}}{C}
& =
\napiernum^{\imunit t P_{\Delta_0,\lambda}H_{\Delta}P_{\Delta_0,\lambda}}
C
\napiernum^{-\imunit t P_{\Delta_0,\lambda}H_{\Delta}P_{\Delta_0,\lambda}}.
\end{aligned}$$
There exist constants $a_{\Lambda_0,T}>0$ and $C_{\Lambda_0,T,a}<\infty$, independent of
$\Lambda$, $\Delta$, $\Delta_0$, and $\lambda$, for which the following two conclusions hold.
\begin{enumerate}
\item
The propagated exponential moments satisfy
\begin{equation}
\begin{aligned}
\sup_{x \in \Gamma}
\sup_{\abs{t} \leq T}
\fun{\oastate[\overline{\psi}_{\Lambda}]}{
\fun{\alpha_{\Delta,t}}{\napiernum^{a_{\Lambda_0,T}N_x}}}
& \leq
C_{\Lambda_0,T,a}.
\end{aligned}
\label{eq:fock-exponential-moment-propagation}
\end{equation}
For every $A\in\fun{\oaresolventalgebra}{\Lambda_0}^{\gamma}$, the inserted propagated moments satisfy
\begin{equation}
\begin{aligned}
\sup_{x\in\Gamma}
\sup_{\abs{t}\leq T}
\max\setone{
\fun{\oastate[\overline{\psi}_{\Lambda}]}{
A\fun{\alpha_{\Delta,t}}{\napiernum^{a_{\Lambda_0,T}N_x}}\faadj{A}},
\fun{\oastate[\overline{\psi}_{\Lambda}]}{
\faadj{A}\fun{\alpha_{\Delta,t}}{\napiernum^{a_{\Lambda_0,T}N_x}}A}}
& \leq
C_{\Lambda_0,T,a}\norm{A}^{2}.
\end{aligned}
\label{eq:fock-propagated-inserted-exponential-moment-bound}
\end{equation}

\item
For every
$A \in \fun{\oaresolventalgebra}{\Lambda_0}^{\gamma}$, and every $B \in \oa{B}$, the two cutoff errors satisfy
\begin{equation}
\begin{aligned}
\sup_{\abs{t} \leq T}
\abs{
\fun{\oastate[\overline{\psi}_{\Lambda}]}{
\rbk{
\fun{\alpha_{\Delta,t}}{A}
-
\fun{\widetilde{\alpha}_{\Delta,t}^{\Delta_0,\lambda}}{P_{\Delta_0,\lambda}AP_{\Delta_0,\lambda}}}B}}
& \leq
C_{\Lambda_0,T,a}\abscard{\Delta_0}\lambda
\napiernum^{-a_{\Lambda_0,T}\lambda}\norm{A}\norm{B},
\\
\sup_{\abs{t} \leq T}
\abs{
\fun{\oastate[\overline{\psi}_{\Lambda}]}{
B\rbk{
\fun{\alpha_{\Delta,t}}{A}
-
\fun{\widetilde{\alpha}_{\Delta,t}^{\Delta_0,\lambda}}{P_{\Delta_0,\lambda}AP_{\Delta_0,\lambda}}}}}
& \leq
C_{\Lambda_0,T,a}\abscard{\Delta_0}\lambda
\napiernum^{-a_{\Lambda_0,T}\lambda}\norm{A}\norm{B}.
\end{aligned}
\label{eq:fock-exponential-cutoff-removal}
\end{equation}
\end{enumerate}
\end{lem}

\begin{proof}
Fix $b>0$ and use the weights in Definition \ref{def:fock-exponential-spatial-weights}.
Let $M_a$ and $c_{a,b}$ be defined by
\eqref{eq:fock-weighted-moment-constant} and \eqref{eq:fock-weighted-commutator-constant}, respectively.
The initial estimate is \eqref{eq:fock-initial-weighted-exponential-moment}.
For $t \geq 0$, put $\fun{\theta}{t}=aS_b^{-1}\napiernum^{-c_{a,b}t}$.
The propagated exponential moment is
$$\begin{aligned}
\fun{G_{x,\Delta}}{t}
& =
\fun{\oastate[\overline{\psi}_{\Lambda}] \circ \alpha_{\Delta,t}}
{\napiernum^{\fun{\theta}{t}Q_{x,\Delta}}}.
\end{aligned}$$
For $t\geq0$, one has $0<\fun{\theta}{t}\leq a/S_b$ and
$\fun{\theta'}{t}=-c_{a,b}\fun{\theta}{t}$.
The positive-sign estimate in \eqref{eq:fock-weighted-exponential-commutator} therefore gives the
quadratic-form inequality
$$\begin{aligned}
\fun{\theta'}{t}Q_{x,\Delta}\napiernum^{\fun{\theta}{t}Q_{x,\Delta}}
+
\imunit\commutator{H_{\Delta}}{\napiernum^{\fun{\theta}{t}Q_{x,\Delta}}}
& \leq
0.
\end{aligned}$$
Its derivative is consequently nonpositive:
$$\begin{aligned}
\fun{G_{x,\Delta}'}{t}
& =
\fun{\theta'}{t}
\fun{\oastate[\overline{\psi}_{\Lambda}] \circ \alpha_{\Delta,t}}
{Q_{x,\Delta}\napiernum^{\fun{\theta}{t}Q_{x,\Delta}}}
+
\fun{\oastate[\overline{\psi}_{\Lambda}] \circ \alpha_{\Delta,t}}
{\imunit\commutator{H_{\Delta}}
{\napiernum^{\fun{\theta}{t}Q_{x,\Delta}}}}
\leq
0.
\end{aligned}$$
Integration from $0$ to $t$ gives
$$\begin{aligned}
\fun{G_{x,\Delta}}{t}
\leq
\fun{G_{x,\Delta}}{0}
=
\fun{\oastate[\overline{\psi}_{\Lambda}]}
{\napiernum^{\frac{a}{S_b}Q_{x,\Delta}}}
\leq
M_a,
\quad
0\leq t\leq T.
\end{aligned}$$
Replacing $H_{\Delta}$ by $-H_{\Delta}$ in the commutator estimate gives the same bound for $-T\leq t\leq0$.
Define
$a_{\Lambda_0,T}=a\rbk{S_b+\abscard{\Lambda_0}}^{-1}\napiernum^{-c_{a,b}T}$.
The constant $c_{a,b}$ in \eqref{eq:fock-weighted-commutator-constant} is positive.
For $\abs{t}\leq T$, the definition of $\theta$ gives the coefficient comparison
$$\begin{aligned}
a_{\Lambda_0,T}
& =
\frac{a}{S_b+\abscard{\Lambda_0}}
\napiernum^{-c_{a,b}T}
\leq
\frac{a}{S_b}
\napiernum^{-c_{a,b}\abs{t}}
=
\fun{\theta}{\abs{t}}.
\end{aligned}$$
For $x\in\Delta$, the definition of $Q_{x,\Delta}$ in
\eqref{eq:fock-exponential-spatial-weights} and
$\fun{q_x}{x}=\napiernum^{-b\fun{\topdist}{x,x}}=1$ give
$$\begin{aligned}
Q_{x,\Delta}
& =
N_x+
\sum_{y\in\Delta\setminus\setone{x}}
\fun{q_x}{y}N_y
\geq
N_x.
\end{aligned}$$
Combining the coefficient and operator comparisons gives
$$\begin{aligned}
a_{\Lambda_0,T}N_x
& \leq
\fun{\theta}{\abs{t}}N_x
\leq
\fun{\theta}{\abs{t}}Q_{x,\Delta}.
\end{aligned}$$
The local number operators commute, so the joint spectral calculus gives
$$\begin{aligned}
\napiernum^{a_{\Lambda_0,T}N_x}
& \leq
\napiernum^{\fun{\theta}{\abs{t}}Q_{x,\Delta}}
\end{aligned}$$
for $x\in\Delta$ and $\abs{t}\leq T$.
For $x\notin\Delta$, one necessarily has $\Delta=\Lambda$ and
$x\in\Lambda^{\mathrm{c}}$.
The Hamiltonian $H_{\Lambda}$ acts only on $\sphilb{F}_{\Lambda}$, whereas $N_x$ acts on
$\sphilb{F}_{\Lambda^{\mathrm{c}}}$, and $N_x\Omega_{\Lambda^{\mathrm{c}}}=0$.
Consequently,
$$\begin{aligned}
\fun{\alpha_{\Lambda,t}}{\napiernum^{a_{\Lambda_0,T}N_x}}
& =
\napiernum^{a_{\Lambda_0,T}N_x},
&
\napiernum^{a_{\Lambda_0,T}N_x}\Omega_{\Lambda^{\mathrm{c}}}
& =
\Omega_{\Lambda^{\mathrm{c}}}.
\end{aligned}$$
Using the density matrix in \eqref{eq:fock-vacuum-extended-gibbs-state} therefore gives
$$\begin{aligned}
\fun{\oastate[\overline{\psi}_{\Lambda}]}{
\fun{\alpha_{\Lambda,t}}{\napiernum^{a_{\Lambda_0,T}N_x}}}
=
\sqfun{\trace_{\sphilb{F}_{\Lambda}}}{
\frac{\napiernum^{-\beta K_{\Lambda,\smchemicalpotential_{\Lambda}}}}
{Z_{\Lambda,\beta}}}
\bkt{\Omega_{\Lambda^{\mathrm{c}}}}{
\napiernum^{a_{\Lambda_0,T}N_x}\Omega_{\Lambda^{\mathrm{c}}}}
_{\sphilb{F}_{\Lambda^{\mathrm{c}}}}
=
1.
\end{aligned}$$
The two cases $x\in\Delta$ and $x\notin\Delta$ prove
\eqref{eq:fock-exponential-moment-propagation}.

Define
$$\begin{aligned}
\fun{\theta_{\Lambda_0}}{t}
& =
\frac{a}{S_b+\abscard{\Lambda_0}}\napiernum^{-c_{a,b}t},
\quad
t\geq0.
\end{aligned}$$
For $A\in\fun{\oaresolventalgebra}{\Lambda_0}^{\gamma}$,
Equation \eqref{eq:fock-inserted-exponential-moment-bound} gives the two initial bounds at
$c=\fun{\theta_{\Lambda_0}}{0}$.
Repeating the derivative calculation with $\theta_{\Lambda_0}$ in place of $\theta$ and using
both signs in \eqref{eq:fock-weighted-exponential-commutator} gives, for positive and negative times,
$$\begin{aligned}
\sup_{x\in\Gamma}
\sup_{\abs{t}\leq T}
\max\setone{
\fun{\oastate[\overline{\psi}_{\Lambda}]}{
A\fun{\alpha_{\Delta,t}}{\napiernum^{a_{\Lambda_0,T}N_x}}\faadj{A}},
\fun{\oastate[\overline{\psi}_{\Lambda}]}{
\faadj{A}\fun{\alpha_{\Delta,t}}{\napiernum^{a_{\Lambda_0,T}N_x}}A}}
& \leq
M_a\norm{A}^{2}.
\end{aligned}$$
After enlarging $C_{\Lambda_0,T,a}$ so that $C_{\Lambda_0,T,a}\geq M_a$,
this proves \eqref{eq:fock-propagated-inserted-exponential-moment-bound}.
The spectral inequality
$$\begin{aligned}
1-P_{\Delta_0,\lambda}
& \leq
\sum_{y \in \Delta_0}\fun{\fndef{\openinterval{\lambda}{\infty}}}{N_y}
\leq
\napiernum^{-a_{\Lambda_0,T}\lambda}
\sum_{y \in \Delta_0}\napiernum^{a_{\Lambda_0,T}N_y}
\end{aligned}$$
and \eqref{eq:fock-exponential-moment-propagation} give the propagated cutoff tail,
for $\abs{t}\leq T$,
$$\begin{aligned}
&
\fun{\oastate[\overline{\psi}_{\Lambda}]}{
\fun{\alpha_{\Delta,t}}{1-P_{\Delta_0,\lambda}}}
\leq
\napiernum^{-a_{\Lambda_0,T}\lambda}
\sum_{y\in \Delta_0}
\fun{\oastate[\overline{\psi}_{\Lambda}]}{
\fun{\alpha_{\Delta,t}}{\napiernum^{a_{\Lambda_0,T}N_y}}}
\leq
C_{\Lambda_0,T,a}\abscard{\Delta_0}
\napiernum^{-a_{\Lambda_0,T}\lambda}.
\end{aligned}$$
Let $C=P_{\Delta_0,\lambda}AP_{\Delta_0,\lambda}$.
The initial cutoff error decomposes as
$$\begin{aligned}
A-C
& =
\rbk{1-P_{\Delta_0,\lambda}}A
+
P_{\Delta_0,\lambda}A\rbk{1-P_{\Delta_0,\lambda}}.
\end{aligned}$$
Cauchy--Schwarz for the two positive functionals obtained by inserting $B$ on the right or on the left gives
\begin{equation}
\begin{aligned}
\abs{
\fun{\oastate[\overline{\psi}_{\Lambda}]}{
\fun{\alpha_{\Delta,t}}{
A-C}B}}
& \leq
C_{\Lambda_0,T,a}
\abscard{\Delta_0}
\napiernum^{-\frac{a_{\Lambda_0,T}\lambda}{2}}
\norm{A}\norm{B},
\\
\abs{
\fun{\oastate[\overline{\psi}_{\Lambda}]}{
B\fun{\alpha_{\Delta,t}}{
A-C}}}
& \leq
C_{\Lambda_0,T,a}
\abscard{\Delta_0}
\napiernum^{-\frac{a_{\Lambda_0,T}\lambda}{2}}
\norm{A}\norm{B}.
\end{aligned}
\label{eq:fock-initial-cutoff-error-bound}
\end{equation}

Let $K_{\Delta_0,\lambda}=P_{\Delta_0,\lambda}H_{\Delta}P_{\Delta_0,\lambda}$.
Differentiation of
$\fun{\alpha_{\Delta,t-s}\circ\widetilde{\alpha}_{\Delta,s}^{\Delta_0,\lambda}}{C}$ gives the Duhamel identity
\begin{equation}
\begin{aligned}
\fun{\alpha_{\Delta,t}}{C}
-
\fun{\widetilde{\alpha}_{\Delta,t}^{\Delta_0,\lambda}}{C}
& =
\imunit
\int_{0}^{t}
\fun{\alpha_{\Delta,t-s}}{
\commutator{
H_{\Delta}-K_{\Delta_0,\lambda}}{
\fun{\widetilde{\alpha}_{\Delta,s}^{\Delta_0,\lambda}}{C}}}
\opdmsr{s}
\end{aligned}
\label{eq:fock-cutoff-duhamel-identity}
\end{equation}
for $t\geq0$, with the orientation of the integral reversed for $t<0$.
Since $C=P_{\Delta_0,\lambda}CP_{\Delta_0,\lambda}$, every nonzero term in the commutator contains a factor
$1-P_{\Delta_0,\lambda}$ next to an oriented hopping operator $\faadj{a_u}a_v$ with
$u\sim v$ and $\setone{u,v}\cap \Delta_0\neq\emptyset$.
For every finite $Z\Subset \Delta$, define
$$\begin{aligned}
P_{Z,\lambda}
& =
\prod_{z\in Z}\fun{\fndef{\closedinterval{0}{\lambda}}}{N_z}.
\end{aligned}$$
For these ordered pairs, the inequalities
$$\begin{aligned}
\norm{P_{\setone{u,v},\lambda}\faadj{a_u}a_v}
& \leq
\sqrt{\lambda\rbk{\lambda+1}}
\leq
\sqrt{2}\lambda,
&
\norm{P_{\setone{u},\lambda}\faadj{a_u}a_v\rbk{N_v+1}^{-1/2}}
& \leq
\sqrt{\lambda}
\end{aligned}$$
follow directly from \eqref{eq:fock-occupation-number-basis}.
Define the spectral projection
\begin{equation}
\begin{aligned}
E_{y,\lambda}
& =
\fun{\fndef{\openinterval{\lambda}{\infty}}}{N_y}.
\end{aligned}
\label{eq:fock-local-cutoff-tail-projection}
\end{equation}
Equation \eqref{eq:fock-local-cutoff-tail-projection} and
\eqref{eq:fock-exponential-moment-propagation} give, for $u\in \Delta_0$ and $\abs{t-s}\leq T$,
$$\begin{aligned}
E_{u,\lambda}
\leq
\napiernum^{-a_{\Lambda_0,T}\lambda}\napiernum^{a_{\Lambda_0,T}N_u},
\quad
\fun{\oastate[\overline{\psi}_{\Lambda}]}{
\fun{\alpha_{\Delta,t-s}}{E_{u,\lambda}}}
\leq
C_{\Lambda_0,T,a}\napiernum^{-a_{\Lambda_0,T}\lambda}.
\end{aligned}$$
Expansion over the ordered pairs $u\sim v$ with $\setone{u,v}\cap \Delta_0\neq\emptyset$ produces a factor
$E_{u,\lambda}$ or $E_{v,\lambda}$ at an endpoint belonging to $\Delta_0$.
Fix $u,v\in \Delta_0$ with $u\sim v$.
The occupation-number basis \eqref{eq:fock-occupation-number-basis} gives
$$\begin{aligned}
\norm{\faadj{a_u}a_vP_{\Delta_0,\lambda}}
& \leq
\sqrt{\lambda\rbk{\lambda+1}}
\leq
\sqrt{2}\lambda,
\quad
\lambda\geq1.
\end{aligned}$$
For any projection $R$ and bounded operator $Z$, Cauchy--Schwarz for the state gives
$$\begin{aligned}
\abs{\fun{\oastate[\overline{\psi}_{\Lambda}]}{RZ}}^2
& \leq
\fun{\oastate[\overline{\psi}_{\Lambda}]}{R}
\fun{\oastate[\overline{\psi}_{\Lambda}]}{\faadj{Z}Z}
\leq
\fun{\oastate[\overline{\psi}_{\Lambda}]}{R}\norm{Z}^2.
\end{aligned}$$
Apply this inequality with
$$\begin{aligned}
R
=
\fun{\alpha_{\Delta,t-s}}{E_{u,\lambda}},
\quad
Z
=
\fun{\alpha_{\Delta,t-s}}{
\faadj{a_u}a_vP_{\Delta_0,\lambda}
\fun{\widetilde{\alpha}_{\Delta,s}^{\Delta_0,\lambda}}{C}}B.
\end{aligned}$$
Since $\norm{C}\leq\norm{A}$ and
$\norm{\fun{\widetilde{\alpha}_{\Delta,s}^{\Delta_0,\lambda}}{C}}=\norm{C}$, one has
$\norm{Z}
\leq
\sqrt{2}\lambda\norm{A}\norm{B}$.
Therefore,
\begin{equation}
\begin{aligned}
&
\abs{\fun{\oastate[\overline{\psi}_{\Lambda}]}{
\fun{\alpha_{\Delta,t-s}}{
E_{u,\lambda}
\faadj{a_u}a_vP_{\Delta_0,\lambda}
\fun{\widetilde{\alpha}_{\Delta,s}^{\Delta_0,\lambda}}{C}}B}}
\leq
\sqrt{2}\lambda\norm{A}\norm{B}
\rbk{\fun{\oastate[\overline{\psi}_{\Lambda}]}{
\fun{\alpha_{\Delta,t-s}}{E_{u,\lambda}}}}^{1/2}
\\
& \leq
\sqrt{2C_{\Lambda_0,T,a}}\lambda
\napiernum^{-\frac{a_{\Lambda_0,T}\lambda}{2}}
\norm{A}\norm{B}.
\end{aligned}
\label{eq:fock-right-cutoff-integrand-bound}
\end{equation}
Applying Cauchy--Schwarz to the adjoint gives
\begin{equation}
\begin{aligned}
\abs{\fun{\oastate[\overline{\psi}_{\Lambda}]}{
B\fun{\alpha_{\Delta,t-s}}{
E_{u,\lambda}
\faadj{a_u}a_vP_{\Delta_0,\lambda}
\fun{\widetilde{\alpha}_{\Delta,s}^{\Delta_0,\lambda}}{C}}}}
\leq
\sqrt{2C_{\Lambda_0,T,a}}\lambda
\napiernum^{-\frac{a_{\Lambda_0,T}\lambda}{2}}
\norm{A}\norm{B}.
\end{aligned}
\label{eq:fock-left-cutoff-integrand-bound}
\end{equation}
If Cauchy--Schwarz produces either of the two positive functionals containing $A$ or $\faadj{A}$,
use \eqref{eq:fock-propagated-inserted-exponential-moment-bound} with $A$ or $\faadj{A}$, respectively.
At most $2C_h\abscard{\Delta_0}$ ordered pairs $u\sim v$ satisfy $\setone{u,v}\cap \Delta_0\neq\emptyset$ by
\eqref{eq:fock-uniform-degree-bound}.
Applying \eqref{eq:fock-right-cutoff-integrand-bound} to each commutator term in
\eqref{eq:fock-cutoff-duhamel-identity}, summing over the ordered pairs, integrating in $s$ from $0$ to $t$,
and adding the first initial cutoff error in \eqref{eq:fock-initial-cutoff-error-bound} gives
$$\begin{aligned}
\sup_{\abs{t}\leq T}
\abs{\fun{\oastate[\overline{\psi}_{\Lambda}]}{
\rbk{
\fun{\alpha_{\Delta,t}}{A}
-
\fun{\widetilde{\alpha}_{\Delta,t}^{\Delta_0,\lambda}}{
P_{\Delta_0,\lambda}AP_{\Delta_0,\lambda}}}B}}
& \leq
C_{\Lambda_0,T,a}\abscard{\Delta_0}\lambda
\napiernum^{-\frac{a_{\Lambda_0,T}\lambda}{2}}
\norm{A}\norm{B}.
\end{aligned}$$
Applying \eqref{eq:fock-left-cutoff-integrand-bound} to each commutator term in
\eqref{eq:fock-cutoff-duhamel-identity}, summing over the ordered pairs, integrating in $s$ from $0$ to $t$,
and adding the second initial cutoff error in \eqref{eq:fock-initial-cutoff-error-bound} gives
$$\begin{aligned}
\sup_{\abs{t}\leq T}
\abs{\fun{\oastate[\overline{\psi}_{\Lambda}]}{
B\rbk{
\fun{\alpha_{\Delta,t}}{A}
-
\fun{\widetilde{\alpha}_{\Delta,t}^{\Delta_0,\lambda}}{
P_{\Delta_0,\lambda}AP_{\Delta_0,\lambda}}}}}
& \leq
C_{\Lambda_0,T,a}\abscard{\Delta_0}\lambda
\napiernum^{-\frac{a_{\Lambda_0,T}\lambda}{2}}
\norm{A}\norm{B}.
\end{aligned}$$
Replacing $a_{\Lambda_0,T}/2$ by $a_{\Lambda_0,T}$ in the notation proves
\eqref{eq:fock-exponential-cutoff-removal}.
\end{proof}

\subsection{Dynamics and equilibrium without a local-number hypothesis}\label{dynamics-and-equilibrium-without-a-local-number-hypothesis}

The model assumptions alone now imply both the finite-volume approximation and the existence of an equilibrium state satisfying the KMS boundary relation for \(\alpha_{\Gamma}\) on the concrete Buchholz algebra.

\begin{thm}[All-temperature finite-volume approximation]\label{thm:fock-all-temperature-local-approximation}
Assume that $v$ is pointwise nonnegative and has finite range, and assume
\eqref{eq:fock-path-chemical-potential-gap}.
For every $\beta > 0$, finite $\Lambda_0 \Subset \Gamma$, and $T > 0$, the finite-volume Gibbs family defined by
\eqref{eq:fock-finite-volume-gibbs-state} satisfies
\begin{equation}
\begin{aligned}
\lim_{\Lambda \nearrow \Gamma}
\sup_{\substack{A \in \fun{\oaresolventalgebra}{\Lambda_0}^{\gamma},\ \norm{A} \leq 1\\
B \in \fun{\oaresolventalgebra}{\Lambda}^{\gamma},\ \norm{B} \leq 1}}
\sup_{\abs{t} \leq T}
\abs{
\fun{\oastate[\overline{\psi}_{\Lambda}]}{
\rbk{\fun{\alpha_{\Lambda,t}}{A} - \fun{\alpha_{\Gamma,t}}{A}} B}}
& =
0.
\end{aligned}
\label{eq:fock-all-temperature-thermodynamic-limit}
\end{equation}
Here $\alpha_{\Gamma,t}$ is the pre-existing automorphism of
Fact \ref{fact:dll-buchholz-automorphism}, and the convergence is uniform on compact time intervals.
No local particle-number estimate is assumed in this statement.
\end{thm}

\begin{proof}
Fix $a \in \openinterval{0}{\beta\varepsilon}$.
For $m\in\semigrposint$, put
$$\begin{aligned}
\Delta_m
& =
\Lambda_0[(2m + 1)r],
&
P_{m,\lambda}
& =
P_{\Delta_m,\lambda},
&
\widetilde{\alpha}_{\Delta,t}^{m,\lambda}
& =
\widetilde{\alpha}_{\Delta,t}^{\Delta_m,\lambda}.
\end{aligned}$$
Equation \eqref{eq:fock-polynomial-ball-volume-growth} gives
$$\begin{aligned}
\abscard{\Delta_m}
& \leq
\sum_{x\in \Lambda_0}\abscard{x[(2m+1)r]}
\leq
\sigma\abscard{\Lambda_0}\rbk{(2m+1)r}^d
\leq
\sigma\abscard{\Lambda_0}\rbk{3r}^dm^d.
\end{aligned}$$
Set
$$\begin{aligned}
C_{\Lambda_0}
& =
\max\setone{1,2C_h}\sigma\abscard{\Lambda_0}\rbk{3r}^d.
\end{aligned}$$
Then $C_{\Lambda_0}m^d$ bounds both $\abscard{\Delta_m}$ and the number of ordered pairs $x\sim y$ having
at least one endpoint in $\Delta_m$, by \eqref{eq:fock-uniform-degree-bound}.
Lemma \ref{lem:fock-exponential-cutoff-removal} is the direct exponential substitute for polynomial-moment
propagation.
Apply its cutoff estimate \eqref{eq:fock-exponential-cutoff-removal} with $\Delta_0=\Delta_m$ on both
$\Delta=\Lambda$ and $\Delta=\Gamma$.
Since $\abscard{\Delta_m}\leq C_{\Lambda_0}m^d$, each of the two cutoff-removal errors is bounded by
\begin{equation}
\begin{aligned}
C_{\Lambda_0,T,a}m^d\lambda\napiernum^{-a_{\Lambda_0,T}\lambda}\norm{A}\norm{B}.
\end{aligned}
\label{eq:fock-direct-exponential-cutoff-error}
\end{equation}

For $\Delta\in\setone{\Lambda,\Gamma}$, equations (75)--(76) in the proof of
\cite[Lemma 4.10]{DeuchertLampartLemm001}, with $X=\Lambda_0$ and $P=P_{m,\lambda}$, become
\begin{equation}
\begin{aligned}
&
D_{\Delta,m,\lambda}
=
P_{m,\lambda}H_{\Delta}P_{m,\lambda}
-
P_{m,\lambda}H_{\Delta\setminus\Lambda_0[2 m r]}P_{m,\lambda}
-
P_{m,\lambda}H_{\Lambda_0[2 m r]}P_{m,\lambda},
\\ 
&
\fun{\widetilde{\alpha}_{\Delta,t}^{m,\lambda}}{P_{m,\lambda} A P_{m,\lambda}}
-
\fun{\widetilde{\alpha}_{\Lambda_0[2 m r],t}^{m,\lambda}}{P_{m,\lambda} A P_{m,\lambda}}
\\
& =
\imunit
\int_0^t
\fun{\widetilde{\alpha}_{\Delta,t-s}^{m,\lambda}}{
\commutator{D_{\Delta,m,\lambda}}{
\fun{\widetilde{\alpha}_{\Lambda_0[2 m r],s}^{m,\lambda}}{
P_{m,\lambda} A P_{m,\lambda}}}}
\opdmsr{s}.
\end{aligned}
\label{eq:fock-shell-duhamel-expansion}
\end{equation}
Finite interaction range makes the terms of orders at most $m$ in the iterated form of
\eqref{eq:fock-shell-duhamel-expansion} vanish.
The estimates and the ordered-time integral used for its order-$m+1$ term are
$$\begin{aligned}
\norm{P_{m,\lambda}\faadj{a_x}a_yP_{m,\lambda}}
& \leq
\sqrt{2}\lambda,
\\ 
\abscard{\Delta_m}
& \leq
C_{\Lambda_0}m^d,
\\ 
\int_{\set{(s_1,\ldots,s_{m+1})}{
0\leq s_{m+1}\leq\cdots\leq s_1\leq\abs{t}}}
\opdmsr{s_1}\cdots\opdmsr{s_{m+1}}
& =
\frac{\abs{t}^{m+1}}{\rbk{m+1}!}.
\end{aligned}$$
Equations (84)--(85) of \cite{DeuchertLampartLemm001} supply $\kappa>0$ such that the product of the $m+1$ transported
commutator maps in \eqref{eq:fock-shell-duhamel-expansion} is bounded by
$\rbk{\kappa\lambda}^{m+1}$.
Substitution of these three estimates into its order-$m+1$ Duhamel integral gives, for $\abs{t}\leq T$,
\begin{equation}
\begin{aligned}
&
\norm{
\fun{\widetilde{\alpha}_{\Delta,t}^{m,\lambda}}{P_{m,\lambda} A P_{m,\lambda}}
-
\fun{\widetilde{\alpha}_{\Lambda_0[2 m r],t}^{m,\lambda}}{P_{m,\lambda} A P_{m,\lambda}}}
\\
& \leq
\sqrt{2}C_{\Lambda_0}\lambda m^d\norm{A}
\int_{\set{(s_1,\ldots,s_{m+1})}{
0\leq s_{m+1}\leq\cdots\leq s_1\leq\abs{t}}}
\rbk{\kappa\lambda}^{m+1}
\opdmsr{s_1}\cdots\opdmsr{s_{m+1}}
\\ 
& =
\sqrt{2}C_{\Lambda_0}\lambda m^d\norm{A}
\rbk{\kappa\lambda}^{m+1}
\frac{\abs{t}^{m+1}}{\rbk{m+1}!}
\\
& \leq
\sqrt{2}C_{\Lambda_0}\lambda m^d\norm{A}
\frac{\rbk{\kappa\lambda T}^{m+1}}{\rbk{m+1}!}.
\end{aligned}
\label{eq:fock-cutoff-shell-duhamel-remainder}
\end{equation}
Taking the supremum over $\abs{t}\leq T$ in \eqref{eq:fock-cutoff-shell-duhamel-remainder} gives
\begin{equation}
\begin{aligned}
\sup_{\abs{t} \leq T}
\norm{
\fun{\widetilde{\alpha}_{\Delta,t}^{m,\lambda}}{P_{m,\lambda} A P_{m,\lambda}}
-
\fun{\widetilde{\alpha}_{\Lambda_0[2 m r],t}^{m,\lambda}}{P_{m,\lambda} A P_{m,\lambda}}}
& \leq
\sqrt{2}C_{\Lambda_0} \lambda m^d \norm{A}
\frac{\rbk{\kappa \lambda T}^{m + 1}}{\rbk{m + 1}!}.
\end{aligned}
\label{eq:fock-cutoff-shell-error}
\end{equation}
The estimate \eqref{eq:fock-cutoff-shell-error} is valid with the outer volume equal to either $\Lambda$ or
$\Gamma$.
It contains no state-dependent term.
For all sufficiently large $m$, choose
$$\begin{aligned}
\lambda
& =
\frac{m}{\napiernum^2\kappa T}
\geq
1.
\end{aligned}$$
Substitution into \eqref{eq:fock-direct-exponential-cutoff-error} gives
$$\begin{aligned}
C_{\Lambda_0,T,a}m^d\lambda\napiernum^{-a_{\Lambda_0,T}\lambda}
& =
\frac{C_{\Lambda_0,T,a}}{\napiernum^2\kappa T}
m^{d+1}
\fnexp{-\frac{a_{\Lambda_0,T}m}{\napiernum^2\kappa T}}.
\end{aligned}$$
Use the lower Stirling bound in the explicit form
$$\begin{aligned}
n!
& \geq
\rbk{\frac{n}{\napiernum}}^n,
\quad
n\in\semigrposint.
\end{aligned}$$
Apply this inequality with $n=m+1$ and substitute the selected value of $\lambda$:
$$\begin{aligned}
\rbk{m+1}!
& \geq
\rbk{\frac{m+1}{\napiernum}}^{m+1},
\\
\frac{\rbk{\kappa\lambda T}^{m+1}}{\rbk{m+1}!}
& =
\frac{\rbk{m/\napiernum^2}^{m+1}}{\rbk{m+1}!}
\leq
\rbk{\frac{m}{\napiernum\rbk{m+1}}}^{m+1}
\leq
\napiernum^{-(m+1)}.
\end{aligned}$$
Multiplication by the prefactor in \eqref{eq:fock-cutoff-shell-error} yields
$$\begin{aligned}
\sqrt{2}C_{\Lambda_0}\lambda m^d
\frac{\rbk{\kappa\lambda T}^{m+1}}{\rbk{m+1}!}
& \leq
\frac{\sqrt{2}C_{\Lambda_0}}{\napiernum^3\kappa T}
m^{d+1}\napiernum^{-m}.
\end{aligned}$$
Combining the two computations defines the joint majorant
\begin{equation}
\begin{aligned}
\fun{\delta_m}{T}
& =
C_{\Lambda_0,T,a,\beta}m^{d+1}
\rbk{
\napiernum^{-m}
+
\fnexp{-\frac{a_{\Lambda_0,T}m}{\napiernum^2\kappa T}}}.
\end{aligned}
\label{eq:fock-direct-exponential-shell-majorant}
\end{equation}
Equation \eqref{eq:fock-direct-exponential-shell-majorant} implies
$\fun{\delta_m}{T}\to0$ as $m\to\infty$.
The two cutoff-removal estimates compare the uncut $\Lambda$ and $\Gamma$ dynamics with
their cutoff counterparts, while the shell estimate compares those counterparts through $\Lambda_0[2 m r]$.
Writing $\fun{A_m}{t}=\fun{\alpha_{\Lambda_0[2mr],t}}{A}$, the resulting two-sided estimate is,
for $\Delta \in \setone{\Lambda,\Gamma}$,
\begin{equation}
\begin{aligned}
\sup_{\abs{t}\leq T}
\max\setone{
\abs{\fun{\oastate[\overline{\psi}_{\Lambda}]}{
\rbk{\fun{\alpha_{\Delta,t}}{A}-\fun{A_m}{t}}B}},
\abs{\fun{\oastate[\overline{\psi}_{\Lambda}]}{
B\rbk{\fun{\alpha_{\Delta,t}}{A}-\fun{A_m}{t}}}}}
\leq
\fun{\delta_m}{T}\norm{A}\norm{B}.
\end{aligned}
\label{eq:fock-two-sided-exponential-shell-approximation}
\end{equation}
For $\Delta=\Lambda$ and $\Delta=\Gamma$,
if $\abs{t} \leq T$,
the two estimates in
\eqref{eq:fock-two-sided-exponential-shell-approximation} give
$$\begin{aligned}
&
\abs{
\fun{\oastate[\overline{\psi}_{\Lambda}]}{
\rbk{\fun{\alpha_{\Lambda,t}}{A}-\fun{\alpha_{\Gamma,t}}{A}}B}}
\\
& \leq
\abs{
\fun{\oastate[\overline{\psi}_{\Lambda}]}{
\rbk{\fun{\alpha_{\Lambda,t}}{A}-\fun{A_m}{t}}B}}
+
\abs{
\fun{\oastate[\overline{\psi}_{\Lambda}]}{
\rbk{\fun{A_m}{t}-\fun{\alpha_{\Gamma,t}}{A}}B}}
\\
& \leq
2\fun{\delta_m}{T}\norm{A}\norm{B}.
\end{aligned}$$
The right-hand side tends to zero uniformly for $\norm{A},\norm{B}\leq1$.
Taking the indicated suprema and then $m\to\infty$ proves
\eqref{eq:fock-all-temperature-thermodynamic-limit}.
\end{proof}

\begin{cor}[All-temperature equilibrium states for $\alpha_{\Gamma}$]\label{cor:fock-all-temperature-kms-states}
Assume that $v$ is pointwise nonnegative and has finite range, and assume
\eqref{eq:fock-path-chemical-potential-gap}.
For every $\beta > 0$ and every weak-$\ast$ accumulation point $\oastate[\psi]$ of
$\seq{\oastate[\overline{\psi}_{\Lambda}]}{\Lambda \Subset \Gamma}$ defined in
\eqref{eq:fock-vacuum-extended-gibbs-state}, the state $\oastate[\psi]$ is invariant under
$\alpha_{\Gamma}$ on $\fun{\oaresolventalgebra}{\Gamma}^{\gamma}$ and satisfies the $\beta$-KMS boundary
condition \eqref{eq:fock-infinite-volume-kms-boundary}.
No local particle-number estimate is assumed, and no continuity of the automorphism group on $\oa{B}$ is asserted.
\end{cor}

\begin{proof}
Choose a subnet $\Lambda_j$ along which
$\oastate[\overline{\psi}_{\Lambda_j}] \to \oastate[\psi]$ weak-$\ast$.
Fix a finite $\Lambda_0\Subset\Gamma$ and
$A,B\in\fun{\oaresolventalgebra}{\Lambda_0}^{\gamma}$, and write
$\fun{A_m}{t} = \fun{\alpha_{\Lambda_0[2 m r],t}}{A}$.
The estimate \eqref{eq:fock-two-sided-exponential-shell-approximation} holds for
$\Delta\in\setone{\Lambda_j,\Gamma}$, with $\fun{\delta_m}{T}\to0$ independently of $j$.
For fixed $m$ and $t$, the operators
$\rbk{\fun{\alpha_{\Gamma,t}}{A} - \fun{A_m}{t}}B$ and
$B\rbk{\fun{\alpha_{\Gamma,t}}{A} - \fun{A_m}{t}}$ belong to $\oa{B}$.
Weak-$\ast$ convergence passes the estimates with $\Delta = \Gamma$ to $\oastate[\psi]$.
For every $j$ satisfying $\Lambda_0\subset\Lambda_j$, define
\begin{equation}
\begin{aligned}
\fun{F_j}{z}
& =
\frac{1}{Z_{\Lambda_j,\beta}}
\sqfun{\trace_{\sphilb{F}_{\Lambda_j}}}{
\napiernum^{-\beta K_{\Lambda_j,\smchemicalpotential_{\Lambda_j}}}
\napiernum^{\imunit zH_{\Lambda_j}}
A
\napiernum^{-\imunit zH_{\Lambda_j}}
B},
\quad
-\beta\leq\opimag z\leq0,
\\
\fun{F_j}{t}
& =
\fun{\oastate[\overline{\psi}_{\Lambda_j}]}{
\fun{\alpha_{\Lambda_j,t}}{A}B},
\\
\fun{F_j}{t-\imunit\beta}
& =
\fun{\oastate[\overline{\psi}_{\Lambda_j}]}{
B\fun{\alpha_{\Lambda_j,t}}{A}}.
\end{aligned}
\label{eq:fock-finite-volume-kms-strip-function}
\end{equation}
The trace argument in \cite[Lemma 5.4]{DeuchertLampartLemm001} proves that $F_j$ is holomorphic for
$-\beta<\opimag z<0$, continuous on the closed strip, and bounded by
$\abs{\fun{F_j}{z}}\leq\norm{A}\norm{B}$.
In the standalone derivative calculation of \cite[Lemma 5.5]{DeuchertLampartLemm001}, every number factor on the
fixed shell is bounded directly by the scalar spectral inequality
$\rbk{1+n}^{4}\leq C_c\fnexp{cn}$ and the exponential propagation
\eqref{eq:fock-exponential-moment-propagation}.
More explicitly, differentiation on the fixed finite shell gives
$$\begin{aligned}
\od{\fun{\oastate[\overline{\psi}_{\Lambda_j}]}{
\fun{A_m}{t}B}}{t}
& =
\fun{\oastate[\overline{\psi}_{\Lambda_j}]}{
\imunit\commutator{H_{\Lambda_0[2mr]}}{\fun{A_m}{t}}B},
\\
\od{\fun{\oastate[\overline{\psi}_{\Lambda_j}]}{
B\fun{A_m}{t}}}{t}
& =
\fun{\oastate[\overline{\psi}_{\Lambda_j}]}{
B\imunit\commutator{H_{\Lambda_0[2mr]}}{\fun{A_m}{t}}}.
\end{aligned}$$
The commutator expansion on the fixed shell and Cauchy--Schwarz have the form
$$\begin{aligned}
\abs{
\fun{\oastate[\overline{\psi}_{\Lambda_j}]}{
\imunit\commutator{H_{\Lambda_0[2mr]}}{\fun{A_m}{t}}B}}
\leq
C_{A,B,m}
\rbk{\fun{\oastate[\overline{\psi}_{\Lambda_j}] \circ \alpha_{\Lambda_0[2mr],t}}
{\prod_{y\in \Lambda_0[2mr]}\rbk{1+N_y}^{4}}}^{1/2}.
\end{aligned}$$
For any $c>0$, the scalar bound is obtained from
$$\begin{aligned}
C_c
& =
\sup_{n\in\monnat}\rbk{1+n}^{4}\napiernum^{-cn}
<
\infty,
\quad
\prod_{y\in \Lambda_0[2mr]}\rbk{1+N_y}^{4}
\leq
C_c^{\abscard{\Lambda_0[2mr]}}
\fnexp{c\sum_{y\in \Lambda_0[2mr]}N_y}.
\end{aligned}$$
Choose $c>0$ small enough that generalized H\"older reduces the last exponential to the one-site propagated
moments in \eqref{eq:fock-exponential-moment-propagation}.
The resulting bound is independent of $j$.
The calculation with $B$ on the left is identical.
For fixed $m$ and $T$, we obtain
$$\begin{aligned}
\sup_{\Lambda_j \supset \Lambda_0[2 m r]}\sup_{\abs{t} \leq T}
\rbk{
\abs{\od{\fun{\oastate[\overline{\psi}_{\Lambda_j}]}{\fun{A_m}{t}B}}{t}}
+
\abs{\od{\fun{\oastate[\overline{\psi}_{\Lambda_j}]}{B\fun{A_m}{t}}}{t}}}
& \leq
C_{A,B,m,T}.
\end{aligned}$$
The derivative bound makes the fixed-shell boundary functions equicontinuous.
For fixed $m$, weak-$\ast$ convergence gives pointwise convergence at every $t$.
Given $\eta>0$, choose a finite $\eta/(3C_{A,B,m,T})$-net of $\closedinterval{-T}{T}$.
Pointwise convergence at the finitely many net points and the derivative bound give
$$\begin{aligned}
\sup_{\abs{t}\leq T}
\abs{
\fun{\oastate[\overline{\psi}_{\Lambda_j}]}{\fun{A_m}{t}B}
-
\fun{\oastate[\psi]}{\fun{A_m}{t}B}}
\leq
\eta
\quad
\rbk{j\text{ sufficiently large}}.
\end{aligned}$$
The same argument applies with $B$ on the left.

The upper boundary in \eqref{eq:fock-finite-volume-kms-strip-function} has the exact decomposition
$$\begin{aligned}
\fun{F_j}{t}
-
\fun{\oastate[\psi]}{\fun{\alpha_{\Gamma,t}}{A}B}
& =
\rbk{
\fun{F_j}{t}
-
\fun{\oastate[\overline{\psi}_{\Lambda_j}]}{\fun{A_m}{t}B}}
\\
& \mathrel{\phantom{=}}
+
\rbk{
\fun{\oastate[\overline{\psi}_{\Lambda_j}]}{\fun{A_m}{t}B}
-
\fun{\oastate[\psi]}{\fun{A_m}{t}B}}
+
\rbk{
\fun{\oastate[\psi]}{\fun{A_m}{t}B}
-
\fun{\oastate[\psi]}{\fun{\alpha_{\Gamma,t}}{A}B}}.
\end{aligned}$$
The first and third terms are bounded by
$\fun{\delta_m}{T}\norm{A}\norm{B}$.
The upper boundary satisfies
\begin{equation}
\begin{aligned}
\sup_{\abs{t} \leq T}
\abs{
\fun{F_j}{t}
-
\fun{\oastate[\psi]}{\fun{\alpha_{\Gamma,t}}{A}B}}
& \leq
\sup_{\abs{t} \leq T}
\abs{
\fun{\oastate[\overline{\psi}_{\Lambda_j}]}{\fun{A_m}{t}B}
-
\fun{\oastate[\psi]}{\fun{A_m}{t}B}}
+
2 \fun{\delta_m}{T} \norm{A} \norm{B}.
\end{aligned}
\label{eq:fock-kms-upper-boundary-comparison}
\end{equation}
The lower boundary in \eqref{eq:fock-finite-volume-kms-strip-function} and the analogous decomposition with
$B$ on the left give the lower-boundary estimate.
The lower boundary satisfies
\begin{equation}
\begin{aligned}
\sup_{\abs{t} \leq T}
\abs{
\fun{F_j}{t - \imunit\beta}
-
\fun{\oastate[\psi]}{B\fun{\alpha_{\Gamma,t}}{A}}}
& \leq
\sup_{\abs{t} \leq T}
\abs{
\fun{\oastate[\overline{\psi}_{\Lambda_j}]}{B\fun{A_m}{t}}
-
\fun{\oastate[\psi]}{B\fun{A_m}{t}}}
+
2 \fun{\delta_m}{T} \norm{A} \norm{B}.
\end{aligned}
\label{eq:fock-kms-lower-boundary-comparison}
\end{equation}
For fixed $m$, the first terms on the right-hand sides of
\eqref{eq:fock-kms-upper-boundary-comparison} and
\eqref{eq:fock-kms-lower-boundary-comparison} tend to zero along the subnet.
Letting first $j$ tend along the subnet and then $m \to \infty$ proves compact-uniform convergence of both
boundaries to the two expressions containing $\alpha_{\Gamma,t}$.
The strip functions satisfy
$$\begin{aligned}
\sup_{0\leq s\leq\beta}
\sup_{t\in\fldreal}
\abs{\fun{F_j}{t-\imunit s}}
& \leq
\norm{A}\norm{B}.
\end{aligned}$$
Montel compactness gives a subsequence converging uniformly on compact subsets of the open strip.
The compact-uniform boundary limits identify the two boundary values of every such subsequential limit.
If two bounded holomorphic limits had these boundary values, their difference would have zero values on both
boundaries.
The three-lines theorem gives
$$\begin{aligned}
\abs{\fun{F_{A,B}^{(1)}-F_{A,B}^{(2)}}{t-\imunit s}}
\leq
\rbk{
\sup_{r\in\fldreal}
\abs{\fun{F_{A,B}^{(1)}-F_{A,B}^{(2)}}{r}}}^{1-\frac{s}{\beta}}
\rbk{
\sup_{r\in\fldreal}
\abs{\fun{F_{A,B}^{(1)}-F_{A,B}^{(2)}}{r-\imunit\beta}}}^{\frac{s}{\beta}}
=
0.
\end{aligned}$$
The limit is unique and defines the bounded strip function $F_{A,B}$ with the boundary values in
\eqref{eq:fock-infinite-volume-kms-boundary}.

For invariance, set $B=1$.
Finite-volume invariance and the two shell comparisons give
$$\begin{aligned}
&
\abs{
\fun{\oastate[\psi]}{\fun{\alpha_{\Gamma,t}}{A}}
-
\fun{\oastate[\psi]}{A}}
\leq
\abs{
\fun{\oastate[\psi]}{\fun{\alpha_{\Gamma,t}}{A}}
-
\fun{\oastate[\psi]}{\fun{A_m}{t}}}
+
\abs{
\fun{\oastate[\psi]}{\fun{A_m}{t}}
-
\fun{\oastate[\overline{\psi}_{\Lambda_j}]}{\fun{A_m}{t}}}
\\
&
+
\abs{
\fun{\oastate[\overline{\psi}_{\Lambda_j}]}{
\fun{A_m}{t}-\fun{\alpha_{\Lambda_j,t}}{A}}}
+
\abs{
\fun{\oastate[\overline{\psi}_{\Lambda_j}]}{
\fun{\alpha_{\Lambda_j,t}}{A}-A}}
+
\abs{
\fun{\oastate[\overline{\psi}_{\Lambda_j}]}{A}
-
\fun{\oastate[\psi]}{A}}.
\end{aligned}$$
The fourth term is zero.
For fixed $m$, the second and fifth terms tend to zero along the subnet.
The first and third terms are bounded by $\fun{\delta_m}{T}\norm{A}$.
Taking the subnet limit and then $m\to\infty$ proves invariance on local observables.
The local gauge-invariant resolvent algebras are norm dense in
$\fun{\oaresolventalgebra}{\Gamma}^{\gamma}$.
For local approximating sequences
$\seq{A_n}{n\in\semigrposint}$ and $\seq{B_n}{n\in\semigrposint}$ satisfying
$A_n\to A$ and $B_n\to B$ as $n\to\infty$, the automorphisms are isometric and the state has norm one.
The norm difference is bounded by
$$\begin{aligned}
\abs{
\fun{\oastate[\psi]}{\fun{\alpha_{\Gamma,t}}{A}B}
-
\fun{\oastate[\psi]}{\fun{\alpha_{\Gamma,t}}{A_n}B_n}}
& \leq
\norm{A-A_n}\norm{B}
+
\norm{A_n}\norm{B-B_n}.
\end{aligned}$$
The same bound holds on the lower boundary.
The uniform strip bound extends invariance and
\eqref{eq:fock-infinite-volume-kms-boundary} to the full gauge-invariant resolvent algebra.
This argument does not invoke the conditional equilibrium theorem of \cite{DeuchertLampartLemm001} and does not
transfer a particle-number functional to $\oastate[\psi]$.
\end{proof}

\section{Conclusion and Outlook}\label{conclusion-and-outlook}

Proposition 3.1 of \cite{DeuchertLampartLemm001} already supplies the generally discontinuous automorphism group on the full concrete Buchholz algebra. The obstruction in Theorem 4.1 and Theorem 5.2 of that paper is the uniform local moment condition, not existence of the algebraic automorphism. Theorem \ref{thm:fock-all-temperature-local-number} replaces that condition by a proved exponential estimate under repulsiveness and the low-activity gap. Theorem \ref{thm:fock-all-temperature-local-approximation} gives the finite-volume dynamical approximation, and Corollary \ref{cor:fock-all-temperature-kms-states} gives invariant states satisfying the KMS boundary relation at every temperature.

The discrete functional integral completes the low-activity comparison proposed in \cite[Remark 5.3]{DeuchertLampartLemm001}. The decisive formula is the interacting cancellation \eqref{eq:fock-interacting-loop-cancellation}. It implies reduced-density domination, the factorial estimate \eqref{eq:fock-factorial-local-number-bound}, and the exponential estimate \eqref{eq:fock-exponential-local-number-bound} for every \(\beta > 0\). These bounds remove the local particle-number hypothesis from both conclusions whenever the chemical potential satisfies the uniform gap \eqref{eq:fock-path-chemical-potential-gap}.

The chemical-potential gap keeps the system on the vacuum side of the one-particle threshold. The present comparison does not cover higher-density regimes in which condensation or superfluid behavior may occur. For general superstable interactions outside this gap, an all-temperature local insertion estimate remains open. At high temperature, the volume-uniform estimates of \cite[Theorem 1]{GongKuwaharaTong001} provide a complementary route to uniqueness and boundary-independent convergence of the entire Gibbs family. The equilibrium theorem proved here is an existence result and does not assert that all volume exhaustions have the same limit.

A structural obstruction precedes the equilibrium problem. The free Bose dynamics acts by automorphisms of the resolvent algebra, whereas an interacting dynamics generally sends elements of the original gauge-invariant resolvent algebra into a larger algebra. The observable-algebra extension constructed in \cite{DetlevBuchholz002,DetlevBuchholz003} contains these evolved observables for a large class of pair interactions. In the lattice formulation of \cite[Definition 2.2 and Proposition 3.1]{DeuchertLampartLemm001}, this role is played by the concrete Buchholz algebra on bosonic Fock space.

The resolvent algebra itself has an abstract universal \(\oacstar\)-algebraic presentation and admits faithful regular representations \cite{BuchholzGrundling002}. The present construction of the Buchholz algebra reintroduces a distinguished Fock representation at the level of the enlarged algebra. This dependence is restrictive because a von Neumann algebra obtained as the weak closure of a GNS representation retains the folium normal to that representation. Representations associated with different temperatures in an infinite system are frequently disjoint, and their normal states cannot then be represented simultaneously as normal states of one of the corresponding factors. The purpose of the present program is to free the Buchholz algebra from the Fock representation as a reference representation and to retain, together with the resolvent algebra, as much abstract and universal \(\oacstar\)-algebraic information as possible.

The concrete argument in this paper separates the inputs needed for such an abstraction. The dynamics uses the Buchholz sector consistency, the direct cutoff estimate, and finite-range shell propagation. The present equilibrium proof realizes its local equilibrium states by the Fock trace and derives the required local-number estimates from creation-annihilation insertions and a positive functional integral. The Fock trace itself need not be retained as an axiom of an abstract formulation. Proposition 3.1 of \cite{DeuchertLampartLemm001}, recorded in Fact \ref{fact:dll-buchholz-automorphism}, verifies in the concrete model that every finite \(\oa{B}_X\) is invariant under its local dynamics \(\alpha_X\). An abstract formulation may require this invariance together with the existence of a \(\beta\)-KMS state for \(\alpha_X\) on every finite local algebra. Impose in addition the standard local-normality requirement of algebraic quantum statistical mechanics. Here local normality means that the restriction of a state to every finite local algebra extends to a normal state in a regular representation of that finite subsystem. The local KMS assumption and local normality should make the choice of the Fock representation inessential on each finite local algebra, without building one global Fock representation into the abstract algebra. The remaining task is to formulate the positive functional integral and the resulting local-number estimates in a form compatible with these locally normal KMS states.

Continuous-space bosons form a separate extension problem. Buchholz established stability of enlarged gauge-invariant resolvent algebras under continuous pair-potential dynamics and developed the corresponding sector structures \cite{DetlevBuchholz002,DetlevBuchholz003}. Adapting the present strategy requires spatial localization estimates that control both particle transport and ultraviolet behavior. The canonical-ensemble analysis in \cite{DetlevBuchholz012} provides another test of whether the particle-sector mechanism can be separated from one global Fock folium.

Perturbation theory is a further direction. Standard-form methods control Liouvilleans and equilibrium vectors after a representation has been selected \cite{DerezinskiJaksicPillet001}. The positive-temperature Euclidean construction in \cite[Sections 17.1.5 and 21.4--21.5]{DerezinskiGerard001} identifies the Gaussian periodic path-space representation with an Araki--Woods von Neumann algebra and uses this identification to construct perturbed dynamics and KMS states by functional integration. For bosonic Gaussian path spaces, it gives a particularly sharp von Neumann-algebraic realization of the abstract equivalence between stochastically positive KMS systems, periodic Osterwalder--Schrader positive processes, and positive semigroup structures established in \cite[Sections 1 and 6--8]{KleinLandau001}. The Klein--Landau construction is formally based on a \(\oacstar\)-dynamical system, but its stochastic reconstruction passes to the GNS Hilbert space and the associated von Neumann algebra.

An intrinsic theory should determine which perturbations preserve the abstract resolvent or Buchholz algebra and compare the resulting KMS states with standard-form and Araki--Woods perturbations only after a representation has been selected. This requires a deeper \(\oacstar\)-algebraic theory of probability and stochastic processes in which Euclidean reconstruction and functional-integral perturbation do not begin with a fixed Fock or Araki--Woods reference representation.

\bibliography{myref.bib}

\end{document}